\documentclass[aps,twocolumn,floatfix,superscriptaddress,nofootinbib]{revtex4-2}
\usepackage{booktabs}
\usepackage{graphicx}
\usepackage{braket}
\usepackage{stix}
\usepackage{xcolor}
\usepackage{appendix}
\usepackage{amsmath}
\usepackage{dcolumn}
\usepackage{soul}
\usepackage{bm}
\usepackage{amsthm}
\usepackage{enumitem}
\usepackage{amsmath}
\newtheorem{theorem}{Theorem}[section]
\newtheorem{lemma}{Lemma}
\newtheorem{corollary}{Corollary}[lemma]
\usepackage{xr-hyper}
\usepackage{hyperref}
\hypersetup{colorlinks,
	citecolor=blue
}
\usepackage{cleveref}
\usepackage{apptools}
\AtAppendix{\counterwithin{lemma}{section}}

\makeatletter
\newcommand*{\addFileDependency}[1]{
  \typeout{(#1)}
  \@addtofilelist{#1}
  \IfFileExists{#1}{}{\typeout{No file #1.}}
}
\makeatother

\begin{document}

\preprint{APS/123-QED}

\title{Imaginary components of out-of-time correlators and information scrambling for navigating the learning landscape of a quantum machine learning model}

\author{Manas Sajjan$^\dagger$}
\affiliation{Department of Chemistry, Purdue University, West Lafayette, IN 47907}

\author{Vinit Singh$^\dagger$}
\affiliation{Department of Chemistry, Purdue University, West Lafayette, IN 47907}

\author{Raja Selvarajan}
\affiliation{Department of Physics and Astronomy, Purdue University, West Lafayette, IN 47907}

\author{Sabre Kais}
\email{kais@purdue.edu}
\affiliation{Department of Chemistry, Purdue University, West Lafayette, IN 47907}

\affiliation{Department of Physics and Astronomy, Purdue University, West Lafayette, IN 47907}

\affiliation{Department of Electrical and Computer Engineering, Purdue University, West Lafayette, IN 47907}


\begin{abstract}
We introduce and analytically illustrate that hitherto unexplored imaginary components of out-of-time correlators can provide unprecedented insight into the information scrambling capacity of a graph neural network. Furthermore, we demonstrate that it can be related to conventional measures of correlation like quantum mutual information and rigorously establish the inherent mathematical bounds (both upper and lower bound) jointly shared by such seemingly disparate quantities. To consolidate the  geometrical ramifications of such bounds during the dynamical evolution of training we thereafter construct an emergent convex space. This newly designed space offers much surprising information including the saturation of lower bound by the trained network even for physical systems of large sizes, transference, and quantitative mirroring  of spin correlation from the simulated physical system across phase boundaries as desirable features within the latent sub-units of the network (even though the latent units are directly oblivious to the simulated physical system) and the ability of the network to distinguish exotic spin connectivity(volume-law vs area law). Such an analysis demystifies the training of quantum machine learning models by unraveling how quantum information is scrambled through such a network introducing correlation surreptitiously among its constituent sub-systems and open a window into the underlying physical mechanism behind the emulative ability of the model. 
\end{abstract}

\maketitle

\def\thefootnote{$\dagger$}\footnotetext{\text{These authors contributed equally to this work}}\def\thefootnote{\arabic{footnote}}

\section{Introduction}

Heralding machine learning algorithms to be the most disruptive technological advancement of the present era would not be an overstatement \cite{burkov2019hundred,sarker2021machine, 10.1007/978-3-030-29407-6_5,tsihrintzis2020machine}. Despite successful inroads of the former to enable scientific applications on both classical and quantum hardware \cite{carleo2019machine,mehta2019high,li2022machine,wu2021application,havlivcek2019supervised,PhysRevLett.122.040504,liu2021rigorous}, a pervasive reluctance prevails in making such algorithms mainstream as indicated by a recent survey\cite{keith2021combining}. A part of the culpability is in the very nature of training of the associated paradigmatic models which often seems agnostic to physical principles or human-acquired domain intuition. Attempting to address this lacuna, the primary objective of our thesis is to gain physical insight into the learning mechanism of a machine learning model (to be called the learner) assigned to simulate the eigenstates of any user-defined system (to be called the driver), a task central to the core of many physico-chemical applications \cite{sajjan2022quantum}.

The major contributions of this work are many-fold. Following a description of the learner, we explicate the role of the information transport and scrambling between the internal sub-units of the learner during the course of its training. To this end, the hitherto unexplored imaginary component of out-of-time correlators (OTOCs)\cite{swingle2018unscrambling} of the learner is defined and analytically characterized using invariants of motion generated from the underlying Lie Algebra \cite{Hall2015}. It is then subsequently employed to act as a compass in navigating the parameter landscape during learning. In recent years OTOCs has been used as a quintessential measure of how fast information propagation away from the source of initiation happens in the real-time post any local excitation in atomic systems \cite{PhysRevA.94.040302, PhysRevLett.126.200603,PhysRevB.97.144304,PhysRevLett.126.070601}, in statistical physics to probe thermalization \cite{PhysRevE.104.034120,FAN2017707,PhysRevA.104.022405}, in quantum-information theory\cite{sharma2021quantum, touil2020quantum}, as a diagnostic tool for quantum chaos \cite{PhysRevLett.124.160603} and even in models mimicking aspects of quantum gravity \cite{PhysRevB.94.035135,PhysRevA.102.022402,PhysRevA.97.042330,PhysRevX.5.041025}. Such correlators have also been measured using quantum circuits\cite{mi2021information, PhysRevLett.128.160502,PhysRevX.11.021010,landsman2019verified,PhysRevLett.129.050602}. We thereafter connect such a quantifier with known measures of quantum correlation and illustrate analytically the relative bounds shared by the two quantities, which are stricter than conventionally known bounds. Equipped with these aforesaid probes, we provide a map of navigating the parameter landscape during training of the network in an emergent space and demonstrate with
with appropriate case studies features of the trained learner like saturation of lower bound in the above inter-relationship and how the footprints of correlation in the driver get imprinted onto a trained learner, thereby empowering the latter to be used as a concrete diagnostic tool in investigating physical phenomena like phase transitions as well as in differentiating between drivers with exotic connectivity/interactions, etc by simply accessing properties of the learner alone.

In the following section (Section \ref{Theory}) we shall describe the generative neural network we use for this work. In Section \ref{OTOC_defn_section} we define OTOCs in a general setting. In Section  \ref{imag_OTOC_defn_section} we introduce and prove that several invariants of motion associated with the phase-space description of OTOC for the aforesaid neural network exists with particular emphasis on the hitherto unexplored imaginary part which as we shall see which form a key player in our analysis. In Section \ref{Lie-algebra_OTOC_section} we describe the generators associated with the said invariants which highlights an underlying Lie Algebra. In Section \ref{I_vs_eta_section} we prove how the imaginary part of OTOC for our network is related to conventional measures of correlation as described previously including the relative bounds which they share and construct a new emergent convex space to understand the training mechanism of the network and the role of its latent sub-units. In Section \ref{training_algo_section} we describe the polynomially scaling algorithm for training the network and subsequent construction of the new space. In Section \ref{Result_Disc} we describe our primary inferences from numerical studies in the said space and conclude in Section \ref{concluding_section}.

\begin{figure}[!htb]
    \centering   \includegraphics[width=0.45\textwidth]{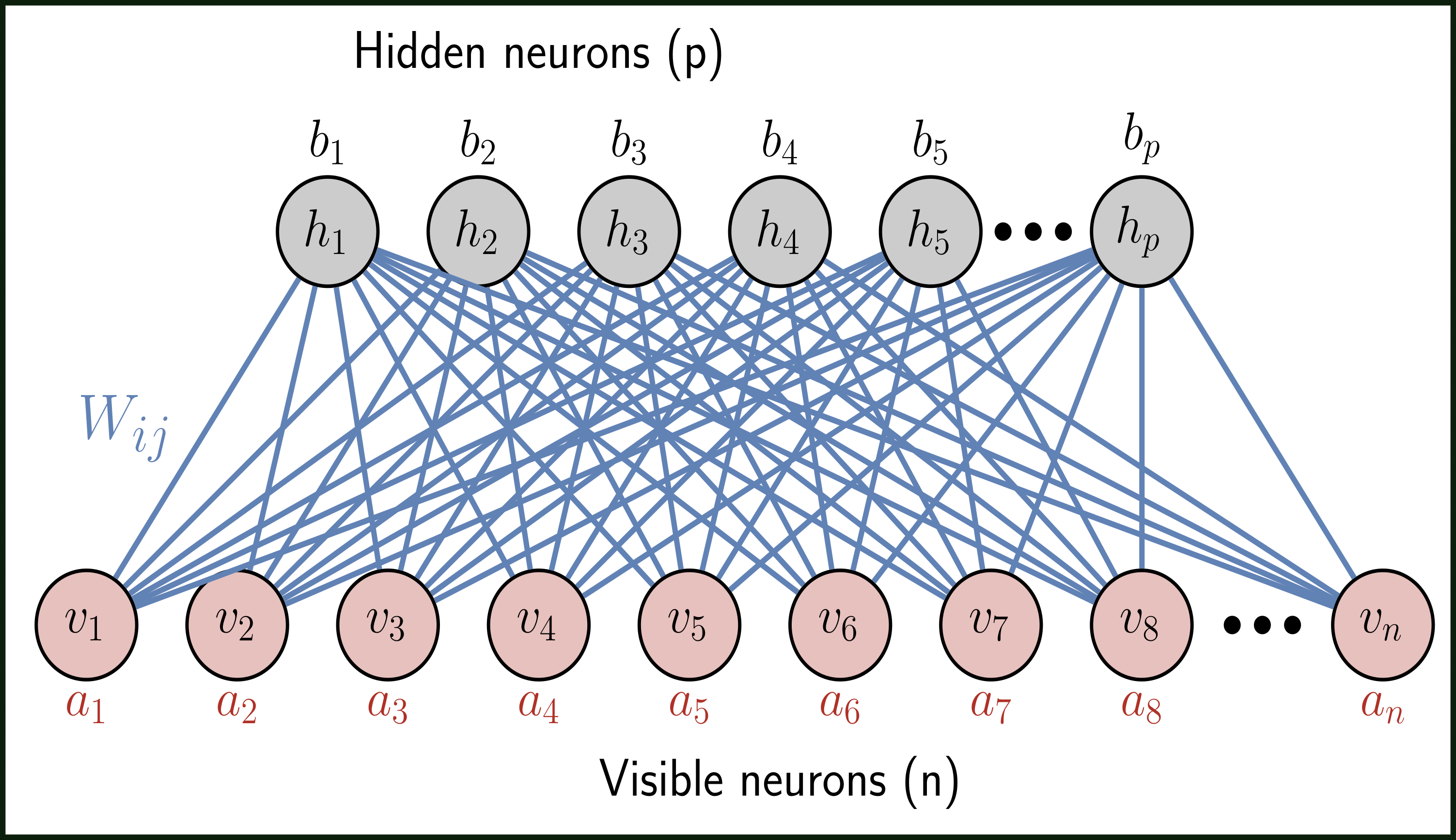}
    \caption{{\color{black}The description of the learner network $G=(V,E)$ as defined in text (also known as Restricted Boltzmann Machine (RBM) in literature\cite{Melko2019a,Hinton504}.) The set of hidden neurons $\{h_j\}_{j=1}^{p}$ are shown in grey with the corresponding bias vector $(\vec{b})$. Similarly, set of visible neurons $\{v_i\}_{i=1}^{n}$ are shown in red with the corresponding bias vector $(\vec{a})$. The interconnecting weight matrix with elements $W_{ij}$ is shown in blue. The parameter set $\vec{X}=(\vec{a}, \vec{b}, \vec{W})$ are tuned during the training of the network}}
    \label{Fig:fig_scheme_QRBM}
\end{figure}

\section{THEORETICAL BACKGROUND}\label{Theory}

\subsection{Description of the Graph Neural Network $G=(V,E)$} \label{Graph_description}

The specific description of the generative network used as the learner in this work is illustrated in Fig.\ref{Fig:fig_scheme_QRBM}. Formally the learner is a connected bipartite graph $G=(V,E)$ (also known as Restricted Boltzmann Machine(RBM)\cite{PhysRevB.94.165134,Melko2019a,Hinton504}). The set $V$ consists of $(p+n)$ neurons with $(p,n) \in \mathcal{Z}_{+}$ and is further classified into two subsets as $V = \{v_i\}_{i=1}^{n} \bigcup  \{h_j\}_{j=1}^p$ (see Fig.\ref{Fig:fig_scheme_QRBM}). Both the sub-sets are endowed with a locally accessible $\sigma^z$ (Pauli-z) operator and their corresponding $bias$- vectors are $\vec{a} \in \mathcal{R}^n$ and $\vec{b} \in \mathcal{R}^p$. The edge set $E$ can be characterized with the adjacency matrix $Adj(G) \in \{0,1\}^{(n+p) \times (n+p)}$ of the graph $G$ defined as follows:
\begin{align}
    Adj(G)_{ij} =
    \begin{cases}
     1, & \text{if}\:\: \nu_i \in \{v_i\}_{i=1}^{n}, \:\:\nu_j \in \{h_i\}_{i=1}^{p}\\ &\forall \:\: (\nu_i, \nu_j) \in V \\
     0 & \text{otherwise}
    \end{cases}
\end{align}
$\forall (i,j) \in \mathcal{Z}^{p+n}$. Corresponding to each non-zero entry in $Adj(G)$ we define an edge $e_{ij} \in E$. This would mean that $|E| = p*n$. Associated with $e_{ij} \in E$ we define a $weight$ matrix $\vec{W} \in \mathcal{R}^{n \times p}$ (shown in blue in Fig.\ref{Fig:fig_scheme_QRBM})
each element of which quantifies the strength of the shared connection between any of the neurons from the subset $\{v_i\}_{i=1}^{n}$ (visible-node register) to every neuron in the subset $\{h_j\}_{j=1}^p$ (hidden-node register). Collectively the tunable parameters $\vec{X}=(\vec{a}, \vec{b}, \vec{W}) \in \mathcal{R}^{n+p+np}$ enables us to define the learner's Hamiltonian $\mathcal{H} : G \mapsto \mathcal{R}^{2^{p+n} \times 2^{p+n}}$ similar to that of a classical Ising model \cite{10.1088/978-0-7503-3843-1ch2, RevModPhys.39.883} as 

\begin{align}
    \mathcal{H}(\vec{X}, \vec{v}, \vec{h}) &= \sum_{i=1}^{n} a_i\sigma^z (v_i) + \hspace{-0.07in}\sum_{j=1}^{p} b_j \sigma^z (h_j) + \hspace{-0.13in}\sum_{i=1,j=1}^{n,p} W^i_j\sigma^z (v_i)\sigma^z (h_j)
    \label{eq: Ising_energy}
\end{align} 
where $\sigma^z(\chi_i)$ represents operator $\sigma^z$ acting at neuron $\chi_i$. The learner is trained to encode a probability distribution that corresponds to the diagonal elements of a thermal state $\rho_{th}$, of the Hamiltonian in Eq.\ref{eq: Ising_energy}, and is defined as follows \cite{Hinton504,Torlai2017,PhysRevB.94.165134}:

\begin{align}
\rho_{th}(\Vec{X}, \vec{v},\vec{h}) &= \frac{e^{-\mathcal{H}(\Vec{X}, \vec{v}, \vec{h})}}
{Tr_{\{v,h\}}e^{-\mathcal{H}(\vec{X}, \vec{v}, \vec{h})}} \label{eq:rbm_dist}
\end{align} 

Any instance of spin configuration $(\vec{v}, \vec{h})$ of the combined registers of $(p + n)$ spins are samples drawn from the said thermal distribution in Eq.\ref{eq:rbm_dist}. Equipped with this, the primary objective of the learner network $G$ is to mimic the amplitude field of the target state $\psi(\vec{X})$ of the driver hamiltonian $H \in \mathcal{C}^{d \times d}$ following the prescription 
\begin{align}
\psi(\vec{v},\vec{X}) &= \sum_{\vec{h}} \text{diag}(\rho_{th}(\vec{X}, \vec{v}, \vec{h})) \nonumber \\
& = \frac{e^{-\sum_{i=1}^n a_i\sigma^z (v_i)}\times \Pi_{j=1}^{p} 2 cosh(b_j + \sum_{i}^n W^i_j\sigma^z (v_i))}{Tr_{\{v,h\}}e^{-\mathcal{H}(\vec{X}, \vec{v}, \vec{h})}}
\label{psi_X}
\end{align}
Whether the training happens on a classical processor or a NISQ device, the flow of the algorithm, in either case, involves randomly initializing $\vec{X}$ to construct $\psi(\vec{v},\vec{X})$ as in Eq.\ref{psi_X} and then evaluating an appropriate merit-function $J(\vec{X}) = \langle \hat{J} \rangle$ (usually $\hat{J}=H$ but other properties can be optimized too). The variational parameters $X=(\vec{a}, \vec{b}, \vec{W})$ are thereafter subsequently tuned using gradient-based updates of the merit function $\partial_{\vec{X}} J (\vec{X})$ until a desired convergence threshold is reached. The algorithmic details of such a training process can be found elsewhere \cite{ChNg2017,PhysRevLett.127.158303,Carrasquilla_2017,sajjan2021quantum}. Formally such a training exploits the isomorphism of the vector space $\mathcal{C}^d$ of the driver and the space of spin configurations of the visible node $\vec{v}$ of the learner of dimension $S=dim(2^n)$ with $n=\lceil log_2 d \rceil$. As $\vec{X} \in \mathcal{R}^{n+p+np}$, we specifically focus on drivers with non-negative coefficients for the target state. Extension to account for the phase of the target wavefunction is straight-forward \cite{Xia_2018, Kanno2021}.

The number of neurons $p$ in $\{h_j\}_{j=1}^p$ is chosen arbitrarily by the user (usually $p \sim n$). It is clear from Eq.\ref{psi_X} that the variational form of the ansatz is independent of $\sigma^z(h_j)$ i.e. the spins of the latent neurons and it is the configurations of $\{\sigma^z(v_i)\}_{i=1}^n$ which forms the requisite basis for the eigenstate of the driver. While from the optimization point of view, the role of the hidden set of neurons ($\{h\}_{j=1}^p$) is thus to enhance the expressive capability of the network by increasing the number of tunable parameters $(\vec{b}, \vec{W})$, from a more physical perspective $\{h\}_{j=1}^p$ induces higher-order correlation between the neurons of the visible layer by relaying the information between a given $(v_i, v_k) \in \vec{v}$ as the latter is devoid of any direct interaction. This relay is sensitive to the connections ($e_{km} \in E$) which defines the $weight$ matrix elements $W_{km}$ of the $k$-th neuron in $\vec{v}$ and $m$-th neuron in $\vec{h}$. Central goal of our manuscript is thus probing how $\{h_j\}_{j=1}^p$ physically fosters correlation between configurations of $\{v_i\}_{i=1}^n$ by analyzing how an initial excitation on a given visible neuron is shared with a given hidden neuron in real-time. This will be a direct neuron-resolved picture of the dynamical evolution of the network during training and provide valuable insight into its learning mechanism.

\subsection{Out of Time Correlators (OTOC)}
\label{OTOC_defn_section}
To attain the aforesaid objective of probing information exchange between the active (visible) and latent (hidden) units of the learner, we shall employ an OTOC which we shall define in detail in this section. An OTOC is primarily composed of two unitary operators 
$U_1(0)$ and $U_2(0)$ wherein $U_i(0)$ is a local operator for a specific site $i$ evolving in time $t$ through Heisenberg prescription \cite{Heisen_qm}. The quantity is sensitive to the extent of information scrambling between sites $\{1,2\}$ \cite{Hashimoto_2017,PhysRevB.96.054503,Braumuller2021ProbingQI,Sundar_2022} and is defined as follows
\begin{align}
C_{U_1,U_2}(t) &= \langle U^\dagger_1 (0) U^\dagger_2(t) U_1(0) U_2(t) \rangle \label{eq:C_t_gen}
\end{align}

Eq.\ref{eq:C_t_gen} is an estimator of growth of the operator $U_2(t) = e^{i H_{\textit{otoc}} t} U_2(0) e^{-i H_{\textit{otoc}} t}$ under the effect of the generator $H_{\textit{otoc}}$ assuming the latter possesses interaction within the different sites of the system.
If the sites $(1,2)$ are far apart, the supporting bases of the operators $U_1(0)$ and $U_2(0)$ are sparsely-overlapping and hence $C_{U_1,U_2}(0)$ will be initially 1. With time, the operator $U_2(t) = e^{i H_{\textit{otoc}} t} U_2(0) e^{-i H_{\textit{otoc}} t}$ will start to extend its support thereby culminating in an eventual overlap between the probe operator $U_1(0)$ and $U_2(t)$ which ultimately lead to changing values of $C_{U_1,U_2}(t)$. The quantity $C_{U_1,U_2}(t)$ thus directly hints at how fast the excitation has traveled from the initially localized point at $2$ to site $1$. The reason for the nomenclature of ``out-of-time correlator" is due to the fact that the expression $C_{U_1,U_2}(t)$ has a time-ordering which is non-monotonic as opposed to forward time correlators like $\langle U_2(t) U_1(0) \rangle$ wherein the operators $U_i$ are sequentially arranged in ascending order of time.
Also, unlike two-point correlators which are known to decay in $O(1)$ time irrespective of the length ($L$) of the system employed, OTOCs like Eq.~\ref{eq:C_t_gen} decay in time proportional to the difference in location of the two sites (hence $\sim L$)\cite{d2016quantum}.

\subsection{Imaginary Component of OTOC of $G=(V,E)$ - Geometrical characterization in phase space} \label{imag_OTOC_defn_section}
Unlike in most reports wherein the real part of Eq.\ref{eq:C_t_gen} is used, we introduce the imaginary part of Eq.\ref{eq:C_t_gen} i.e. $Im(C_{U_1,U_2}(t))$ and shall see that it is also  an important player in our analysis. To this end we offer in Section \ref{Gen_OTOC_form} in Appendix a general formulation for obtaining both the $Re(C_{U_1,U_2}(t)$ and $Im(C_{U_1,U_2}(t))$ through positive semi-definite construction of other appropriate operators. To understand how $\{h_j\}_{j=1}^p$ and $\{v_i\}_{i=1}^n$ scrambles information internally 
we now construct a specific OTOC and establish the contents of the following theorem.

\begin{theorem}\label{lem1}
For a given parameter vector $\vec{X}$, one can define $\mathcal{H}(\vec{X}, \vec{v}, \vec{h})$ (see Eq.\ref{eq: Ising_energy}) and a thermal state $\rho_{th}(\vec{X}, \vec{v}, \vec{h})$. Let us thereafter define the following OTOC with $U_1(0) = \tilde{\sigma_\alpha}= \sigma^\alpha(v_k,0)-\kappa_1\mathcal{I}$, and operator $U_2(0) = \tilde{\sigma_\beta}=\sigma^\beta(h_m,0) -\kappa_2\mathcal{I}$ and the generator $H_{\textit{otoc}} = \mathcal{H}(\vec{X}, \vec{v}, \vec{h})$ in Eq.\ref{eq:C_t_gen}
$\:\:\forall \:\: \{\alpha, \beta \} \in \{x,y\}$. 

\begin{align}
C_{\sigma^\alpha, \sigma^\beta}(\kappa_1, \kappa_2,\vec{X}, t) &= \langle \tilde{\sigma}^\alpha(v_k,0) \tilde{\sigma}^\beta(h_m,t)\tilde{\sigma}^\alpha(v_k,0)\tilde{\sigma}^\beta(h_m,t)\rangle\label{eq:C_t_rbm}
\end{align}
Note that $\{\kappa_1, \kappa_2\} \in \mathcal{C}^2$ are arbitrary user-defined mean translations. Also $\langle \cdot \rangle$ indicates averaging over the thermal state $\rho_{th}(\vec{X}, \vec{v}, \vec{h})$ defined in Eq.\ref{eq:rbm_dist} which activates the $\vec{X}$ dependence. Using \ref{eq:C_t_rbm}, one can then make the following statements:
\begin{enumerate}
{\item For $(\kappa_1, \kappa_2)$ $\in$ $\mathcal{C}^2$ \\
    $\begin{aligned}[t]
        C_{\sigma^\alpha, \sigma^\beta}(\kappa_1, \kappa_2, \vec{X}, t) &= C_{\sigma^\alpha, \sigma^\beta}(0, 0, \vec{X}, t) 
        + |\kappa_1|^2|\kappa_2|^2 \nonumber \\ &+ |\kappa_2|^2 +|\kappa_1|^2
    \end{aligned}$\\
 }  
{\item The following invariants of motion exists for $C_{\sigma^\alpha, \sigma^\beta}(0,0,\vec{X}, t)$ :
\begin{enumerate}
{\item
     $\begin{aligned}[t]  I_1=&-2\dot{\xi}_{\sigma^\alpha,\sigma^\beta}(\vec{X},\tau)Cos(\tau) - 2\xi_{\sigma^\alpha,\sigma^\beta}(\vec{X},\tau)Sin(\tau)
     \end{aligned}$
     }
     {\item 
     $\begin{aligned}[t]
     I_2= & -2\dot{\xi}_{\sigma^\alpha,\sigma^\beta}(\vec{X},\tau)Sin(\tau) + 2\xi_{\sigma^\alpha,\sigma^\beta}(\vec{X},\tau)Cos(\tau)
     \end{aligned}$
     }
\end{enumerate}}
where $\xi_{\sigma^\alpha,\sigma^\beta}(\vec{X},\tau)$ can either be the real or the imaginary part of $(C_{\sigma^\alpha,\sigma^\beta}(0,0,\vec{X},\tau))$ and  $\dot{\dottedsquare}$ is $\frac{\partial \dottedsquare}{\partial \tau} $ with $\tau=4W^k_m t$

\end{enumerate}
\end{theorem}
\begin{proof}
See Section (\ref{thm_1_proof}) in Appendix
\end{proof}

\begin{figure}[!htb]
    \centering   \includegraphics[width=0.45\textwidth]{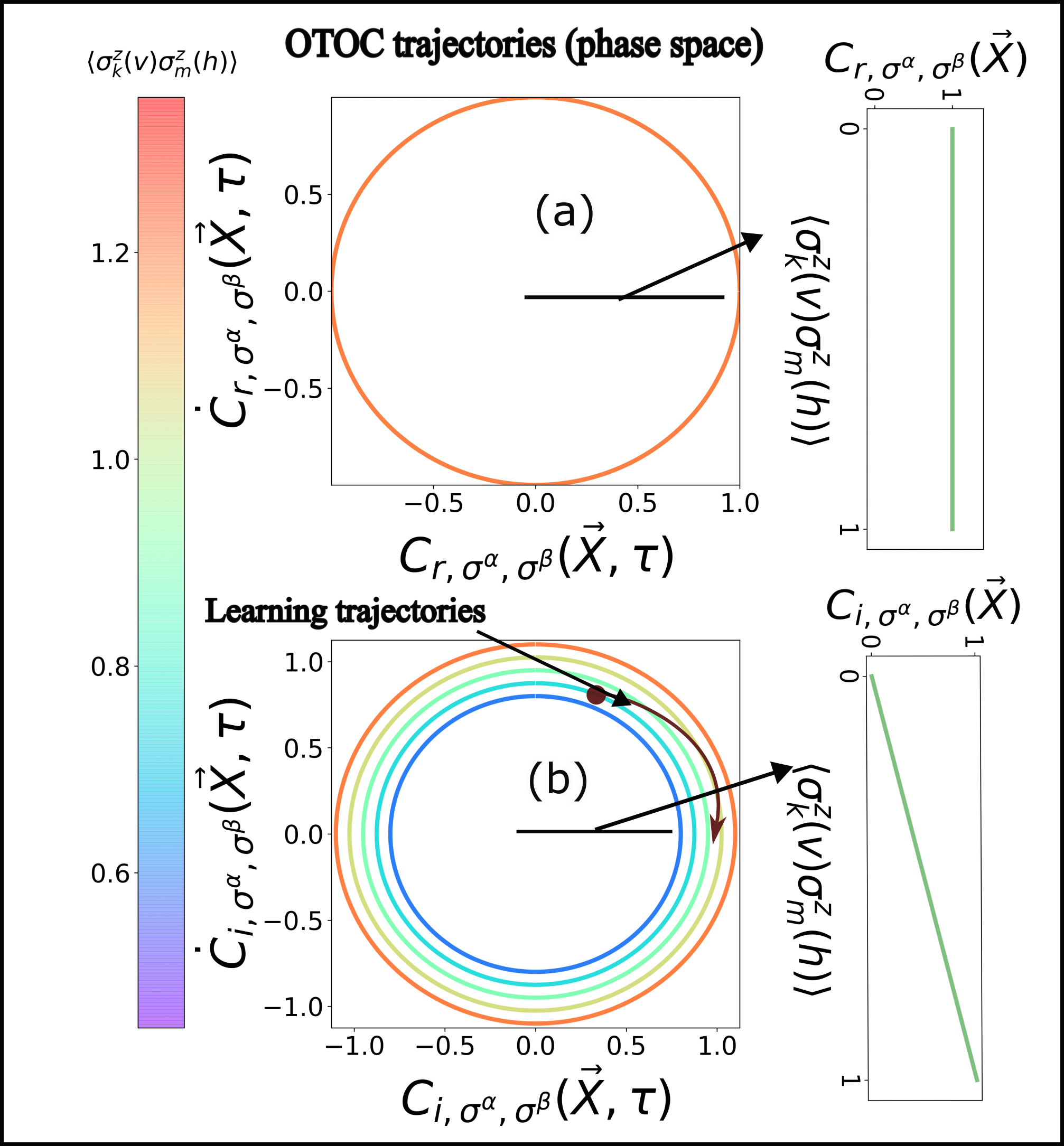}
    \caption{{\color{black} (a) The profile for the 
    real part of the compound invariant $\frac{I_1^2+I_2^2}{4}$ (see Section \ref{thm_1_proof} in Appendix, Corollary \ref{Corr_2}) obtained from the simple invariants of Theorem \ref{lem1} where $\alpha = \beta = x$ and $C_{r, \sigma^\alpha, \sigma^\beta}(\vec{X}, \tau) = Re(C_{\sigma^\alpha, \sigma^\beta}(0,0,\vec{X}, \tau))$ has been substituted for $\xi_{\sigma^\alpha, \sigma^\beta}(\vec{X},\tau)$  in Theorem \ref{lem1}(2) $\:$ (b) Similar to (a) but for  $C_{i, \sigma^\alpha, \sigma^\beta}(\vec{X}, \tau) = Im(C_{\sigma^\alpha, \sigma^\beta}(0,0,\vec{X}, \tau))$. Note the subscript `r'('i') denotes the real(imaginary) part of Eq.\ref{eq:C_t_rbm} in the plots.
    Both quantities are evaluated at $(\kappa_1=0,\kappa_2=0)$, hence the  explicit dependence on $(\kappa_1,\kappa_2)$ has been dropped for notational brevity. 
    The trajectory in the phase space for both cases is a circle (Corollary \ref{Corr_2}) whose radius remains conserved during training in (a) but changes in (b). This is further corroborated by the two green subplots alongside which for (b) shows that the radius is exactly equal to $\langle\sigma^x(v_k,0),\sigma^x(h_m,0)\rangle_{\rho_{th}(\vec{X},\vec{v},\vec{h})}$ and is sensitive to $\Vec{X}$ unlike in (a). This indicates why the $Im(C_{\sigma^\alpha, \sigma^\beta}(0,0,\vec{X}, \tau))$ of the OTOC string in Eq.\ref{eq:C_t_rbm} can yield important insight about the training process of $G$.
    }}
    \label{Fig:Inv_OTOC}
\end{figure}

{\color{black} A combination of the invariants from 
Theorem \ref{lem1}(2) is plotted in Fig.\ref{Fig:Inv_OTOC}(a-b) for the $Re(C_{\sigma^\alpha, \sigma^\beta}(0, 0, \vec{X}, t))$ and 
$Im(C_{\sigma^\alpha, \sigma^\beta}(0, 0, \vec{X}, t))$ and shows decisively the importance of the imaginary part. Similar results for other invariants are illustrated in Section \ref{thm_1_proof} in Appendix}. Certain assertions are apparent from Theorem \ref{lem1}. Firstly, as a direct corollary (proven in Section \ref{thm_1_proof} in Appendix as Corollary Eq.\ref{soln_C_t}), one can deduce analytical expressions for 
$C_{\sigma^\alpha,\sigma^\beta}(0,0,\vec{X},\tau)$ as
\begin{align}
C_{\sigma^\alpha,\sigma^\beta}(0,0,\vec{X},\tau) \hspace{-0.03in} &= Cos(\tau) + i \langle\sigma^z(v_k,0)\sigma^z(h_m,0)\rangle_{\rho_{th}(\vec{X})} Sin(\tau)\label{eq:C_t_form}
\end{align}
Eq.\ref{eq:C_t_form} guarantees that temporal behavior of the OTOC in Eq.\ref{eq:C_t_rbm} is oscillatory in nature  due to the $SU(2)$ algebra associated with the unitary rotation of $\sigma^\beta(h_m,t)$ $(\beta \in \{x,y\})$ around the z-axis induced by the generator $H_{otoc}=\mathcal{H}$ in Eq.\ref{eq: Ising_energy}. 
The frequency associated with the rotation for both terms in Eq. \ref{eq:C_t_form} is expectedly dictated by the strength of the interaction ($W^k_m$) shared by the $k$-th visible neuron and $m$-th hidden neuron. 
Note that training of $G$ amounts to hopping between OTOC trajectories (each of which is generated with a frozen incumbent instance of $\vec{X}$) in phase space as we traverse the parameter space by varying $\vec{X}$ (see Fig.\ref{Fig:Inv_OTOC}(b)). We see that the amplitude of the $Im(C_{\sigma^\alpha, \sigma^\beta}(0, 0, \vec{X}, t))$ yields directly a two-body correlation function $\langle\sigma^z(v_k,0)\sigma^z(h_m,0)\rangle_{\rho_{th}(\vec{X}, \vec{v}, \vec{h})}$ of spins in the visible and hidden register of the learner (see Fig.\ref{Fig:Inv_OTOC}(b))
which is sensitive to $\vec{X}$ and thus probes the changing correlation content between $h_m$ and $v_k$ during training. 

\subsection{Lie-Algebraic generators associated with invariants of OTOC} \label{Lie-algebra_OTOC_section}
It is possible to define generators associated with the invariants defined in Theorem \ref{lem1} possessing an underlying Lie Algebraic structure. For example for the invariant type given in Theorem\ref{lem1}(1) with the following expression :
\begin{align}   &I_1=-2\dot{\xi}_{\sigma^\alpha,\sigma^\beta}(\vec{X},\tau)Cos(\tau) - 2\xi_{\sigma^\alpha,\sigma^\beta}(\vec{X},\tau)Sin(\tau) \label{I_v_one}
\end{align}
where $\xi_{\sigma^\alpha, \sigma^\beta}(\vec{X},\tau)$ can be either $C_{r\sigma^\alpha, \sigma^\beta}(0,0,\vec{X},\tau)$ or $C_{i\sigma^\alpha, \sigma^\beta}(0,0,\vec{X},\tau)$ one can deduce the following transformation $\hat{A}(\phi, \tau)$:
\begin{align}
&\hat{A}(\phi, \tau)(\xi_{\sigma^\alpha, \sigma^\beta}(\vec{X},\tau))  \nonumber \\
& = e^{\phi Cos(\tau) \frac{\partial}{\partial \xi_{\sigma^\alpha, \sigma^\beta}(\vec{X},\tau)}} (\xi_{\sigma^\alpha, \sigma^\beta}(\vec{X},\tau)) \nonumber \\ &=\xi^\prime_{\sigma^\alpha, \sigma^\beta}(\vec{X},\tau)
\end{align}
Note that the derivative in the exponent is with respect to $\xi_{\sigma^\alpha, \sigma^\beta}(\vec{X},\tau)$ itself. Such a transformation $\hat{A}(\phi, \tau): \{\xi\}_{I_1} \mapsto \{\xi\}_{I_1}$  where $\{\xi\}_{I_1}$ is a solution space marked by  the given value of the invariant $I_1$. In other words, a given solution $\xi_{\sigma^\alpha, \sigma^\beta}(\vec{X},\tau)$ with a specific value of the invariant $I_1$, the transformation changes it to another solution $\xi^\prime_{\sigma^\alpha, \sigma^\beta}(\vec{X},\tau)$ which possesses the same value for the invariant. This can be verified by explicit computation too. For instance if $\xi_{\sigma^\alpha, \sigma^\beta}(\vec{X},\tau) = C_{r\sigma^\alpha, \sigma^\beta}(0,0,\vec{X},\tau)$, then from the Corollary \ref{soln_C_t} it is clear that $\xi_{\sigma^\alpha, \sigma^\beta}(\vec{X},\tau)=C_{r\sigma^\alpha, \sigma^\beta}(0,0,\vec{X},\tau) = Cos(W^k_m t)=Cos(\tau)$. Substituting this in Eq.\ref{I_v_one}, one gets the value of $I_1$ as 0. If the transformation $A(\phi, \tau)$ is applied on $\xi_{\sigma^\alpha, \sigma^\beta}(\vec{X},\tau)=Cos(\tau)$, the new solution is $(1+\phi)Cos(\tau)$ which has the same value of the invariant as before. Similarly one can also deduce an invariant-preserving transformation for $I_2$ type invariants in Theorem\ref{lem1}(2) as $e^{\phi Sin(\tau) \frac{\partial}{\partial \xi_{\sigma^\alpha, \sigma^\beta}}}$. It can be shown that the generators of the two transformations commutes. For other invariants in in Corollary \ref{Corr_2} one can similarly deduce other transformations like this. The generators of a full-set of such transformations forms a closed single parameter Lie-group (with respect to parameter $\phi$) which can be shown using their commutation algebra. Such a structure is characteristic of systems with harmonic degrees of freedom, but discovering it within the abstract phase space of OTOC string for the learner is interesting and worth further investigation. The ramifications of such generators on the full phase-space of the OTOC strings and its effect on the training of the learner will be explored in the future. We thus see that instead of a direct evaluation of the OTOC, evaluation through the invariants described in Theorem 1 offers a fuller characterization of the phase space with richer insight into the geometry of the manifold.

\subsection{Inter-relationship with Covariance and Quantum Mutual Information ($\mathcal{I}(v_k,h_m)$) - Mapping training trajectories to $\mathcal{I}-\eta(\vec{X})$ space } \label{I_vs_eta_section}

We further consolidate the importance of $Im(C_{\sigma^\alpha, \sigma^\beta}(0, 0, \vec{X}, t))$ concretely in this section by showing that the quantity can be related to quantum mutual information. To do that it is imperative to first establish a direct relationship of the said quantity with $Cov(\sigma^z(v_k,0),\sigma^z(h_m, 0))_{\rho_{th}}$ by using the results of Theorem \ref{lem1}(1-2)  as follows
\begin{align}
\hspace{0.01in}\eta(\vec{X}) 
&= Cov(\sigma^z(v_k,0),\sigma^z(h_m, 0))_{\rho_{th}}\nonumber \\
&=C_{\sigma^\alpha, \sigma^\beta}(\kappa_1, 0,\vec{X}, t_1) 
+C_{\sigma^\alpha, \sigma^\beta}(0, \kappa_2,\vec{X}, t_1)\nonumber \\ &-C_{\sigma^\alpha, \sigma^\beta}(\kappa_1,\kappa_2,\vec{X}, t_1)
\label{eq:Cov_imag_part}
\end{align}
with $\kappa_1 = \sqrt{i\langle \sigma^z(v_k,0)\rangle}$, $\kappa_2 = \sqrt{i\langle \sigma^z(h_m,0)\rangle}$ and  $t_1 = \frac{\pi}{8W^k_m}$.
$\eta(\vec{X})$ in Eq.\ref{eq:Cov_imag_part} can only attain a value of zero if the neurons  $v_k$ and $h_m$ are uncorrelated.
We shall now connect $\eta(\vec{X})$ to other well-known correlation measures like Von-Neumann entropies of and incipient mutual information ($\mathcal{I}(v_k,h_m)$) \cite{nielsen_chuang_2010,Watrous_book, doi:10.1073/pnas.0806782106, PhysRevLett.118.016804} of subsystems of $G$ using the following theorem.

\begin{theorem}\label{lem2}
The two particle-density matrices $^2\rho(v_k, h_m)$ and one particle-density matrices $^1\rho(v_k)$ and $^1\rho(h_m)$ of the learner $G$ can be computed for a given $\vec{X}$ from Eq.\ref{eq:rbm_dist} (see Section \ref{thm_2_proof} in Appendix). Using these one can construct $\mathcal{I}(v_k, h_m) = S(^1\rho(v_k)) + S(^1\rho(h_m)) - S(^2\rho(v_k, h_m))$ where $S(Y)=-Tr(YlnY)$. The following statement is true
\begin{align}
LB \leq \mathcal{I}(v_k,h_m) \leq UB 
\end{align}

wherein the lower bound (LB) is
\begin{align}
LB &= 2 -  2\mathcal{L}\big( \frac{1+\eta(\vec{X})}{4} \big) - 2 \mathcal{L}\big( \frac{1-\eta(\vec{X})}{4} \big) \label{eq:LB}
\end{align}
and the corresponding upper bound (UB) is 
\begin{align}
UB &= \mathcal{L}\bigg(\frac{1+\sqrt{1+(-1)^\gamma\eta(\vec{X})}}{2} \bigg) + \mathcal{L}\bigg(\frac{1-\sqrt{1+(-1)^\gamma\eta(\vec{X})}}{2} \bigg) \label{eq:UB}
\end{align}
with $-1 \le \eta(\vec{X}) \le 1$ as defined in Eq.\ref{eq:Cov_imag_part} and $\mathcal{L}(x)=-xln(x)\:\: \forall\: x \in \mathcal{R}_+$ and $\gamma = 0$ \:if \:$\eta(\vec{X}) < 0$ or $\gamma = 1$ \:if \:$\eta(\vec{X}) \ge 0$

\end{theorem}
\begin{proof}
See Section \ref{thm_2_proof} in Appendix
\end{proof}

The implication of Theorem \ref{lem2} is profound and is illustrated diagrammatically in Fig.\ref{Fig:MI_eta}. It introduces an newly emergent $\mathcal{I}-\eta$ space to probe learning trajectories. In this space, the bounds together define a convex set (see Fig \ref{Fig:MI_eta}) within which resides the acceptable values of $\mathcal{I}(v_k,h_m)$ and $\eta(\vec{X})$, that the learner $G$ can access during the course of its training. One must note that LB in Eq.\ref{eq:LB} is more stringent compared to previously known bound of $\frac{\eta(\vec{X})^2}{2}$ \cite{PhysRevLett.100.070502} as seen in Fig.\ref{Fig:MI_eta}. 

\begin{figure}[!htb]
     \centering   \includegraphics[width=0.46\textwidth]{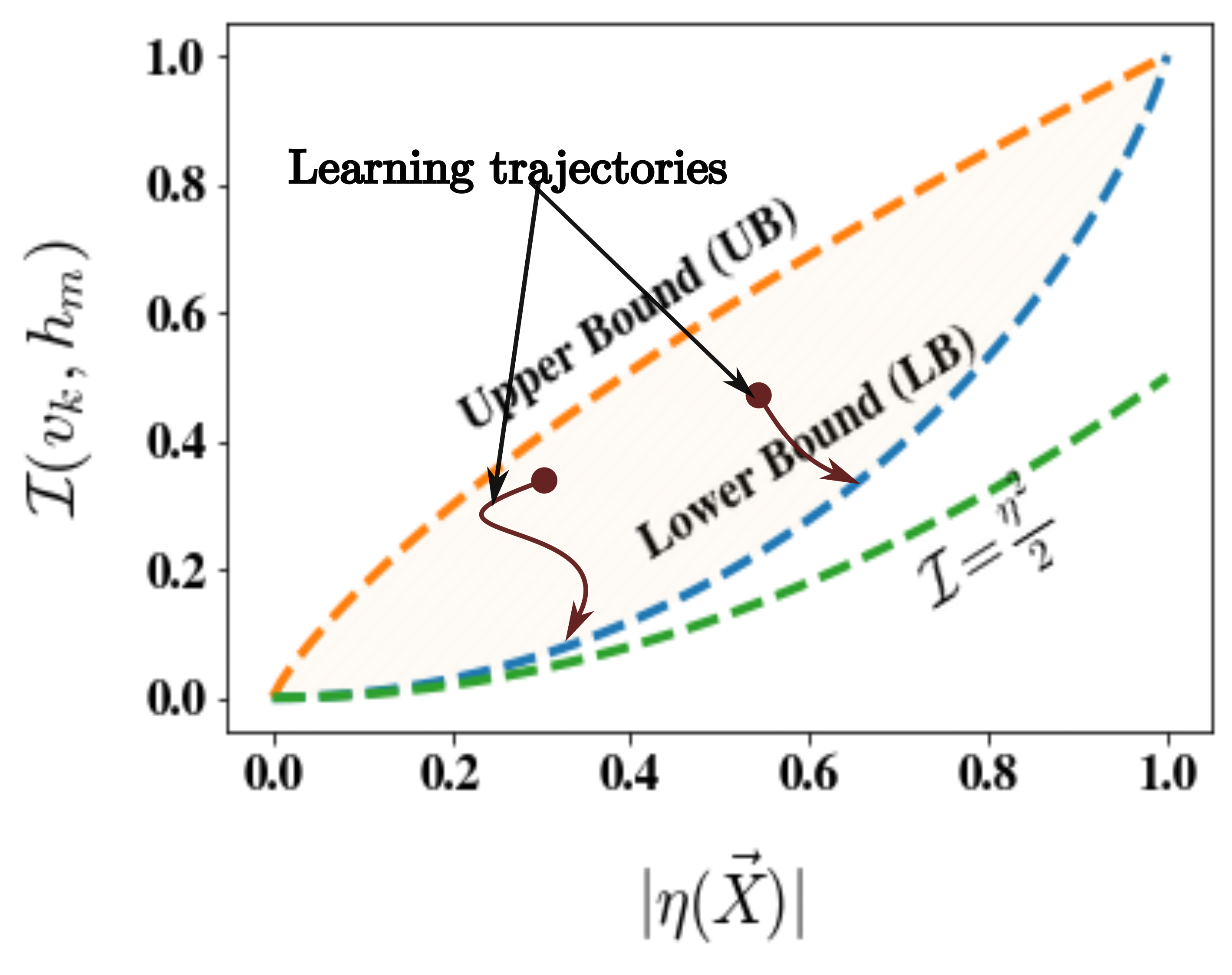}
    \caption{The upper bound(UB) and lower bound (LB) in the convex $\mathcal{I}-\eta$ space for network $G$ as described in Theorem\ref{lem2} in Eq. \ref{eq:LB},\ref{eq:UB} respectively. In all plots in this panel, we shall use $\alpha = \beta = x$ in Eq.\ref{eq:Cov_imag_part}. Provided alongside in green is the conventionally known lower bound for any general bi-partition of an arbitrary system \cite{PhysRevLett.100.070502}. The lower bound (LB) for $G$ is thus stricter than the known general bound. Representative learning trajectories of the network $G$ are shown}
     \label{Fig:MI_eta}
 \end{figure}

\section{Computational details}\label{training_algo_section}

\subsection {Training Algorithm and efficient polynomially scaling construction of $\mathcal{I}- \eta(\vec{X})$ space}

In this section we discuss in detail the computational algorithm used for training the network $G=(V,E)$ and for subsequent estimation of $\mathcal{I}$ and $\eta(\vec{X})$ through sampling to extract the features of the learner in this emergent space.

For a given driver hamiltonian $H \in \mathcal{C}^{d\times d}$, the variational form of the target state which is used to train the learner as ansatz is defined in Eq.\ref{psi_X}. Since we are interested in ground state of the respective drivers, training as explained before is done minimizing the cost function $J(\vec{X})$ (in this case the energy) with respect to the parameters $\vec{X}$ of the trial ansatz $\psi(\vec{v},\vec{X})$ as follows:
\begin{align}
J(\vec{X}) & = \frac{\langle \psi(\vec{v},\vec{X})|H|\psi(\vec{v},\vec{X})\rangle}{\langle \psi(\vec{v},\vec{X})|\psi(\vec{v},\vec{X})\rangle} \label{energy_cost}
\end{align}

For updating the parameters of Eq.\ref{energy_cost} we follow the algorithm in Ref \cite{carleo2017solving} closely which is based on Sorella's Stochastic Reconfiguration technique \cite{doi:10.1063/1.2746035}. In this technique the parameters of the cost function are updated using a pre-conditioner $F \in \mathcal{C}^{p*n\times p*n}$ defined as follows
\begin{align} 
\vec{X} &\rightarrow \vec{X} - l*F^{-1}(\vec{X})S(\vec{X}) \label{train_algo1}
\end{align}
where
\begin{align}
F^{i}_{j} &= \langle D_i^\dagger D_j \rangle - \langle D_i^\dagger \rangle \langle D_j \rangle  \label{train_algo2}\\
S_{i} &= \langle E_{loc} D_i^\dagger \rangle - \langle E_{loc} \rangle \langle D_i^\dagger \rangle \label{train_algo3}\\
D_i|\vec{v}\rangle &= \frac{\partial_{\vec{X_i}} \psi(\vec{v},\vec{X})}{\psi(\vec{v},\vec{X})}|\vec{v}\rangle\:\:\:\:\:\:\: \forall \:\:\vec{X_i} \in (\vec{a}, \vec{b}, \vec{W})  \label{train_algo4}\\
 E_{loc} &= \frac{\langle \vec{v}|H|\psi(\vec{v},\vec{X})\rangle}{\psi(\vec{v},\vec{X})} \label{train_algo5}
\end{align}
The indices $(i,j)$ in Eq.\ref{train_algo1}- \ref{train_algo5} run over all $p*n$ parameters in $\vec{X}$. Also $l \in \mathcal{R}_{+}$ is known as the learning rate and $|\vec{v}\rangle$ is the computational basis of configuration of the visible node register. The above expressions are true for any variational ansatz. For the particular graph neural network encoding $G=(V,E)$, the components of $D_i$ can be analytically expressed as diagonal matrix elements in the computational basis as follows:
\begin{align}
\langle\vec{v}^\prime|D_{a_i}|\vec{v}\rangle &= \sigma^z(v_i)\delta_{\vec{v}\vec{v}^\prime} \\
\langle\vec{v}^\prime|D_{b_j}|\vec{v}\rangle &= tanh(b_j + \sum_{i}W^i_j \sigma^z(v_i))\delta_{\vec{v}\vec{v}^\prime} \\
\langle\vec{v}^\prime|D_{W^i_j}|\vec{v}\rangle &= \sigma^z(v_i)tanh(b_j + \sum_{i}W^i_j \sigma^z(v_i))\delta_{\vec{v}\vec{v}^\prime}
\end{align}
The averaging $\langle ... \rangle$ defined in Eq.\ref{train_algo2}-\ref{train_algo3} is over the distribution $\frac{|\psi(\vec{v},\vec{X})|^2}{\langle \psi(\vec{v},\vec{X})| \psi(\vec{v},\vec{X}) \rangle}$. This distribution in accordance to Eq.\ref{psi_X} is dependant on configurations of $\{\sigma(v_i)\}_{i=1}^n$. Exactly computing all such $2^n$ configurations would yield an exponentially scaling protocol. Hence the algorithm in Ref \cite{carleo2017solving} uses a Monte-Carlo based sampling protocol using Metropolis-Hastings procedure. The primary workhorse of the protocol relies on starting with one of the arbitrary (say $\{\sigma(v_i)\}_{i=1}^n$) $2^n$ configuration and then randomly choosing one of the spins $\sigma(v_i)$ and mutating/flipping it to generate a new configuration (say $\{\sigma(v_i)^{\prime}\}_{i=1}^n$). The new configuration is accepted or rejected using the following rule
\begin{align}
 P(\{\sigma(v_i)\}_{i=1}^n \rightarrow
 \{\sigma(v_i)^{\prime}\}_{i=1}^n) &= min (1, \Bigg|\frac{\psi(\vec{v}^\prime,\vec{X})}{\psi(\vec{v},\vec{X})}\Bigg|^2)
 \end{align}
 
Ref\cite{carleo2017solving} establishes that the computational cost of each such Monte Carlo sweep is $O(n*p)$. If the number of sweeps used is $N_S$ then the total cost of computing the gradient updates by computing $F$ would be $O(n*p*N_S)$ as there are $O(n*p)$ parameters over which the matrix $F$ and vector $S$ needs to be computed. We use Netket library \cite{10.21468/SciPostPhysCodeb.7} to perform computations using this algorithm. If the condition number of $F$ is large, invertibility may be a problem which is obviated using a infinitesimal shift which has been fixed at 0.01 in our calculations. The total number of independent Markov chains is set to 1000 and 60 sweeps are used at each step along a chain. All parameters of $\vec{X}$ is randomly initialized from a normal distribution of zero mean and 0.01 standard deviation. The maximum number of iterations kept were 200 and learning rate $(l)$ for training $G=(V,E)$ is set to 0.1. With these set of parameters we were able to achieve an error threshold of $\le$ 0.1\% for convergence. 

Once the training of $G=(V,E)$ has commenced to the desired accuracy threshold, we have now obtained the trained parameter vector (say $\vec{X}^*=(\vec{a}^*, \vec{b}^*, \vec{W}^*)$ where $*$ is not complex conjugation as all parameters are real but denotes a specific instance of the trained $\vec{X}$ procured after training)  of the learner $G$. Using these one can easily construct the eigenvalues of the two $^2\rho(v_k,h_m)$ and one-particle reduced density matrices ($^1\rho(v_k)$ or $^1\rho(h_m)$) as obtained from the contraction of Eq.\ref{eq:rbm_dist}. 
The four eigenvalues $\{\lambda_i(^2\rho(v_k,h_m))\}_{i=1}^4$ of the two-particle density matrix $^2\rho(v_k,h_m)$ for the learner $G$ between a specific pair of visible and hidden spins (say $(k,m)$) as deduced in Section \ref{eig_2rdm_deduction} are expressed as follows:

\begin{widetext}
\begin{align}
\lambda_1(^2\rho(v_k,h_m)) &= \lambda(^2\rho(v_k=1,h_m=1)) = \frac{N_{1}}{Z} e^{-a_k^* - b_m^* - W^{k*}_m}\times \langle \Pi_{j\ne m}^p 2cosh (b_j^* + \sum_{i\ne k}^n W^{i*}_j v_i + W^{k*}_j)\rangle_{P(\{v_i\}_{i\ne k}^n, h_m =1)} \label{lambda_2RDM_1}\\
\lambda_2(^2\rho(v_k,h_m)) &= \lambda(^2\rho(v_k=1,h_m=-1)) = \frac{N_{-1}}{Z} e^{-a_k^* + b_m^* + W^{k*}_m} \times \langle \Pi_{j\ne m}^p 2cosh (b_j^* + \sum_{i\ne k}^n W^{i*}_j v_i + W^{k*}_j)\rangle_{P(\{v_i\}_{i\ne k}^n, h_m =-1)} \label{lambda_2RDM_2}\\
\lambda_3(^2\rho(v_k,h_m)) &= \lambda(^2\rho(v_k=-1,h_m=1)) = \frac{N_{1}}{Z} e^{a_k^* - b_m^* + W^{k*}_m}\times \langle \Pi_{j\ne m}^p 2cosh (b_j^* + \sum_{i\ne k}^n W^{i*}_j v_i - W^{k*}_j)\rangle_{P(\{v_i\}_{i\ne k}^n, h_m =1)} \label{lambda_2RDM_3}\\
\lambda_4(^2\rho(v_k,h_m)) &= \lambda(^2\rho(v_k=-1,h_m=-1)) = \frac{N_{-1}}{Z} e^{a_k^* + b_m^* - W^{k*}_m}\times \langle \Pi_{j\ne m}^p 2cosh (b_j^* + \sum_{i\ne k}^n W^{i*}_j v_i - W^{k*}_j)\rangle_{P(\{v_i\}_{i\ne k}^n, h_m =-1)} \label{lambda_2RDM_4}
\end{align}
\end{widetext}

where each of the averages in Eq.\ref{lambda_2RDM_1}- \ref{lambda_2RDM_4} are computed over the distribution $P(\{v_i\}_{i\ne k}^n, h_m)$ and $N_{h_m}$ is the associated normalization constant. These are defined as 
\begin{align}
P(\{v_i\}_{i\ne k}^n, h_m) &= \frac{e^{-(\sum_{i \ne k}^n a_i^* v_i + W^{i*}_m v_ih_m))}}{N_{h_m}} \\
N_{h_m} & = \Pi_{i\ne k}^n 2cosh(a_i^* + \sum_{i\ne k}^n W^{i*}_m h_m) \label{P_2RDM}
\end{align}
The corresponding eigenvectors of the two particle density matrix for the eigenvalues in Eq.\ref{lambda_2RDM_1}-\ref{lambda_2RDM_4} are $|0(v_k)0(h_m)\rangle$, $|0(v_k)1(h_m)\rangle$ and $|1(v_k)0(h_m)\rangle$, $|1(v_k)1(h_m)\rangle$ respectively for the four eigenvalues Eq.\ref{lambda_2RDM_1}-Eq.\ref{lambda_2RDM_4} where $(0,1)$ is notationally equivalent to $(1,-1)$ for each spins $(v_k, h_m)$

The quantity $Z$ is the partition function defined as 
\begin{align}
Z&=\sum_{(\vec{v}, \vec{h})}\hspace*{-0.01in} e^{(-\sum_{i}^n -a_i^* v_i - \sum_{i}^m b_j^* h_j -\sum_{i,j}^{n,m} W^{i*}_j v_i h_j)}
\end{align}
However $Z$ need not be explicitly computed as it can be eliminated using the unit normalization condition of the eigenvalues. 
The corresponding eigenvalues for one-particle density matrix $^1\rho(\xi_i,0)$ for a neuron $\xi_i$ in the learner $G$, can thereafter be constructed by contraction from 
Eq.\ref{lambda_2RDM_1}, \ref{lambda_2RDM_2}, \ref{lambda_2RDM_3}, \ref{lambda_2RDM_4} are
\begin{align}
\lambda_1(^1\rho(\xi_i)) &= \lambda_i(^2\rho(v_k,h_m)) + \lambda_j(^2\rho(v_k,h_m)) \label{1RDM_eig_1}\\&\:\:\: (\text{if} \:\: \xi_i=v_k, \:\: (i,j)=(1,3)) \nonumber \\&\:\:\:
(\text{if} \:\: \xi_i=h_m, \:\: (i,j)=(1,4)) \nonumber \\
\lambda_2(^1\rho(\xi_i)) &= \lambda_i(^2\rho(v_k,h_m)) + \lambda_j(^2\rho(v_k,h_m)) \label{1RDM_eig_2}\\&\:\:\: (\text{if} \:\: \xi_i=v_k, \:\: (i,j)=(2,4)) \nonumber \\&\:\:\:
\text{if} \:\: (\xi_i=h_m, \:\: (i,j)=(2,3)) \nonumber 
\end{align}
with respective eigenvectors are $|0(\xi_i)\rangle$ and $|1(\xi_i)\rangle$
where $\xi_i \in (v_k, h_m)$.

Using the eigenvalues of $^2\rho(v_k,h_m)$ and one-particle reduced density matrices ($^1\rho(v_k)$ or $^1\rho(h_m)$) one can compute $I(v_k, h_m)$ vs $|\eta(\vec{X})|$ as illustrated in Section \ref{thm_2_proof} in Appendix and construct the entire space. From Eq.\ref{lambda_2RDM_1}-\ref{lambda_2RDM_4} and Eq.\ref{P_2RDM} it is clear that the underlying probability distribution $P(\{v_i\}_{i=1, i\ne k}^n, h_m)$ from which the eigenvalues are computed is defined over the $2^{n-1}$ configurations $\{\sigma(v_i)\}_{i\ne k}^n$ of the visible node register (also independent of $\{h_i\}_{i\ne m}$ but is dependant on $h_m=\pm 1$ which is kept fixed for a given eigenvalue)
. Accurate estimation of the entire distribution would require exponentially scaling resources. However we shall now demonstrate a polynomially scaling algorithm based on Gibbs sampling technique. The distribution can be marginalized easily as it is completely factorizable over individual $\{v_i\}_{i\ne k}$ spins due to lack of connectivity among the visible spins in the network $G$ which is the key feature required in our sampling. We estimate the eigenvalues using Gibbs sampling from this distribution \ref{P_2RDM}. For each pair $(k,m)$ of visible and hidden neurons, the sampling technique for a given drawn/sampled configuration of visible neurons (say $\vec{v_1}$) performs a sum over all visible neurons to construct each of the primitive operators $2cosh (b_j^* + \sum_{i\ne k}^n W^{i*}_j v_i + W^{k*}_j)$. Each such primitive operator is indexed by $j$. The primitive operators are then constructed for every $j \in \{1,2...p\}$ and then multiplied together to yield the compound operator $\Pi_{j\ne m} 2cosh (b_j^* + \sum_{i\ne k}^n W^{i*}_j v_i + W^{k*}_j)$. This step thus incurs a cost of $O(np)$ alone. This step is repeated for many drawn samples of visible node configurations (say $\vec{v_2}$ now) with the total number of samples drawn be $N_E$ thereby introducing a total cost of $O(np*N_E)$ for this step. This yields the four eigenvalues as given in Eq.\ref{lambda_2RDM_1}-\ref{lambda_2RDM_4} and completes the computation of $I(v_k, h_m)$ vs $|\eta(\vec{X})|$ for one pair of $(k,m)$ neurons. The entire process is repeated for every pair of $(k,m)$ neurons. Since there are $O(np)$ neurons and each of which gives a two particle density matrix, the total computation of all $pn$ two-particle density matrices incurs a cost of $O(n^2p^2*N_E)$. Thus from start to finish our entire protocol of training the network $G=(V,E)$ and the subsequent construction of the $I(v_k, h_m)$ vs $|\eta(\vec{X})|$ space is $O(poly(n,p))$.

\subsection{Hamiltonian of the Drivers}\label{Ham_drivers_section}
To exemplify the consequences further, we now use two drivers namely the {\color{black} Transverse Field Ising model (TFIM) and the  concentric-TFIM \cite{vitagliano2010volume} (c-TFIM) for a system of $N=4,6,8,10,12,14,16,18,20,24$ spins}. The generic hamiltonian for the drivers can be written as 

\begin{align}
H &= -B\sum_{i_d}^N \sigma^x(i_d) - \sum_{{i_d}{j_d}} J_{{i_d}{j_d}} \sigma^z(i_d) \sigma^z(j_d) \label{gen_Ham}
\end{align}
For TFIM the matrix elements of $J_{{i_d}{j_d}}$ are 
\begin{align}
J_{{i_d}{j_d}} &= 
\begin{cases}
J_0 , & \textit{if}\:\: if j_d= i_d \pm 1 \:\:\: \forall \:\: i_d \\
0 & \text{otherwise} 
\end{cases}
\label{eq:TFIM_Ham}
\end{align}

For c-TFIM, the elements $J_{{i_d}{j_d}}$ are
\begin{align}
J_{{i_d}{j_d}} &= 
\begin{cases}
J_0 , & \textit{if}\:\: i_d = \frac{N}{2}-(q-1), \nonumber \\  & j_d= \frac{N}{2}+q \\
0 & \text{otherwise} 
\end{cases}\\
\text{with} \: &\forall\:\: q \in [1, \frac{N}{2}]
\label{eq:conc_TFIM_Ham}
\end{align}
While Eq.\ref{eq:TFIM_Ham} due to nearest-neighbor interactions (see Fig.\ref{Fig:LB_saturation}(a)) admits an area-law scaling ground state which can only be augmented to a logarithmic correction \cite{koffel2012entanglement, kuwahara2020area}, the connectivity graph of Eq.\ref{eq:conc_TFIM_Ham} (see Fig.\ref{Fig:LB_saturation}(b)) necessitates a volume-law scaling (refer to Section \ref{VNE_plots} in Appendix for direct corroboration). Since we choose $(B \ge 0,J_0 \ge 0)$, the ground state of both drivers have non-negative coefficients due to Perron-Frobenius theorem \cite{PhysRevB.100.195125,horn_johnson_2012} and undergoes a phase transition from an ordered ferromagnet to the disordered phase owing to spontaneous breaking of $\mathcal{Z}_2$ symmetry ($\sigma^z(i_d) \to -\sigma^z(i_d)$ or $\pi$ rotation around $\sigma^x(i_d)$) 
as $g=\frac{B}{J_0}$ is enhanced. 
\begin{figure*}[!htb]
     \centering   \includegraphics[width=0.99\textwidth]{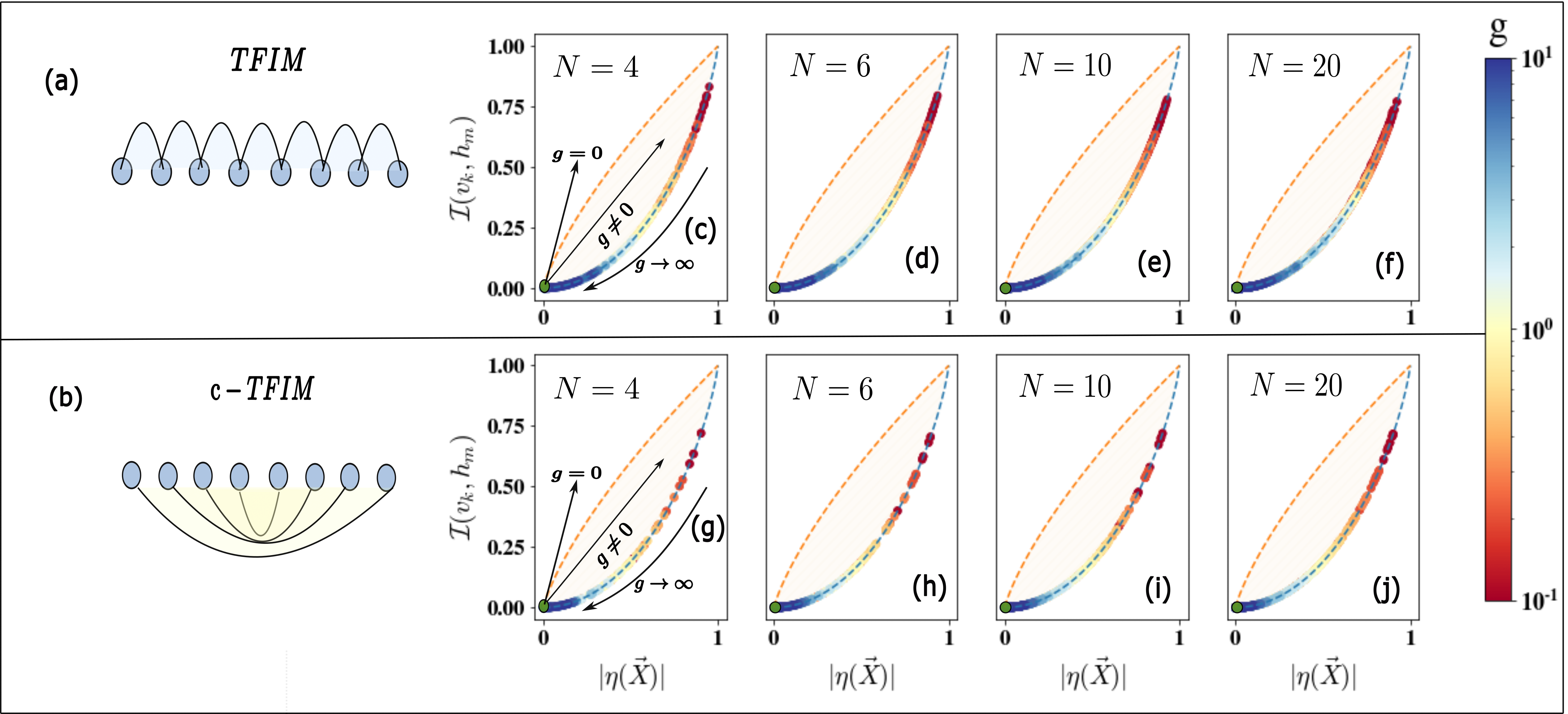}
    \caption{(a) The schematics of the interactions among the spin registers in the area-law entangled ground state of the driver Hamiltonian defined in Eq.\ref{eq:TFIM_Ham}(TFIM) \cite{PFEUTY197079,pfeuty1971ising,stinchcombe1973ising} (b) The schematics of the interactions among the spin registers in the volume-law entangled ground state of another driver system defined in Eq.\ref{eq:conc_TFIM_Ham}(c-TFIM) \cite{Vitagliano_2010} (see text for more details and Section \ref{VNE_plots} in Appendix). The $\mathcal{I}(v_k, h_m)$ and $\eta(\vec{X})$ values corresponding to the finally converged parameter $\vec{X}^*$ of the trained network $G=(V,E)$ for the ground state of Eq.\ref{eq:TFIM_Ham} for (c) N=4, (d) N=6, (e) N=10, (f) N=20 spins in the driver in Eq.\ref{eq:TFIM_Ham}(TFIM). The $\mathcal{I}(v_k, h_m)$ and $\eta(\vec{X})$ values corresponding to the finally converged parameter $\vec{X}$ of the trained network $G$ for the ground state of Eq.\ref{eq:conc_TFIM_Ham}(c-TFIM) for (g) N=4, (h) N=6, (i) N=10, (j) N=20 spins in the driver. In all eight plots (c-j) we see that the representation chosen by the trained network  in the emergent $\mathcal{I}-\eta$ space always saturates the LB (Eq. \ref{eq:LB}) for all $g \in (0,\infty]$. The behavior remains intact irrespective of the changing size of the driver system and has been observed for a wider class of drivers too arising in different physical problems (See also Section \ref{Sat_LB_Section} in Appendix) }
     \label{Fig:LB_saturation}
\end{figure*}

\begin{figure*}[!htb]
     \centering   \includegraphics[width=0.99\textwidth]{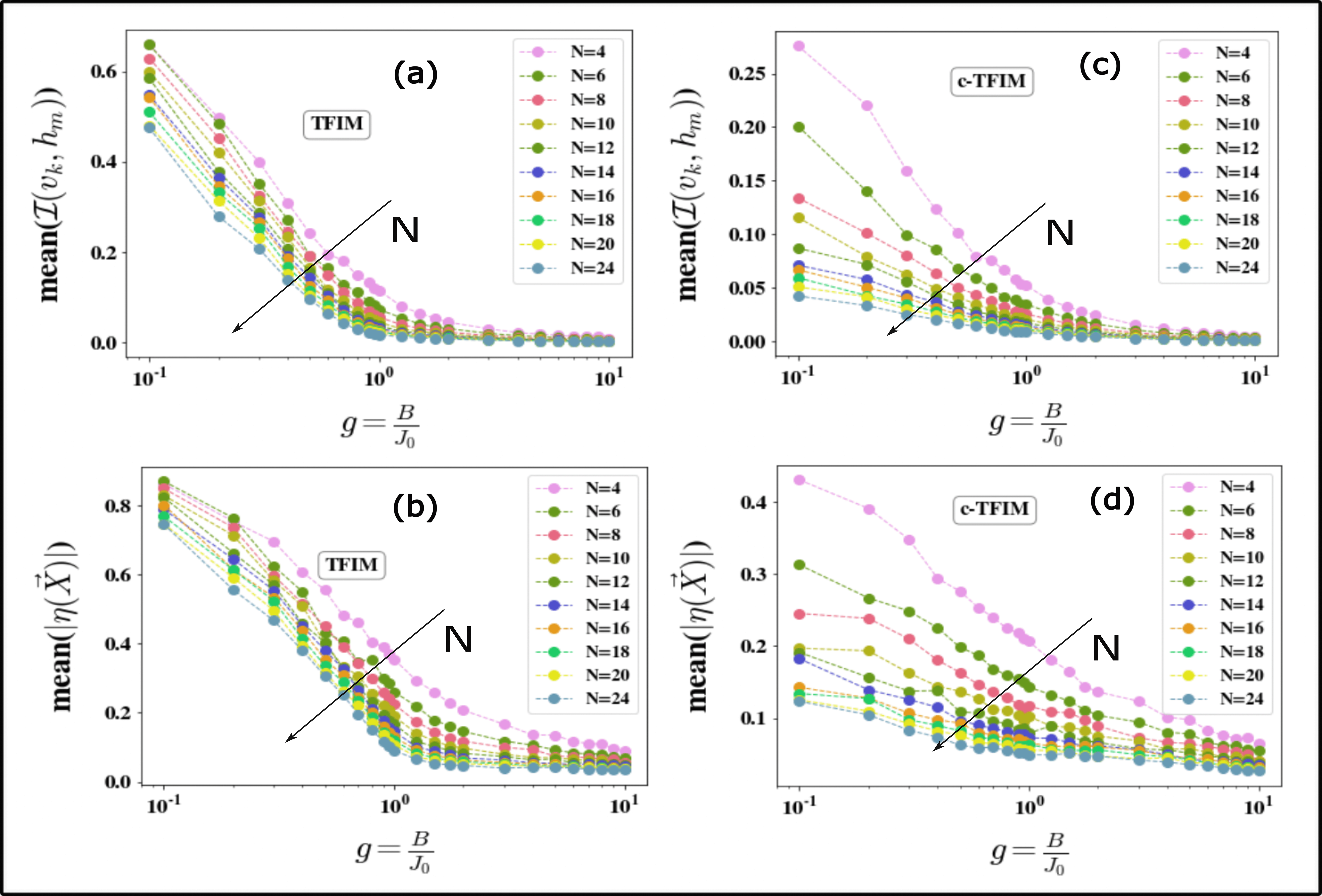}
    \caption{(a,c) The plot of the mean of $\mathcal{I}(v_k, h_m)$ vs $g$ from the likes of Fig.\ref{Fig:LB_saturation}(c-j) where the averaging is done over all pairs of $(k,m)$ and over many initialization at a given $g$ for various values of driver spins $N=4,6,8,10,12,14,16,18,20,24$ for both models defined by Eq.\ref{eq:TFIM_Ham} in (a) and Eq.\ref{eq:conc_TFIM_Ham} in (c) (b,d) Similar plot as in Fig.\ref{Fig:MI_eta_spin_transfer}(a,c) but for $\eta(\vec{X})$ as a function of $g$ for various sizes of driver spins in both the models. In all plots Fig.\ref{Fig:MI_eta_spin_transfer}(a,d) transferrence of spin correlation from both drivers to the visible and latent units of the learner across phase boundaries for drivers of even large sizes is apparent (see also Fig.\ref{Fig:TFIM_cTFIM_all_N_LB_sat}) and so is the higher range of variability of the mean in volume law connected ground states in Fig.\ref{Fig:MI_eta_spin_transfer}(c,d). It must be emphasized that the mean $\mathcal{I}(v_k, h_m)$or mean $\eta(\vec{X})$ (see Eq.\ref{eq:Cov_imag_part}) defined is between the visible and the latent/hidden spins. The hidden spins are oblivious to the driver and are only known to the learner network. The configurations of the visible spins of the driver provides a convenient basis to the spins of the driver in both the models. The hidden spins on the other hand are just responsible for supporting the visible spins in adequate feature extraction. Yet we see in these plots that the spin correlation that exists between the spins/sub-units of the driver across a phase transition is directly transferred or mirrored between the visible and latent spins of the trained learner as desirable features. Such an observation affords a direct numerical quantification to the assistance provided by the hidden spins in training the network $G=(V,E)$. This kind of emergent insight into the training of the graph-neural network is central to our work and as evidentiated can easily be procured through the newly constructed $I(v_k, h_m)$ vs $|\eta(\vec{X})|$ space }
     \label{Fig:MI_eta_spin_transfer}
\end{figure*}

\section{Results and Discussion}\label{Result_Disc}
To simulate the drivers discussed in Section \ref{Ham_drivers_section} for $N=4,6,8,10,12,14,16,18,20,24$ spins, we use $n=p=N$ in the learner $G$ and run several numerical experiments with different randomly chosen initial parameters with the error threshold for convergence set to $\le 0.1 \%$ (see Section \ref{training_algo_section} for the  training algorithm and Section \ref{Train_plots} in Appendix for convergence plots of training) for each. To see the effect of changing hidden node density $p/n$ see Section \ref{Hid_node_density_Section} in Appendix. We use the finally converged $\vec{X^*}$ obtained from the training to construct the eigenvalues of $^2\rho(v_k,h_m)$ and one-particle reduced density matrices ($^1\rho(v_k)$ or $^1\rho(h_m)$) and eventually compute $I(v_k, h_m)$ and $|\eta(\vec{X})|$ as illustrated in Section \ref{training_algo_section}. We do this for each pair $(k,m)$ choosing one from the set of visible and the other from the set of hidden neurons in the learner $G=(V,E)$.
We display the results of our computation in the $\mathcal{I}-\eta$ space (illustrated in Fig.\ref{Fig:MI_eta}) in Fig.\ref{Fig:LB_saturation}(c-j) for all such pairs at various $g$ values of the drivers (see Section \ref{Ham_drivers_section} for a definition of $g$). For $N=4,6,10,20$ in the respective models, we see surprisingly that the representation chosen by the trained learner in the $\mathcal{I}-\eta$ space always saturates LB in Fig.\ref{Fig:LB_saturation}(c-f)and slides along it monotonically for $g \in (0,\infty]$. For $g=0$, the representation of $G$ is devoid of any correlation between sub-systems of the visible and hidden neurons $\forall (k,m)$
with a cluster of points near the green dot (0,0) (especially marked in Fig.\ref{Fig:LB_saturation}(c,g) but is true for all plots in Fig.\ref{Fig:LB_saturation}(c-j)) corresponding to the uncorrelated, two-fold degenerate ferromagnetic ground state ($|0000\rangle$ or $|1111\rangle$)\cite{PhysRevB.100.195125}.
For a direct corroboration for all other sizes and in a wide variety of other driver spin models even beyond the drivers discussed in Section \ref{Ham_drivers_section}, see Section \ref{Sat_LB_Section} in Appendix. From such numerical evidences, we infer a newly discovered learning principle that has never been discussed or investigated before. For a wide variety of drivers, we see that the trained network $G=(V,E)$ when entrusted with learning a probability density function to mimic the amplitude field of a desired quantum state, invariably chooses a representation that minimizes mutual information ($\mathcal{I}(v_k, h_m)$) between the visible and the hidden sub-units for a given covariance ($\eta(\vec{X})$). This further highlights the importance of the $\mathcal{I}(v_k, h_m)$ vs $\eta(\vec{X})$ space we investigate here. The pursuance of universality of this result through a formal proof of the statement or under what conditions it fails if at all may be undertaken in the future.

While lower bound saturation remains true at all values of $g$, as is evident from Fig.\ref{Fig:LB_saturation}(c,j) that for 
$0 \le g \le 1$ the density of points associated with the representation chosen by $G$ shifts dramatically away from the green dot towards the red points and then eventually returns back towards the uncorrelated state (blue points) ($\forall\: (k \in \lceil n \rceil, m \in \lceil p \rceil$ where $n=p=N$ i.e. $n*p$ pairs of $(k,m)$ values) . This is explicitly marked in Fig.\ref{Fig:LB_saturation}(c,g) but is true for all plots in Fig.\ref{Fig:LB_saturation}(c-j).
To study the details of such an occurrence and consolidate the observation, we plot in Fig.\ref{Fig:MI_eta_spin_transfer}(a-d) the sample means of $I(v_k, h_m)$ and $|\eta(\vec{X})|$ as constructed from the finally converged $\vec{X^*}$ by training $G$ averaged not only over all pairs $(k,m)$ for a given experiment but also over converged runs arising from different initialization (i.e. the averages of points plotted along LB in Fig.\ref{Fig:LB_saturation}(c-j) over all pairs of $(k,m)$ and over several initialization such that for a given $g$ there is a single representative $I(v_k, h_m)$ and $|\eta(\vec{X})|$ value). To get better idea about the statistic, the standard deviations associated with the averaging process and the standard error of the mean is displayed in Section \ref{Sample_SE_TFIM_cTFIM} in Appendix. For $g \to 0_{+}$ the ground state of the driver exhibits a superposition of several bit-strings/spin configurations which the network $G$ now mirrors by choosing a representation that has a significantly higher $\mathcal{I}(v_k, h_m)$ and $\eta(\vec{X})$ (and hence correlation) among the bi-partitions between visible and hidden neurons (red dots in Fig.\ref{Fig:LB_saturation}(c-j)). This is further reflected in the higher mean $\mathcal{I}(v_k,h_m)$ in Fig.\ref{Fig:MI_eta_spin_transfer}(a,c)
and higher mean $\eta(\vec{X})$ in Fig.\ref{Fig:MI_eta_spin_transfer}(b,d) for each non-zero value of $g$ than in the $g=0$ case. 

For $g\to \infty$, the drivers once again display an uncorrelated state with each spin in state $\frac{|0\rangle + |1\rangle}{\sqrt{2}}$. To ape this limit, the representation chosen by $G$ thus gradually slides towards the $(0,0)$ point in $\mathcal{I}-\eta$ space (blue dots in Fig.\ref{Fig:LB_saturation}(c-j)) with a concomitant decline in the respective averages in Fig.\ref{Fig:MI_eta_spin_transfer}(a-d) as $g$ is enhanced.
For a given finite $g$, the crucial difference between the drivers in Eq.\ref{eq:TFIM_Ham} and Eq.\ref{eq:conc_TFIM_Ham} is captured in the higher variability in the respective means (see Fig.\ref{Fig:MI_eta_spin_transfer}(c-d)) indicating many compatible/equivalent representations chosen by the network for a correlated volume-law entangled state for all sizes. This is true for each individual size of the driver model used i.e. for each $N$ and is best illustrated from the plot of the standard deviations (associated with the averages in Fig.\ref{Fig:MI_eta_spin_transfer}) displayed in Section \ref{Sample_SE_TFIM_cTFIM} in Appendix which shows that in c-TFIM at a given $g$ (especially $g \to 0_{+}$ regime), the trained learner consists of several different $(k,m)$ pairs with widely varying correlation properties ($I(v_k, h_m)$ and $|\eta(\vec{X})|$ values) arising from compatible chosen for the same learned state. This is much more than the case for TFIM even at a given size $N$ and given $g$ value indicating the ability of the learner to distinguish area-law vs volume-law connectivity. However with increasing $g$ unanimity sets in as both models displays uncorrelated ground state with unique configurations (see also Section \ref{Sample_SE_TFIM_cTFIM} in Appendix). 

Thus in summary the observations indicate three important inferences. Apart from the saturation of LB in $\mathcal{I}(v_k, h_m)$ vs $|\eta(\vec{X})|$ space (corroborated for all sizes and many other spin models in Section \ref{Sat_LB_Section} in Appendix) as illustrated in Fig.\ref{Fig:LB_saturation}(c-j), we also see mirroring of the spin correlation behavior across phase transition between spins of the driver in the correlation introduced between the spins of the visible and latent neurons of the trained state of the learner $G$ in Fig.\ref{Fig:MI_eta_spin_transfer}(a-d). We also see many equivalent representations of the network for differentiating exotic volume-law connectivity in the driver (see Fig.\ref{Fig:MI_eta_spin_transfer} and Section \ref{Sample_SE_TFIM_cTFIM} in Appendix). It must be emphasized that the last two assertions are true even though the latent neurons are directly oblivious and unrelated to the driver (see the description of the network and how it acts an a variational ansatz for the neural-network encoding of the quantum state as illustrated in Section \ref{Graph_description}).
Only the configurations of the visible node register are directly related to the spins of the driver and forms a basis for the eigenspace of the driver. The conventional wisdom is that latent neurons with their respective additional parameters provide support by enhancing the expressibility of the network $G$. Our observations thus collectively can serve as a stepping stone towards formalizing and quantifying the important role performed by the latent spins of the learner $G$ from a newly obtained perspective of the $\mathcal{I}(v_k, h_m)$ vs $|\eta(\vec{X})|$ where physical behavior of such neurons are investigated from the lens of correlation exchange which happens surreptitiously under the hood during training of $G$.

\section{Conclusion}\label{concluding_section}

In this work, we established a number of key physical insights about the training of the learner network $G$. The choice of this specific network is attributed to its astonishing success in simulating a wide variety of quantum systems in condensed matter physics like strongly correlated fermionic assemblies \cite{Strong_corr_RBM,choo2020fermionic,ChNg2017}, topologically non-trivial phases\cite{Topo_state_RBM,doi:10.1063/5.0128283}, anyonic symmetries\cite{RBM_anyons_symm}, in quantum dynamical evolution\cite{PhysRevResearch.3.023095,carleo2017solving}, in chemistry like 2D-materials\cite{sajjan2021quantum,sureshbabu2021implementation,sajjan2022quantum} and molecules with multi-reference correlation\cite{RBM_molecule} under geometric distortion\cite{Xia_2018}, and even in classification tasks\cite{Ciliberto2017,Carrasquilla_2017} with quantum or classical data. In fact, $G$ has been proven to act as a universal approximator for any probability density \cite{Melko2019a, Roux_RBM} thereby providing a guarantee to the range of its representational capacity. Prior work has also established that the network is capable of mimicking the amplitude of a $2^n$ dimensional volume-law entangled quantum state even with a sparse representation\cite{PhysRevX.7.021021} i.e. using $O(n)$ parameters as opposed to $O(np)$ in the usual case. Ref\cite{Long2010} has established that the task of retrieving the full distribution encoded within $G$ would always entail exponential classical resources unless the polynomial hierarchy collapses. However, quantum circuits to efficiently sample from the same with quadratic qubit and gate costs $(O(np))$ already exist. Another feature enjoyed by the network is its easy extension to $d$-dimensional spins \cite{Topo_state_RBM} which makes encoding higher dimensional quantum states of a driver tractable.  

For such a widely recognized network, we have illustrated how changing communication within the sub-units of the learner can be understood by introducing imaginary components of OTOC and have analytically established its relationships with bipartite mutual information. Use of the real part of OTOCs are gaining attention in learning algorithms \cite{shen2020information,wu2021scrambling,garcia2022quantifying} including how quantum learning may be advantageous with such scrambling measures \cite{schuster2022learning}, experimental measurement of OTOC on quantum circuits \cite{garttner2017measuring,mi2021information}, bounds on loss function by OTOCs \cite{garcia2022quantifying}, presence of barren plateaus for simulating large scrambling unitaries using parameterized circuits \cite{holmes2021barren}. We see from our analytical deduction that the real part of OTOC between the $k$-th visible neuron and $m$-th hidden neuron even though sensitive to $W^k_m$ through an oscillatory temporal dependance, is completely insensitive to $(\vec{a},\vec{b})$ unlike the imaginary part. It must be emphasized that the imaginary-part of OTOC is hitherto unexplored and unreported and is not only new for this network but even for other problems wherein OTOC has been used.

A direct by-product of our mathematical approach are the discovery of the several conservation theorems/invariants of motion (see Theorem \ref{lem1}) each of which is equipped with a Lie Algebraic generator that preserves the phase space of OTOC by mapping the the invariants onto itself. We plot both the invariants in Section \ref{inv_training_plot_section} of the Appendix for a prototypical example of TFIM with $N=10$ spins. It is clear from the analysis that only the ones constructed using imaginary components actually change during training. The invariants from real components remain fixed and is insensitive to training epoch. This decisively shows that the importance of imaginary components emphasizing the fact that to procure any information about the learning landscape from these invariants, these components are invincible and the only options available. Such insight into the phase space of OTOC trajectories of network $G$ and the connection to a hidden Lie-Algebraic framework was not procured before and thus provides a richer characterization of the problem which has escaped prior attention. The mathematical machinery so developed could be used for making similar deductions of OTOC strings in other physical systems and other neural network architectures too.

A direct inter-relationship between the imaginary part of such a four-point correlation function and a two-point correlation function and eventually with mutual information was thereafter deduced. Using the newly framed lens of $\mathcal{I}-\eta$ space, we have established three different conclusions associated with training the network - saturation of lower bound for a wide variety of spin models, quantifiable transferrence of spin correlation from the driver to the hidden and visible units of the learner (this is further corroborated using another property computed solely from the learner as demonstrated in Section \ref{Fischer_plots} in Appendix) and compatible representations chosen by the learner for distinguishing complex connectivity. Such a study begins to probe into the rich underlying universe of the training mechanism and shows that the representations chosen by the learner during the training epoch are quantifiably tuned to re-adjust the correlation content among the pairwise bi-partitions of visible and hidden neurons commensurate with the changing quantum correlation in the actual driver system. This is despite the fact that the latent spins of the learner are not directly involved in emulating the spins of the driver, they act merely as supportive instruments to accessorize the spins in the visible-node register by enhancing their expressivity. 


Applications of such observations may be beneficial to
physics-inspired learning \cite{bellinger2020reinforcement,musil2021physics,karniadakis2021physics} and also can be leveraged to expedite training through apriori informed initialization near lower bound (LB). Extension to explore the ramifications of the findings in classical ML tasks like collaborative filtering and to other important networks like dreaming neural networks \cite{FACHECHI201924,agliari2019dreaming,aquaro2022recurrent}, single-layer auto-encoders \cite{Autoencode_review} may be undertaken. 
Although all results are analytically established, since direct measurement of OTOCs is a possibility even on a quantum circuit\cite{mi2021information, PhysRevLett.128.160502}, experimental measurement of the imaginary component of the OTOC string for our learner can be undertaken and its relationship with $\mathcal{I}$ be exploited. With the present-day promise of machine/deep learning, the authors hope that more such studies will be initiated which in spirit `humanizes' already-established black-box models by probing into their learning universe, removes the shroud of mystery behind their training, and will hopefully lead to enhanced cross-pollination with physical sciences and ameliorated standards in model-development. The benefits of such a pipeline can surely be harvested to achieve unprecedented feats in simulating the natural world.


\section{ACKNOWLEDGEMENTS}
The authors would like to acknowledge the financial support of the U.S. Department of Energy, Office of Science, National Quantum Information Science Research Centers, Quantum Science Center, and the National Science Foundation under Award Number 1955907.
\bibliography{ref,Refs_Exc_st_RBM}

\widetext
\pagebreak
\appendix
\AtAppendix{\counterwithin{lemma}{section}}



\def\mathunderline#1#2{\color{#1}\underline{{\color{black}#2}}\color{black}}

\section{General formulation for OTOC}\label{Gen_OTOC_form}

As indicated in the text, every OTOC comprises of two unitary operators $U_1(0), U_2(0)$ chosen usually at two non-local sites (here $\{1,2\}$ for instance) within a system. The rate at which the information propagates through the system in real-time is thereafter quantified using the
$C_{U_1,U_2}(t)=\langle U^\dagger_1 (0) U^\dagger_2(t) U_1(0) U_2(t) \rangle$, 
For measurement purposes, the $Re(C_{U_1,U_2}(t))$ is often related in literature to the commutator product given as $\langle [U_1(0), U_2(t)]^\dagger [U_1(0), U_2(t)] \rangle$ where by virtue of construction of the observable, positive semi-definiteness is ensured. The origin of such a choice is attributed to studying chaotic classical systems wherein an analogous expression for Poisson bracket $\{x(0), p(t)\}_{PB}$ is used to probe sensitivity to initial conditions.  Since in our work we show that the imaginary part  $Im(C_{U_1,U_2}(t))$ can also be informative about correlation content within the sub-units of the learner $G$, we offer herein a general formulation for OTOCs of arbitrary systems using newly constructed positive semi-definite operators involving not only commutators ($[\cdot]$) as above but anti-commutators ($\{\cdot\}_{+}$). The advantage of the formulation is the offered generality in the theory of OTOCs in the quantum domain and also the flexibility in projecting both the real and imaginary component of $C_{U_1,U_2}(t)$ 
based on user-defined preferences while still maintaining positive semi-definiteness of the observables probed. To this end, let us define a probe as 
\begin{align}
L_{U_1,U_2}(\lambda_1, \lambda_2, t) &= \lambda_1 A(U_1, U_2, t) + i \lambda_2 B(U_1, U_2, t) \label{L_defn}\\
A(U_1, U_2, t) &= \{U_1(0),U_2(t)\}_{+} \label{Anti_Comm_defn}\\
B(U_1, U_2, t) &= [U_1(0),U_2(t)] \label{Comm_defn}\\
\lambda_1, \lambda_2 &\in \mathcal{R} 
\end{align}
Using Eq.\ref{L_defn} and the definition for $C_{U_1,U_2}(0, 0, \vec{X}, t)$ , it is easy to show
\begin{align}
Re(C_{U_1,U_2}(0, 0, \vec{X}, t)) &= Re(\langle U^\dagger_1 (0) U^\dagger_2(t) U_1(0) U_2(t) \rangle) \nonumber \\
&=1 - \frac{\langle B(U_1, U_2, t)^\dagger B(U_1, U_2, t)\rangle}{2} \nonumber \:\:\:\because\:\: Eq.\ref{Comm_defn} \\
& = 1 - \frac{\langle L_{U_1,U_2}(0, 1, t)^\dagger L_{U_1,U_2}(0, 1, t)\rangle}{2} \:\:\:\because\:\:Eq.\ref{L_defn}
\end{align}
and 
\begin{align}
Im(C_{U_1,U_2}(t)) &= Im(\langle U^\dagger_1 (0) U^\dagger_2(t) U_1(0) U_2(t) \rangle) \nonumber \\
&=-\frac{1}{2}(\langle A(U_1, U_2, t)^\dagger iB(U_1, U_2, t)\rangle) \:\:\:\because\:\:Eq.\ref{Anti_Comm_defn}, \ref{Comm_defn}\nonumber \\
&=\frac{1}{4}\langle L_{U_1,U_2}(0, 1, t)^\dagger L_{U_1,U_2}(0, 1, t)\rangle + \frac{1}{4}\langle L_{U_1,U_2}(1, 0, t)^\dagger L_{U_1,U_2}(1, 0, t)\rangle \nonumber \\ &-\frac{1}{4}\langle L_{U_1,U_2}(1, 1, t)^\dagger L_{U_1,U_2}(1, 1, t)\rangle \:\:\:\because\:\:Eq.\ref{L_defn}
\end{align}

Note that all combinations of $ L_{U_1,U_2}(\lambda_1, \lambda_2, t)^\dagger L_{U_1,U_2}(\lambda_1, \lambda_2, t)$ are positive semi-definite by construction and hence is used in the same vein as the usual commutators for OTOCs are traditionally defined but nonetheless offers a much more general framework for investigating the operator string $C_{U_1,U_2}(t)$\\

\newpage
\section{Time dependence of $\{\sigma^\alpha(v_k,0),\sigma^\beta(h_m,t)\}_+$ and $[\sigma^\alpha(v_k,0),\sigma^\beta(h_m,t)]$}
\label{time_dep_paulis}

To define the time-dependence of $\{\sigma^\alpha(v_k,0),\sigma^\beta(h_m,t)\}_+$ and $[\sigma^\alpha(v_k,0),\sigma^\beta(h_m,t)]$, it is essential to establish time-dependence for $\sigma^\beta(h_m, t)$ where $\alpha, \beta \in \{0,x,y,z\}$ with $\sigma^0=\mathcal{I}_{2 \times 2}$. To this end, we prove the following lemma. 

\begin{lemma} \label{lem:sigma_t_dep}
For an operator $\sigma^\beta(h_m, t)$ satisfying the equation $\dot{\sigma}^\beta(h_m, t) = i[\mathcal{H}, \sigma^\beta(h_m, t)]$, (generator $\mathcal{H}$ defined in Eq.1 in main manuscript), the solution would be $\sigma^\beta(h_m, t) = e^{2i\mathcal{H}^\prime t}\sigma^\beta(h_m, 0)$ $\forall\:\beta \in \{0,x,y,z\}$
with
\begin{align*}
\mathcal{H}^\prime &= \delta_{\beta\in\{x, y\}}(\mathcal{H} - \sum_l a_l \sigma^z(v_l,0) - \sum_{j \neq m}b_j \sigma^z(h_j,0) - \sum_{l, j \neq m} W^l_j \sigma^z(v_l,0)\sigma^z(h_j,0))
\end{align*}
\end{lemma}

\begin{proof}
\begin{align}
\dot{\sigma}^\beta(h_m, t) &= i[\mathcal{H}, \sigma^\beta(h_m, t)] \nonumber \\
&= ie^{i\mathcal{H} t} [\mathcal{H}, \sigma^\beta(h_m, 0)] e^{-i\mathcal{H} t} \nonumber \\
&= ie^{i\mathcal{H} t} (b_m[\sigma^z(h_m, 0), \sigma^\beta(h_m, 0)] + \sum_{l} W^l_m[\sigma^z(v_l,0)\sigma^z(h_m, 0), \sigma^\beta(h_m, 0)])e^{-i\mathcal{H} t} \nonumber \\
& = 2i\delta_{\beta\in\{x, y\}} e^{i\mathcal{H} t}(b_m\sigma^z(h_m, 0)\sigma^\beta(h_m, 0) 
+ \sum_{l} W^l_m\sigma^z(v_l,0)\sigma^z(h_m, 0) \sigma^\beta(h_m, 0))e^{-i\mathcal{H} t} \nonumber \\
&= 2i\delta_{\beta\in\{x, y\}} (b_m\sigma^z(h_m, 0) + \sum_{l} W^l_m\sigma^z(v_l,0)\sigma^z(h_m, 0))\sigma^\beta(h_m, t) \nonumber \\
& = 2i\delta_{\beta\in\{x, y\}}(\mathcal{H} - \sum_l a_l \sigma^z(v_l,0) - \sum_{j \neq m}b_j \sigma^z(h_j,0) - \sum_{l, j \neq m} W^l_j \sigma^z(v_l,0)\sigma^z(h_j,0))\sigma^\beta(h_m, t) \nonumber \\
&= 2i\mathcal{H}^\prime \sigma^\beta(h_m, t) \nonumber \\
\sigma^\beta(h_m, t) &= e^{2i\mathcal{H}^\prime t}\sigma^\beta(h_m, 0) \nonumber
\end{align}
\end{proof}

\begin{lemma}\label{lem:Comm_t_dep}
The explicitly time dependant forms of $[\sigma^\alpha(v_k,0), \sigma^\beta(h_m, t)]$ and $\{\sigma^\alpha(v_k,0), \sigma^\beta(h_m, t)\}_+$ are
\begin{enumerate}
\item{
$\begin{aligned}[t]
&\hat{\Theta}_{\alpha, \beta}(t)=[\sigma^\alpha(v_k,0), \sigma^\beta(h_m, t)]\nonumber\\ 
&= -2i\: Sin(2W^k_m t)\:\det\begin{pmatrix} \delta_{\alpha z} & \delta_{\alpha \beta} & \delta_{\alpha \gamma} \\ 1 & \delta_{z \beta} & \delta_{z \gamma} \\
\delta_{\chi z} & \delta_{\chi \beta} & \delta_{\chi \gamma}\end{pmatrix}\mathcal{\hat{F}}\:\sigma^\chi(v_k,0)\sigma^\gamma(h_m,0) \nonumber \\
\end{aligned}$}

\item{
$\begin{aligned}[t]
&\hat{\Phi}_{\alpha, \beta}(t)=\{\sigma^\alpha(v_k,0), \sigma^\beta(h_m, t)\}_{+} \nonumber \\
&= -2i\: Sin(2W^k_m t)\:\det\begin{pmatrix} \delta_{\alpha z} & \delta_{\alpha \beta} & \delta_{\alpha \gamma} \\ 1 & \delta_{z \beta} & \delta_{z \gamma} \\
\delta_{\chi z} & \delta_{\chi \beta} & \delta_{\chi \gamma}\end{pmatrix}\mathcal{\hat{F}}\:\sigma^\chi(v_k,0)\sigma^\gamma(h_m,0) \nonumber +2 e^{2i\mathcal{H}^\prime t} \sigma^\alpha(v_k,0) \sigma^\beta(h_m,0) \nonumber\\
\end{aligned}$}
\end{enumerate}
with $\mathcal{\hat{F}} = e^{2ib_m\sigma^z(h_m,0)t}\Pi_{l \ne k}e^{2iW^l_m \sigma^z(v_l,0)\sigma^z(h_m,0)t}$
being a unitary operator. For definition of $\mathcal{H}^\prime$ see Lemma.\:\ref{lem:sigma_t_dep}. Here we restrict ${\alpha, \beta} \in \{x,y,z\}$ as the commutation / anti-commutation relations with $\mathcal{I}_{2\times 2}$ are trivial.

\end{lemma}
\begin{proof}
$\newline$
\begin{enumerate}
\item{
$\begin{aligned}[t]
\hat{\Theta}_{\alpha, \beta}(t) &= [\sigma^\alpha(v_k,0), \sigma^\beta(h_m, t)] \nonumber \\
&= [\sigma^\alpha(v_k,0), e^{2i\mathcal{H}^\prime t}]\sigma^\beta(h_m,0)\nonumber \:\:\: \because \:\:\:\:\text{Lemma}\:\ref{lem:sigma_t_dep} \\
&= e^{2ib_m\sigma^z(h_m,0)t}\Pi_{l \ne k}e^{2iW^l_m \sigma^z(v_l,0)\sigma^z(h_m,0)t}\nonumber [\sigma^\alpha(v_k,0), e^{2iW^k_m \sigma^z(v_k,0)\sigma^z(h_m,0)t}]\sigma^\beta(h_m,0)\nonumber \\
&= \mathcal{\hat{F}}[\sigma^\alpha(v_k,0), i\sigma^z(v_k,0)\sigma^z(h_m,0)]\sigma^\beta(h_m,0) \:Sin(2W^k_m t)\nonumber \\
&= -2i (\epsilon_{\alpha z \chi})(\epsilon_{z\beta\gamma})\:Sin(2W^k_m t)\mathcal{\hat{F}}\:\:\sigma^\chi(v_k,0) \sigma^\gamma(h_m,0) \nonumber \\
&= -2i\: Sin(2W^k_m t)\:\det\begin{pmatrix} \delta_{\alpha z} & \delta_{\alpha \beta} & \delta_{\alpha \gamma} \\ 1 & \delta_{z \beta} & \delta_{z \gamma} \\
\delta_{\chi z} & \delta_{\chi \beta} & \delta_{\chi \gamma}\end{pmatrix}\mathcal{\hat{F}}\:\sigma^\chi(v_k,0)\sigma^\gamma(h_m,0) \nonumber
\end{aligned}$}

\item{
$\begin{aligned}[t]
\hat{\Phi}_{\alpha, \beta}(t)&=\{\sigma^\alpha(v_k,0), \sigma^\beta(h_m, t)\}_{+} \nonumber \\ 
&= \{\sigma^\alpha(v_k,0), e^{2i\mathcal{H}^\prime t}\sigma^\beta(h_m,0)\}_{+}\nonumber \:\:\: \because \:\:\:\:\text{Lemma}\:\ref{lem:sigma_t_dep} \\ 
&= [\sigma^\alpha(v_k,0), e^{2i\mathcal{H}^\prime t}]\sigma^\beta(h_m,0)+ e^{2i\mathcal{H}^\prime t}\{\sigma^\alpha(v_k,0), \sigma^\beta(h_m, 0)\}_+\nonumber \\
&= [\sigma^\alpha(v_k,0), e^{2i\mathcal{H}^\prime t}]\sigma^\beta(h_m,0)+ 2e^{2i\mathcal{H}^\prime t}\sigma^\alpha(v_k,0)\sigma^\beta(h_m, 0) \nonumber \\
&= -2i\: Sin(2W^k_m t)\:\det\begin{pmatrix} \delta_{\alpha z} & \delta_{\alpha \beta} & \delta_{\alpha \gamma} \\ 1 & \delta_{z \beta} & \delta_{z \gamma} \\
\delta_{\chi z} & \delta_{\chi \beta} & \delta_{\chi \gamma}\end{pmatrix}\mathcal{\hat{F}}\:\sigma^\chi(v_k,0)\sigma^\gamma(h_m,0) + 2e^{2i\mathcal{H}^\prime t}\sigma^\alpha(v_k,0)\sigma^\beta(h_m, 0)\nonumber \\ 
&\:\:\:\because \:\:\:\:\text{Lemma}\:\ref{lem:Comm_t_dep}(1)
\end{aligned}$}
\end{enumerate}
\end{proof}

\begin{lemma}\label{lemma_Diff_eqn}
As defined in Lemma \ref{lem:Comm_t_dep}, if $\hat{\Theta}_{\alpha, \beta}(t)= [\sigma^\alpha(v_k,0), \sigma^\beta(h_m, t)]$ \\ and 
$\hat{\Phi}_{\alpha, \beta}(t) = \{\sigma^\alpha(v_k,0), \sigma^\beta(h_m, t)\}_{+}$, then $\forall \:\: {\alpha, \beta} \in \{0,x,y,z\}$ , $\hat{\Theta}_{\alpha, \beta}(t)$ and $\hat{\Phi}_{\alpha, \beta}(t)$ satisfies the following operator differential equations:
\begin{enumerate}
\item{
$\begin{aligned}[t]
&\frac{1}{2}(\frac{\partial^2 \hat{\Theta}_{\alpha, \beta}(t)^\dagger \hat{\Theta}_{\alpha, \beta}(t) }{\partial t^2})\nonumber = \delta_{\alpha \in \{x,y\}}\delta_{\beta \in \{x,y\}}(4W^k_m)^2 (\mathcal{I}_{2\times2}-\frac{\hat{\Theta}_{\alpha, \beta}(t)^\dagger \hat{\Theta}_{\alpha, \beta}(t)}{2})  
\end{aligned}$
}
\item{
$\begin{aligned}[t]
&\frac{\partial^2 \hat{\Phi}_{\alpha, \beta}(t)^\dagger \hat{\Theta}_{\alpha, \beta}(t)}{\partial t^2} = \delta_{\alpha \in \{x,y\}}
\delta_{\beta \in \{x,y\}}(4W^k_m)^2 \hat{\Phi}_{\alpha, \beta}(t)^\dagger \hat{\Theta}_{\alpha, \beta}(t)  
\end{aligned}$
}
\end{enumerate}
\end{lemma}

\begin{proof}
$\newline$

\begin{enumerate}
\item Using products like $\hat{\Theta}_{\alpha, \beta}(t)^\dagger \hat{\Theta}_{\alpha, \beta}(t)$ ensures the unitary operator $\mathcal{F}$ cancels. The satisfaction of (1) can thereafter be verified explicitly as:\\\\
$\begin{aligned}[t]
&\frac{1}{2}(\frac{\partial^2 \hat{\Theta}_{\alpha, \beta}(t)^\dagger \hat{\Theta}_{\alpha, \beta}(t) }{\partial t^2})\nonumber \\
&= 2\delta_{\alpha \in \{x,y\}}
\delta_{\beta \in \{x,y\}}\nonumber\:\:(\sigma^\chi(v_k,0)\sigma^\gamma(h_m,0))^2\:\:\frac{\partial^2 (Sin^2(2W^k_m t)\mathcal{F^\dagger F})}{\partial t^2}\:\:\mathcal{I}_{2\times2}\nonumber \:\:\:\:\:\:\because \:\:\:\: Lemma \:\:\ref{lem:Comm_t_dep}(1) \\
& = \delta_{\alpha \in \{x,y\}}
\delta_{\beta \in \{x,y\}}(4W^k_m) \frac{\partial Sin(4W^k_m t)}{\partial t}\mathcal{I}_{2\times2}\nonumber \\
& = \delta_{\alpha \in \{x,y\}}
\delta_{\beta \in \{x,y\}}(4W^k_m)^2 (1-2Sin^2(4W^k_m t))\mathcal{I}_{2\times2} \nonumber \\
& = \delta_{\alpha \in \{x,y\}}
\delta_{\beta \in \{x,y\}}(4W^k_m)^2 (\mathcal{I}_{2\times2}-\frac{\hat{\Theta}_{\alpha, \beta}(t)^\dagger \hat{\Theta}_{\alpha, \beta}(t)}{2}) \nonumber \\
\end{aligned}$
$\newline$

\item Similarly (2) above can be verified explicitly too as follows \\
\begin{align}
 &\frac{1}{2}(\frac{\partial^2 \hat{\Phi}_{\alpha, \beta}(t)^\dagger \hat{\Theta}_{\alpha, \beta}(t)}{\partial t^2}) \label{eq:der2_phi_theta_a} \\
 &= \frac{1}{2}(\frac{\partial^2}{\partial t^2}(\hat{\Theta}_{\alpha, \beta}^\dagger(t)\hat{\Theta}_{\alpha, \beta}(t) + 2\sigma^\alpha(v_k,0)\sigma^\alpha(h_m,0)e^{-2i\mathcal{H}^\prime t}\hat{\Theta}_{\alpha, \beta}(t))) \:\:\:\:\: \because Lemma  \:\:\ref{lem:Comm_t_dep}(2) \nonumber \\
 &= \frac{1}{2}(\frac{\partial^2}{\partial t^2}(\hat{\Theta}_{\alpha, \beta}^\dagger(t)\hat{\Theta}_{\alpha, \beta}(t))) + \sigma^\alpha(v_k,0)\sigma^\alpha(h_m,0) (\frac{\partial^2}{\partial t^2}e^{-2i\mathcal{H}^\prime t}\hat{\Theta}_{\alpha, \beta}(t)) \nonumber \\
 &= \frac{1}{2}(\frac{\partial^2}{\partial t^2}(\hat{\Theta}_{\alpha, \beta}^\dagger(t)\hat{\Theta}_{\alpha, \beta}(t))) -2i\epsilon_{\alpha z\chi}\epsilon_{z\beta\gamma}\sigma^\alpha(v_k,0)\sigma^\alpha(h_m,0)\nonumber \\
 &\:\:\:\:\:\:\:\:\:\:\frac{\partial^2}{\partial t^2}(e^{-2i\mathcal{H}^\prime t}\mathcal{F}\sigma^\chi(v_k,0)\sigma^\gamma(h_m,0)Sin(2W^k_m t)) \:\:\:\:\:\:\:\:\:\:\:see \:\:Lemma \:\:\ref{lem:Comm_t_dep}(1)\:\: for\:\: \mathcal{F}\nonumber \\
 &=\frac{1}{2}(\frac{\partial^2}{\partial t^2}(\hat{\Theta}_{\alpha, \beta}^\dagger(t)\hat{\Theta}_{\alpha, \beta}(t)))- 2i\epsilon_{\alpha z\chi}\epsilon_{z\beta\gamma}\sigma^\alpha(v_k,0)\sigma^\alpha(h_m,0)\nonumber \\
 &\:\:\:\:\:\:\:\:\:\:\frac{\partial^2}{\partial t^2}(\sigma^\chi(v_k,0)\sigma^\gamma(h_m,0)Cos(W^k_m t)Sin(2W^k_m t) \nonumber + \epsilon_{z\chi\omega}\epsilon_{z\gamma\eta}\sigma^\omega(v_k,0)\sigma^\eta(h_m,0)Sin^2(2W^k_m t) ) \nonumber \\
 &=\frac{1}{2}(\frac{\partial^2}{\partial t^2}(\hat{\Theta}_{\alpha, \beta}^\dagger(t)\hat{\Theta}_{\alpha, \beta}(t)))- 2i\epsilon_{\alpha z\chi}\epsilon_{z\beta\gamma}\sigma^\alpha(v_k,0)\sigma^\alpha(h_m,0)\nonumber \\
 &\:\:\:\:\:\frac{\partial^2}{\partial t^2}(\sigma^\chi(v_k,0)\sigma^\gamma(h_m,0)Cos(W^k_m t)Sin(2W^k_m t)) \nonumber - 
 \frac{1}{2}(\frac{\partial^2}{\partial t^2}(\hat{\Theta}_{\alpha, \beta}^\dagger(t)\hat{\Theta}_{\alpha, \beta}(t)))\nonumber \\
 &=- 2i\epsilon_{\alpha z\chi}\epsilon_{z\beta\gamma}\sigma^\alpha(v_k,0)\sigma^\alpha(h_m,0)\frac{\partial^2}{\partial t^2}(\sigma^\chi(v_k,0)\sigma^\gamma(h_m,0)Cos(W^k_m t)Sin(2W^k_m t)) \nonumber \\
 &=2i \epsilon_{\alpha z\chi}\epsilon_{z\beta\gamma}\epsilon_{\alpha \chi\kappa}\epsilon_{\beta\gamma\omega}\:\:\frac{\partial^2}{\partial t^2}(\sigma^\kappa(v_k,0)\sigma^\omega(h_m,0)Cos(W^k_m t)Sin(2W^k_m t)) \nonumber \\
 &=-2i\frac{\partial^2}{\partial t^2}(\sigma^z(v_k,0)\sigma^z(h_m,0)Cos(W^k_m t)Sin(2W^k_m t)) \nonumber \\
 &=\frac{1}{2}\frac{\partial^2}{\partial t^2}(-2i\sigma^z(v_k,0)\sigma^z(h_m,0)Sin(4W^k_m t)) \label{eq:phi_t_b} \\
 &= -(4W^k_m)^2 \hat{\Phi}_{\alpha, \beta}(t)^\dagger \hat{\Theta}_{\alpha, \beta}(t) \nonumber\\\nonumber
\end{align}
where the last equality follows from  Eq.\ref{eq:der2_phi_theta_a} and Eq.\ref{eq:phi_t_b} which shows $\frac{1}{2}(\frac{\partial^2 \hat{\Phi}_{\alpha, \beta}(t)^\dagger \hat{\Theta}_{\alpha, \beta}(t)}{\partial t^2}) = \frac{1}{2}\frac{\partial^2}{\partial t^2}(-2i\sigma^z(v_k,0)\sigma^z(h_m,0)Sin(4W^k_m t))$ thereby implying $\hat{\Phi}_{\alpha, \beta}(t)^\dagger \hat{\Theta}_{\alpha, \beta}(t) \\= -2i\sigma^z(v_k,0)\sigma^z(h_m,0)Sin(4W^k_m t)$
 
\end{enumerate} 
\end{proof}

\newpage
\section{Invariants of motion-Proof of Theorem 1 in main manuscript}\label{thm_1_proof}
We are now in a position to prove the assertions of Theorem 1 in the main manuscript
$\forall \alpha, \beta \in \{x,y\}$\\

1) To establish the assertions of Theorem 1 (1) we need to prove the following primitive lemmas first for completeness

\begin{lemma}
\label{pauli_inv}
$\sigma^\alpha f(\sigma^z) (\sigma^\alpha)^\dagger = f(-\sigma^z) \:\:\:\:\: \forall \alpha \in \{x,y\}$
\begin{proof}
\begin{align}
&\sigma^\alpha \sigma^z (\sigma^\alpha)^\dagger \qquad \qquad  \nonumber\\
&= i \sigma^\alpha \epsilon_{z \alpha k} \sigma^k = i^2 \epsilon_{z \alpha k} \epsilon_{\alpha k \omega} \sigma^\omega = i^2 \delta_{\omega z} \sigma^\omega = - \sigma^z \:\:\:\:\: \because(\alpha \in \{x,y\}) \nonumber
\end{align}
\begin{align}
&\sigma^\alpha f(\sigma^z) (\sigma^\alpha)^\dagger \nonumber\\
&=\sigma^\alpha [\beta_0 + \beta_1 \sigma^z + \beta_2 (\sigma^z)^2 + \cdots ] (\sigma^\alpha)^\dagger \nonumber\\
&=\beta_0 \sigma^\alpha (\sigma^\alpha)^\dagger + \beta_1 \sigma^\alpha \sigma^z (\sigma^\alpha)^\dagger + \beta_2 [\sigma^\alpha \sigma^z (\sigma^\alpha)^\dagger]^2 + \cdots \nonumber \\
&=\beta_0 -\beta_1 \sigma^z + \beta_2 (\sigma^z)^2 - \beta_3 (\sigma^z)^3 + \cdots \nonumber \\
&= f(-\sigma^z) \nonumber
\end{align}
\end{proof}
\end{lemma}

\begin{lemma}
\label{tr_non_diag}
For a diagonal matrix A and an off-diagonal matrix B, $Tr(AB) = 0$\\
Given: $B^k_i=0\quad$ $\forall$ $\;i=k$ and 
$A^k_i=a_i \delta^k_i$
\begin{proof}
\begin{align}
&Tr(AB)=\sum_i (AB)^i_i = \sum_i A^i_k B^k_i = \sum_i a_i \delta^i_k B^k_i =\sum a_i B^i_i = 0 \nonumber
\end{align}
\end{proof}
\end{lemma}

\begin{theorem}\label{Theorem_1_proof1} (Theorem \ref{lem1}(1) in main manuscript)
For a given parameter vector $\vec{X}$, one can define $\mathcal{H}(\vec{X}, \vec{v}, \vec{h})$ (see Eq.\ref{eq: Ising_energy}) and a thermal state $\rho_{th}(\vec{X}, \vec{v}, \vec{h})$. Let us thereafter define the following OTOC with $U_1(0) = \tilde{\sigma_\alpha}= \sigma^\alpha(v_k,0)-\kappa_1\mathcal{I}$, and operator $U_2(0) = \tilde{\sigma_\beta}=\sigma^\beta(h_m,0) -\kappa_2\mathcal{I}$ and the generator $H_{\textit{otoc}} = \mathcal{H}(\vec{X}, \vec{v}, \vec{h})$ (in Eq.\ref{eq:C_t_gen} in main manuscript)
$\:\:\forall \:\: \{\alpha, \beta \} \in \{x,y\}$. 

\begin{align}
C_{\sigma^\alpha, \sigma^\beta}(\kappa_1, \kappa_2,\vec{X}, t) &= \langle \tilde{\sigma}^\alpha(v_k,0) \tilde{\sigma}^\beta(h_m,t)\tilde{\sigma}^\alpha(v_k,0)\tilde{\sigma}^\beta(h_m,t)\rangle\label{eq:C_t_rbm_s1}
\end{align}
Note that $\{\kappa_1, \kappa_2\} \in \mathcal{C}^2$ are arbitrary user-defined mean translations. Also $\langle \cdot \rangle$ indicates averaging over the thermal state $\rho_{th}(\vec{X}, \vec{v}, \vec{h})$ which activates the $\vec{X}$ dependence. Using \ref{eq:C_t_rbm}, one can then make the following statements:
For $(\kappa_1, \kappa_2)$ $\in$ $\mathcal{C}^2$ \\
$$C_{\sigma^\alpha,\sigma^\beta} (\kappa_1,\kappa_2,\vec{X},t)=
\langle \tilde{\sigma}^\alpha(v_k,0) \tilde{\sigma}^\beta(h_m,t) \tilde{\sigma}^\alpha(v_k,0) \tilde{\sigma}^\beta(h_m,t) \rangle = C_{\sigma^\alpha,\sigma^\beta} (0,0,\vec{X},t) + |\kappa_1|^2 |\kappa_2|^2 + |\kappa_2|^2 + |\kappa_1|^2$$
where $\alpha$, $\beta$ $\in \{x,y\}$

\begin{proof}
Let us define $\tilde{U}_1(0) = U_1(0) - \kappa_1 \mathcal{I}$ and $\tilde{U}_2(0)= U_2(0) - \kappa_2 \mathcal{I}$ where $U_1(0)$ and $U_2(0)$ are two unitary operators as defined in the main manuscript

\begin{align} \label{scaled_OTOC}
&\langle \tilde{U}_1(0)^\dagger \tilde{U}_2(t)^\dagger \tilde{U}_1(0) \tilde{U}_2(t) \rangle \nonumber\\
&=\langle U_1(0)^\dagger U_2(t)^\dagger U_1(0) U_2(t) \rangle + |\kappa_1|^2 |\kappa_2|^2 \nonumber \\
&\quad -\kappa_2 \langle U_1(0)^\dagger U_2(t)^\dagger U_1(0) \rangle  
-\kappa_1 \langle U_1(0)^\dagger U_2(t)^\dagger U_2(t) \rangle \nonumber \\
&\quad -\kappa_2^* \langle U_1(0)^\dagger U_1(0) U_2(t) \rangle  
-\kappa_1^* \langle U_2(t)^\dagger U_1(0) U_2(t) \rangle \nonumber \\
&\quad +|\kappa_2|^2 \langle U_1^\dagger(0) U_1(0) \rangle + |\kappa_1|^2 \langle U_2^\dagger(t) U_2(t) \rangle +\kappa_1^* \kappa_2 \langle U_2^\dagger(t) U_1(0) \rangle \nonumber \\
&\quad +\kappa_1 \kappa_2^* \langle U_1^\dagger(0) U_2(t) \rangle  +\kappa_1^* \kappa_2^* \langle U_1(0) U_2(t) \rangle +\kappa_1 \kappa_2 \langle U_1^\dagger(0) U_2^\dagger(t) \rangle \nonumber \\
&\quad -\kappa_1 |\kappa_2|^2 \langle U_1(0)^\dagger \rangle -|\kappa_1|^2 \kappa_2 \langle U_2(t)^\dagger \rangle -\kappa_1^* |\kappa_2|^2 \langle U_1(0) \rangle -|\kappa_1|^2 \kappa_2^\dagger \langle U_2(t) \rangle
\end{align}

Let us now substitute $U_1(0) = \sigma^\alpha(v_k,0)$ and $U_2(t) = \sigma^\beta(h_m,t)$ in Eq. \ref{scaled_OTOC}

By definition,
\begin{align}\label{key_term}
\langle \sigma^\alpha(v_k,0) \sigma^\beta(h_m,t)  \sigma^\alpha(v_k,0) \sigma^\beta(h_m,t) \rangle &= C_{\sigma^\alpha,\sigma^\beta} (0,0,\vec{X},t) 
\end{align}
Also one can show the following:

\begin{enumerate}[label=\alph*)]

\item 
\begin{align}\label{lemma_term1}
&\langle \sigma^\alpha(v_k,0) \sigma^\beta(h_m,t) \sigma^\alpha(v_k,0) \rangle \nonumber\\
&= Tr(\rho_{th}(\vec{X}, v_k, h_m) \sigma^\alpha(v_k,0) \sigma^\beta(h_m,t) \sigma^\alpha(v_k,0) ) \nonumber\\
&= Tr(\sigma^\alpha(v_k,0) \rho_{th}(\vec{X}, v_k, h_m) \sigma^\alpha(v_k,0) \sigma^\beta(h_m,t)) \nonumber\\
&= Tr(\rho_{th}(\vec{X}, -v_k, h_m) e^{i\mathcal{H}t} \sigma^\beta(h_m,0) e^{-i\mathcal{H}t}) \qquad \qquad See \;\; Lemma \;\; \ref{pauli_inv} \:\: and \:\:Lemma \:\:\ref{lem:sigma_t_dep}\nonumber\\
&= Tr(e^{-i\mathcal{H}t} \rho_{th}(\vec{X}, -v_k, h_m) e^{i\mathcal{H}t} \sigma^\beta(h_m,0)) \nonumber\\
&= Tr(\rho_{th}(\vec{X}, -v_k, h_m) \sigma^\beta(h_m,0)) \qquad \qquad \because \text{$e^{i\mathcal{H}t}$ \& $\rho$ are diagonal} \nonumber\\
&=0 \qquad \qquad \qquad See \;\; Lemma \;\; \ref{tr_non_diag}
\end{align}

\item 
\begin{align}\label{lemma_term2}
&\langle \sigma^\beta(h_m,t) \sigma^\alpha(v_k,0) \sigma^\beta(h_m,t) \rangle \nonumber\\
&= Tr(\rho_{th}(\vec{X}, v_k, h_m) \sigma^\beta(h_m,t) \sigma^\alpha(v_k,0) \sigma^\beta(h_m,t)) \nonumber\\
&= Tr(\sigma^\beta(h_m,t) \rho_{th}(\vec{X}, v_k, h_m) \sigma^\beta(h_m,t) \sigma^\alpha(v_k,0)) \nonumber\\
&= Tr(e^{i\mathcal{H}t} \sigma^\beta(h_m,0) e^{-i\mathcal{H}t} \rho_{th}(\vec{X}, v_k, h_m) e^{i\mathcal{H}t} \sigma^\beta(h_m,0) e^{-i\mathcal{H}t} \sigma^\alpha(v_k,0)) \nonumber\\
&= Tr(e^{i\mathcal{H}t} \sigma^\beta(h_m,0) \rho_{th}(\vec{X}, v_k, h_m) \sigma^\beta(h_m,0) e^{-i\mathcal{H}t} \sigma^\alpha(v_k,0)) \qquad \because \text{$e^{i\mathcal{H}t}$ \& $\rho$ are diagonal} \nonumber\\ 
&= Tr(e^{i\mathcal{H}t} \rho_{th}(\vec{X}, v_k, -h_m) e^{-i\mathcal{H}t} \sigma^\alpha(v_k,0)) \qquad \qquad See \;\; Lemma \;\; \ref{pauli_inv} \nonumber\\ 
&= Tr(\rho_{th}(\vec{X}, v_k, -h_m) \sigma^\alpha(v_k,0)) \qquad \qquad \because \text{$e^{i\mathcal{H}t}$ \& $\rho$ are diagonal} \nonumber\\ 
&=0 \qquad \qquad \qquad See \;\; Lemma \;\; \ref{tr_non_diag}
\end{align}

\item 
\begin{align}\label{lemma_term3}
&\langle \sigma^\alpha(v_k,0) \sigma^\beta(h_m,t) \sigma^\beta(h_m,t) \rangle \nonumber \\
&= \langle \sigma^\alpha(v_k,0) \rangle \nonumber\\
&= Tr(\rho_{th}(\vec{X}, v_k, h_m) \sigma^\alpha(v_k,0) \nonumber\\
&=0 \qquad \qquad See \;\; Lemma \;\; \ref{tr_non_diag}
\end{align}

\item
\begin{align}\label{lemma_term4}
&\langle \sigma^\alpha(v_k,0) \sigma^\alpha(v_k,0) \sigma^\beta(h_m,t) \rangle = \langle \sigma^\beta(h_m,t) \rangle \nonumber\\
&= Tr(\rho_{th}(\vec{X}, v_k, h_m) \sigma^\beta(h_m,t)) \nonumber\\
&=0 \qquad \qquad See \;\; Lemma \;\; \ref{tr_non_diag}
\end{align}

\item
\begin{align}\label{lemma_term5}
&\langle \sigma^\beta(h_m,t) \sigma^\alpha(v_k,0) \rangle = Tr(\rho_{th}(\vec{X}, v_k, h_m) \sigma^\beta(h_m,t) \sigma^\alpha(v_k,0)) \nonumber\\
&= Tr(\rho_{th}(\vec{X}, v_k, h_m) e^{2i\mathcal{H}'t}\sigma^\beta(h_m,0) \sigma^\alpha(v_k,0)) \qquad \qquad \text{See Lemma \ref{lem:sigma_t_dep}} \nonumber\\
&=0 \qquad \qquad See \;\; Lemma \;\; \ref{tr_non_diag} 
\end{align}

Here $\rho e^{2i\mathcal{H}'t}$ is diagonal, $\sigma^\beta(h_m,0) \sigma^\alpha(v_k,0)$ is off-diagonal\\

Similarly we can show $\langle \sigma^\alpha(v_k,0) \sigma^\beta(h_m,t)\rangle = 0\\$ 

$\langle \sigma^\alpha(v_k,0)\rangle = \langle \sigma^\beta(h_m,t)\rangle = 0 \qquad \qquad$ See Lemma \ref{tr_non_diag}

Substituting the results of Eq.\ref{key_term}, \ref{lemma_term1},\ref{lemma_term2},\ref{lemma_term3},\ref{lemma_term4},\ref{lemma_term5}  in Eq.\ref{scaled_OTOC} establishes the claim of the theorem
\end{enumerate}
\end{proof}
\end{theorem}

2) Now we establish the assertions of Theorem \ref{lem1}(2) in main manuscript

\begin{theorem}\label{Theorem_1_proof} (Theorem \ref{lem1}(2) in main manuscript)
For a given parameter vector $\vec{X}$, one can define $\mathcal{H}(\vec{X}, \vec{v}, \vec{h})$ (see Eq.\ref{eq: Ising_energy}) and a thermal state $\rho_{th}(\vec{X}, \vec{v}, \vec{h})$. Let us thereafter define the following OTOC with $U_1(0) = \tilde{\sigma_\alpha}= \sigma^\alpha(v_k,0)-\kappa_1\mathcal{I}$, and operator $U_2(0) = \tilde{\sigma_\beta}=\sigma^\beta(h_m,0) -\kappa_2\mathcal{I}$ and the generator $H_{\textit{otoc}} = \mathcal{H}(\vec{X}, \vec{v}, \vec{h})$ (in Eq.\ref{eq:C_t_gen} in main manuscript)
$\:\:\forall \:\: \{\alpha, \beta \} \in \{x,y\}$. 

\begin{align}
C_{\sigma^\alpha, \sigma^\beta}(\kappa_1, \kappa_2,\vec{X}, t) &= \langle \tilde{\sigma}^\alpha(v_k,0) \tilde{\sigma}^\beta(h_m,t)\tilde{\sigma}^\alpha(v_k,0)\tilde{\sigma}^\beta(h_m,t)\rangle\label{eq:C_t_rbm_s2}
\end{align}
Note that $\{\kappa_1, \kappa_2\} \in \mathcal{C}^2$ are arbitrary user-defined mean translations. Also $\langle \cdot \rangle$ indicates averaging over the thermal state $\rho_{th}(\vec{X}, \vec{v}, \vec{h})$ which activates the $\vec{X}$ dependence. Using \ref{eq:C_t_rbm}, one can then show the following 
invariants of motion exists for $C_{\sigma^\alpha, \sigma^\beta}(0,0,\vec{X}, t)$ :
\begin{enumerate}
     {\item 
     $\begin{aligned}[t]
     &I_1=-2\dot{\xi}_{\sigma^\alpha,\sigma^\beta}(\vec{X},\tau)Cos(\tau) \nonumber - 2\xi_{\sigma^\alpha,\sigma^\beta}(\vec{X},\tau)Sin(\tau)
     \end{aligned}$
     }
     {\item 
     $\begin{aligned}[t]
     &I_2= -2\dot{\xi}_{\sigma^\alpha,\sigma^\beta}(\vec{X},\tau)Sin(\tau) \nonumber + 2\xi_{\sigma^\alpha,\sigma^\beta}(\vec{X},\tau)Cos(\tau)
     \end{aligned}$
     }
\end{enumerate}
where $\xi_{\sigma^\alpha,\sigma^\beta}(\vec{X},\tau)$ can either be the real or the imaginary part of $(C_{\sigma^\alpha,\sigma^\beta}(0,0,\vec{X},\tau))$ and  $\dot{\dottedsquare}$ is $\frac{\partial \dottedsquare}{\partial \tau} $ with $\tau=4W^k_m t$\\\\
\end{theorem}

\underline{\textbf{Real part}} \\
For the real part of $C_{\sigma^\alpha,\sigma^\beta}(0,0,\vec{X},\tau)$ we substitute in Theorem \ref{Theorem_1_proof} (1)

\begin{align}
\xi_{\sigma^\alpha,\sigma^\beta}(\vec{X},\tau) = Re(C_{\sigma^\alpha,\sigma^\beta}(0,0,\vec{X},\tau)) =C_{r,\sigma^\alpha,\sigma^\beta}(0,0,\vec{X},\tau)
\end{align}

Thus we have
\begin{enumerate}[label=\alph*)]
\item
{
\begin{align}
\mathunderline{black} {I_{r1} = -2\dot{C}_{r,\sigma^\alpha,\sigma^\beta}(0,0 \vec{X},\tau)Cos(\tau)- 2C_{r,\sigma^\alpha,\sigma^\beta}(0,0,\vec{X},\tau)Sin(\tau)} \label{eq:IV1} 
\end{align}

\begin{proof}
$\newline$

To prove $I_{r1}$ as an invariant, we have to show that $\dot{I_{r1}} = 0$. Using the definition of $C_{r,\sigma^\alpha, \sigma^\beta}(0,0,\vec{X},\tau)$ as 
\begin{align}
C_{r\sigma^\alpha, \sigma^\beta}(0,0,\vec{X},\tau) = \langle \mathcal{I}_{2 \times 2}- \frac{1}{2}\hat{\Theta}^\dagger_{\alpha, \beta}(\tau)\hat{\Theta}_{\alpha, \beta}(\tau) \rangle_{\rho_{th}} 
\end{align}
where $\hat{\Theta}_{\alpha, \beta}(t)=[\sigma^\alpha(v_k,0), \sigma^\beta(h_m, t)]$ (see Lemma \ref{lem:Comm_t_dep}) one can show the following\\

$\begin{aligned}[t]
\dot{I_{r1}} &= \langle\frac{\partial^2 \hat{\Theta}^\dagger_{\alpha, \beta}(\tau)\hat{\Theta}_{\alpha, \beta}(\tau)}{\partial \tau^2}Cos(\tau) + \frac{\partial \hat{\Theta}^\dagger_{\alpha, \beta}(\tau)\hat{\Theta}_{\alpha, \beta}(\tau)}{\partial \tau}Sin(\tau) - \frac{\partial \hat{\Theta}^\dagger_{\alpha, \beta}(\tau)\hat{\Theta}_{\alpha, \beta}(\tau)}{\partial t}Sin(\tau) - 2(\mathcal{I}_{2\times 2}- \frac{1}{2}\hat{\Theta}^\dagger_{\alpha, \beta}(\tau)\hat{\Theta}_{\alpha, \beta}(\tau))Cos(\tau)\nonumber\rangle_{\rho_{th}} \\
& = \langle\frac{\partial^2 \hat{\Theta}^\dagger_{\alpha, \beta}(\tau)\hat{\Theta}_{\alpha, \beta}(\tau)}{\partial \tau^2}Cos(\tau) \nonumber - 2(\mathcal{I}_{2 \times 2}- \frac{1}{2}\hat{\Theta}^\dagger_{\alpha, \beta}(t)\hat{\Theta}_{\alpha, \beta}(\tau))Cos(\tau)\rangle_{\rho_{th}} \nonumber \\
& = 2\langle\frac{1}{2}\frac{\partial^2 \hat{\Theta}^\dagger_{\alpha, \beta}(\tau)\hat{\Theta}_{\alpha, \beta}(\tau)}{\partial \tau^2}- (\mathcal{I}_{2 \times 2}- \frac{1}{2}\hat{\Theta}^\dagger_{\alpha, \beta}(t)\hat{\Theta}_{\alpha, \beta}(\tau))\rangle_{\rho_{th}} Cos(\tau) \nonumber \\
& = \frac{2}{(4W^k_m)^2}\langle\frac{1}{2}\frac{\partial^2 \hat{\Theta}^\dagger_{\alpha, \beta}(t)\hat{\Theta}_{\alpha, \beta}(t)}{\partial t^2} - (4W^k_m)^2(\mathcal{I}_{2 \times 2}- \frac{1}{2}\hat{\Theta}^\dagger_{\alpha, \beta}(t)\hat{\Theta}_{\alpha, \beta}(t))\rangle_{\rho_{th}} Cos(4W^k_m t) \nonumber \:\:\:\:\:\:\:\:\:\because \:\:\tau\: = \:4W^k_m t\\
&= 0  \:\:\:\:\:\:\:\:\:\: \because \:\: Lemma\:\:\:\: \ref{lemma_Diff_eqn}(1) \:\: \forall \:\: \alpha, \beta \in \{x,y\}
\end{aligned}$
\end{proof}}
$\newline$
\item
{For this one we have to substitute 
in in Theorem \ref{Theorem_1_proof} (2)
\begin{align}
\xi_{\sigma^\alpha,\sigma^\beta}(\vec{X},\tau) = Re(C_{\sigma^\alpha,\sigma^\beta}(0,0,\vec{X},\tau)) =C_{r,\sigma^\alpha,\sigma^\beta}(0,0,\vec{X},\tau)
\end{align}
After substitution, we have 
\begin{align}
\mathunderline{black} {I_{r2} = -2\dot{C}_{r,\sigma^\alpha,\sigma^\beta}(0,0 \vec{X},\tau)Sin(\tau) + 2C_{r,\sigma^\alpha,\sigma^\beta}(0,0,\vec{X},\tau)Cos(\tau)} \label{eq:IV2}
\end{align}

\begin{proof}
$\newline$
As before to prove invariance one has to show $\dot{I_{r2}}=0$. Similar to Eq.\ref{eq:IV1}, Eq.\ref{eq:IV2} can either be verified explicitly or by noting $I_{r2}=I_{r1}(\tau \to \tau - \frac{\pi}{2})$. The invariance of $I_{r2}$ thereafter follows from the invariance of $I_{r1}$ proven above 
\end{proof}}
\end{enumerate}
$\newline$
$\newline$

\newpage
\underline{\textbf{Imag part}} \\
For the imaginary part of $C_{\sigma^\alpha,\sigma^\beta}(0,0,\vec{X},\tau)$ we substitute in Theorem \ref{Theorem_1_proof} (1)
\begin{align}
\xi_{\sigma^\alpha,\sigma^\beta}(\vec{X},\tau) = Im(C_{\sigma^\alpha,\sigma^\beta}(0,0,\vec{X},\tau)) =C_{i,\sigma^\alpha,\sigma^\beta}(0,0,\vec{X},\tau)
\end{align}

Thus we have now \\
\begin{enumerate}[label=\alph*)]
\item{
\begin{align}
\mathunderline{black} {I_{i1}= -2\dot{C}_{i,\sigma^\alpha,\sigma^\beta}(0,0 \vec{X},\tau)Cos(\tau) - 2C_{i,\sigma^\alpha,\sigma^\beta}(0,0,\vec{X},\tau)Sin(\tau)}\label{eq:IV3}
\end{align}

\begin{proof}
To prove $I_{i1}$ as an invariant, we have to show that $\dot{I_{i1}} = 0$. Using the definition of $C_{i\sigma^\alpha, \sigma^\beta}(0,0,\vec{X},\tau)$ as 
\begin{align}
C_{i\sigma^\alpha, \sigma^\beta}(0,0,\vec{X},\tau) = -\frac{i}{2}\langle \hat{\Phi}^\dagger_{\alpha, \beta}(t)\hat{\Theta}_{\alpha, \beta}(\tau) \rangle_{\rho_{th}} 
\end{align}
where $\hat{\Phi}_{\alpha, \beta}(t)=\{\sigma^\alpha(v_k,0), \sigma^\beta(h_m, t)\}_+$ (see Lemma \ref{lem:Comm_t_dep}), one can show the following\\

$\begin{aligned}[t]
\dot{I_{i1}} &= i\langle\frac{\partial^2 \hat{\Phi}^\dagger_{\alpha, \beta}(\tau)\hat{\Theta}_{\alpha, \beta}(\tau)}{\partial \tau^2}Cos(\tau) - \frac{\partial \hat{\Phi}^\dagger_{\alpha, \beta}(\tau)\hat{\Theta}_{\alpha, \beta}(\tau)}{\partial \tau}Sin(\tau) + \frac{\partial \hat{\Phi}^\dagger_{\alpha, \beta}(\tau)\hat{\Theta}_{\alpha, \beta}(t)}{\partial \tau}Sin(\tau) +\hat{\Phi}^\dagger_{\alpha, \beta}(\tau)\hat{\Theta}_{\alpha, \beta}(\tau)Cos(\tau)\nonumber\rangle_{\rho_{th}} \\
& = i\langle\frac{\partial^2 \hat{\Phi}^\dagger_{\alpha, \beta}(\tau)\hat{\Theta}_{\alpha, \beta}(\tau)}{\partial \tau^2}Cos(\tau) + \hat{\Phi}^\dagger_{\alpha, \beta}(\tau)\hat{\Theta}_{\alpha, \beta}(\tau)Cos(\tau)\nonumber\rangle_{\rho_{th}}\nonumber \\
& = \frac{i}{(4W^k_m)^2}\langle\frac{\partial^2 \hat{\Phi}^\dagger_{\alpha, \beta}(t)\hat{\Theta}_{\alpha, \beta}(t)}{\partial t^2}Cos(4W^k_m t) \nonumber +(4W^k_m)^2\hat{\Phi}^\dagger_{\alpha, \beta}(t)\hat{\Theta}_{\alpha, \beta}(t)Cos(4W^k_mt)\nonumber\rangle_{\rho_{th}} \:\:\:\:\:\because \:\:\tau\: = \:4W^k_m t \nonumber \\
&= 0  \:\:\:\:\:\:\: \because \:\: Lemma \:\:\:\:\ref{lemma_Diff_eqn}(2) \:\: \forall \:\: \alpha, \beta\in \{x,y\}
\end{aligned}$
\end{proof}}

\item{ For this one we substitute in Theorem \ref{Theorem_1_proof} (2)
\begin{align}
\xi_{\sigma^\alpha,\sigma^\beta}(\vec{X},\tau) = Im(C_{\sigma^\alpha,\sigma^\beta}(0,0,\vec{X},\tau)) =C_{i,\sigma^\alpha,\sigma^\beta}(0,0,\vec{X},\tau)
\end{align}
Thus after substitution we have 
\begin{align}
\mathunderline{black} {I_{i2} = -2\dot{C}_{i,\sigma^\alpha,\sigma^\beta}(0,0 \vec{X},\tau)Sin(\tau) + 2C_{i,\sigma^\alpha,\sigma^\beta}(0,0,\vec{X},\tau)Cos(\tau)}  \label{eq:IV4}
\end{align}
\begin{proof}
$\newline$
Similar to Eq.\ref{eq:IV3}, Eq.\ref{eq:IV4} can either be verified explicitly or by noting $I_{i2}=I_{i1}(\tau \to \tau - \frac{\pi}{2})$. The invariance of $I_{i2}$ thereafter follows from the invariance of $I_{i1}$ proven above 
\end{proof}}

\end{enumerate}
\newpage

Using the definitions for $\Theta_{\alpha, \beta, \tau} = [\sigma^\alpha(v_k, 0), \sigma^\beta(h_m,t)]$ and  $\Phi_{\alpha, \beta, \tau} = \{\sigma^\alpha(v_k, 0), \sigma^\beta(h_m,t)\}_+$ as given in Lemma \ref{lem:Comm_t_dep}, one can enlist the initial conditions satisfied by $C_{r\sigma^\alpha, \sigma^\beta}(0,0,\vec{X},\tau)$ and $C_{r\sigma^\alpha, \sigma^\beta}(0,0,\vec{X},\tau)$ as follows:

\begin{align}
C_{r\sigma^\alpha, \sigma^\beta}(0,0,\vec{X},0) &= \langle \mathcal{I}_{2 \times 2}- \frac{1}{2}\hat{\Theta}^\dagger_{\alpha, \beta}(0)\hat{\Theta}_{\alpha, \beta}(0) \rangle_{\rho_{th}} = 1 \label{eq:C_r_t0}\\
&(\because \:\:\:\hat{\Theta}_{\alpha, \beta}(0) \:\:\:= 0)\nonumber\\
\dot{C}_{r\sigma^\alpha, \sigma^\beta}(0,0,\vec{X},0) &= -\frac{1}{2}\langle\frac{\partial}{\partial \tau}\hat{\Theta}^\dagger_{\alpha, \beta}(\tau)\hat{\Theta}_{\alpha, \beta}(\tau) \rangle_{\rho_{th}}|_0 = 0 \label{eq:Cdot_r_t0}\\
&(\because \:\:\:\hat{\Theta}_{\alpha, \beta}(0) \:\:\:= \:\hat{\Theta}^\dagger_{\alpha, \beta}(0)=0)\nonumber\\
C_{i\sigma^\alpha, \sigma^\beta}(0,0,\vec{X},0) &= -\frac{i}{2}\langle \hat{\Phi}^\dagger_{\alpha, \beta}(0)\hat{\Theta}_{\alpha, \beta}(0) \rangle_{\rho_{th}} = 0 \label{eq:C_i_t0}\\
&(\because \:\:\:\hat{\Theta}_{\alpha, \beta}(0) \:\:\:= 0)\nonumber\\
\dot{C}_{i\sigma^\alpha, \sigma^\beta}(0,0,\vec{X},0) &= -\frac{i}{2}\langle\frac{\partial}{\partial \tau}\hat{\Phi}^\dagger_{\alpha, \beta}(\tau)\hat{\Theta}_{\alpha, \beta}(\tau) \rangle_{\rho_{th}}|_0 \nonumber \\
&= \langle\frac{\partial}{\partial \tau} \sigma^z(v_k,0)\sigma^z(h_m,0)Sin(4W^k_m t) \rangle_{\rho_{th}}|_0 \nonumber \\
&= \langle \sigma^z(v_k,0)\sigma^z(h_m,0) \rangle_{\rho_{th}} \label{eq:Cdot_i_t0}\\
&(\:\:see \:\:\: Lemma \:\:\ref{lem:Comm_t_dep}(2)\:) \nonumber
\end{align}
Using Eq.\ref{eq:C_r_t0}, \ref{eq:Cdot_r_t0}, \ref{eq:C_i_t0}, \ref{eq:Cdot_i_t0} one can establish the following corollary to Theorem 1 in the main manuscript

\begin{corollary}\label{soln_C_t}
The following statements are true
\begin{itemize}
\item {
$\begin{aligned}[t]
C_{r\sigma^\alpha, \sigma^\beta}(0,0,\vec{X},\tau) &= Cos(4W^k_m t) 
\end{aligned}$
$\newline$
\begin{proof}
$\newline$
$\begin{aligned}[t]
C_{r\sigma^\alpha, \sigma^\beta}(0,0,\vec{X},\tau) &= \frac{I_{r2}}{2}Cos(\tau) + \frac{I_{r1}}{2} Sin(\tau)\:\:\:\:\: (\:\:\because\:\: Eq.(\ref{eq:IV1}\times Sin(\tau) - Eq.(\ref{eq:IV2}\times Cos(\tau))\nonumber \\
&= Cos(\tau)   \:\:\:\:\:\:\:\:\:  (see \: Eq.\ref{eq:C_r_t0}, \ref{eq:Cdot_r_t0})\nonumber \\
&= Cos(4W^k_m t)  \:\:\:\:\:\because\:\: \tau = 4W^k_m t
\end{aligned}$
\end{proof}}

\item {
$\begin{aligned}[t]
C_{i\sigma^\alpha, \sigma^\beta}(0,0,\vec{X},\tau) &= \langle \sigma^z(v_k,0) \sigma^z(h_m,0) \rangle Sin(4W^k_m t) 
\end{aligned}$
$\newline$
\begin{proof}
$\newline$
$\begin{aligned}[t]
C_{i\sigma^\alpha, \sigma^\beta}(0,0,\vec{X},\tau) &= \frac{I_{i2}}{2}Cos(\tau) + \frac{I_{i1}}{2} Sin(\tau) \:\:\:\:\:(\:\:\because\:\: Eq.(\ref{eq:IV3}\times Sin(\tau) - Eq.(\ref{eq:IV4}\times Cos(\tau))\nonumber \\
&= \langle \sigma^z(v_k,0) \sigma^z(h_m,0) \rangle Sin(\tau)  \:\:\:\:\:\:\:\:\:  (see \:Eq.\ref{eq:C_i_t0}, \ref{eq:Cdot_i_t0})\nonumber \\
&= \langle \sigma^z(v_k,0) \sigma^z(h_m,0) \rangle Sin(4W^k_m t) \:\:\:\:\:(\because\:\: \tau = 4W^k_m t)
\end{aligned}$
\end{proof}}
\end{itemize}
\end{corollary}

\begin{corollary}\label{Corr_2}
In addition to the above invariants, following are also the invariants of motion for $C_{r\sigma^\alpha, \sigma^\beta}(0,0,\vec{X},\tau)$ and $C_{i\sigma^\alpha, \sigma^\beta}(0,0,\vec{X},\tau)$

\begin{itemize}
\item {
$\begin{aligned}[t]
 &(\dot{C}_{r\sigma^\alpha, \sigma^\beta}(0,0,\vec{X},\tau)^2 - C_{r\sigma^\alpha, \sigma^\beta}(0,0,\vec{X},\tau)^2)Sin(2\tau) -2\dot{C}_{r\sigma^\alpha, \sigma^\beta}(0,0,\vec{X},\tau)C_{r\sigma^\alpha, \sigma^\beta}(0,0,\vec{X},\tau)Cos(2\tau)
\end{aligned}$
$\newline$
\begin{proof}
$\newline$
Can be verified through explicit evaluation using the solution $C_{r\sigma^\alpha, \sigma^\beta}(0,0,\vec{X},\tau)$ in Corollary \ref{soln_C_t}, or by combining the invariants $I_{r1}, I_{r1}$ in Eq.\ref{eq:IV1} and Eq.\ref{eq:IV2} as
$\frac{1}{2} I_{r1}I_{r2}$. This invariant is plotted  in Fig. \ref{Fig:Inv_OTOC_real} for the real part. Exactly similar profile for the invariant exists for the imaginary part too.
\end{proof}}

\item {
$\begin{aligned}[t]
 &(\dot{C}_{r\sigma^\alpha, \sigma^\beta}(0,0,\vec{X},\tau)^2 + C_{r\sigma^\alpha, \sigma^\beta}(0,0,\vec{X},\tau)^2) \nonumber  
\end{aligned}$
\begin{proof}
$\newline$
Can be verified through explicit evaluation using the solution $C_{r\sigma^\alpha, \sigma^\beta}(0,0,\vec{X},\tau)$ in Corollary \ref{soln_C_t}, or by combining the invariants $I_{r1}, I_{r1}$ in Eq.\ref{eq:IV1} and Eq.\ref{eq:IV2} as
$\frac{1}{4}(I_{r1}^2+I_{r2}^2)$. This invariant is plotted in Fig.\ref{Fig:Inv_OTOC}(a) in the main manuscript for the real and imaginary part. The profile for the invariant is different in two cases with the imaginary part being sensitive to the training process of the network $G$ unlike the real part and hence can be used to deliver meaningful insight about the learning dynamics. 
\end{proof}}
\end{itemize}
Exactly similar invariants can be obtained for $C_{r\sigma^\alpha, \sigma^\beta}(0,0,\vec{X},\tau)$ by substituting $C_{r\sigma^\alpha, \sigma^\beta}(0,0,\vec{X},\tau) \to C_{i\sigma^\alpha, \sigma^\beta}(0,0,\vec{X},\tau)$ in the above expressions.
\end{corollary}

\begin{figure}[!htb]
\centering \includegraphics[width=0.40\textwidth]{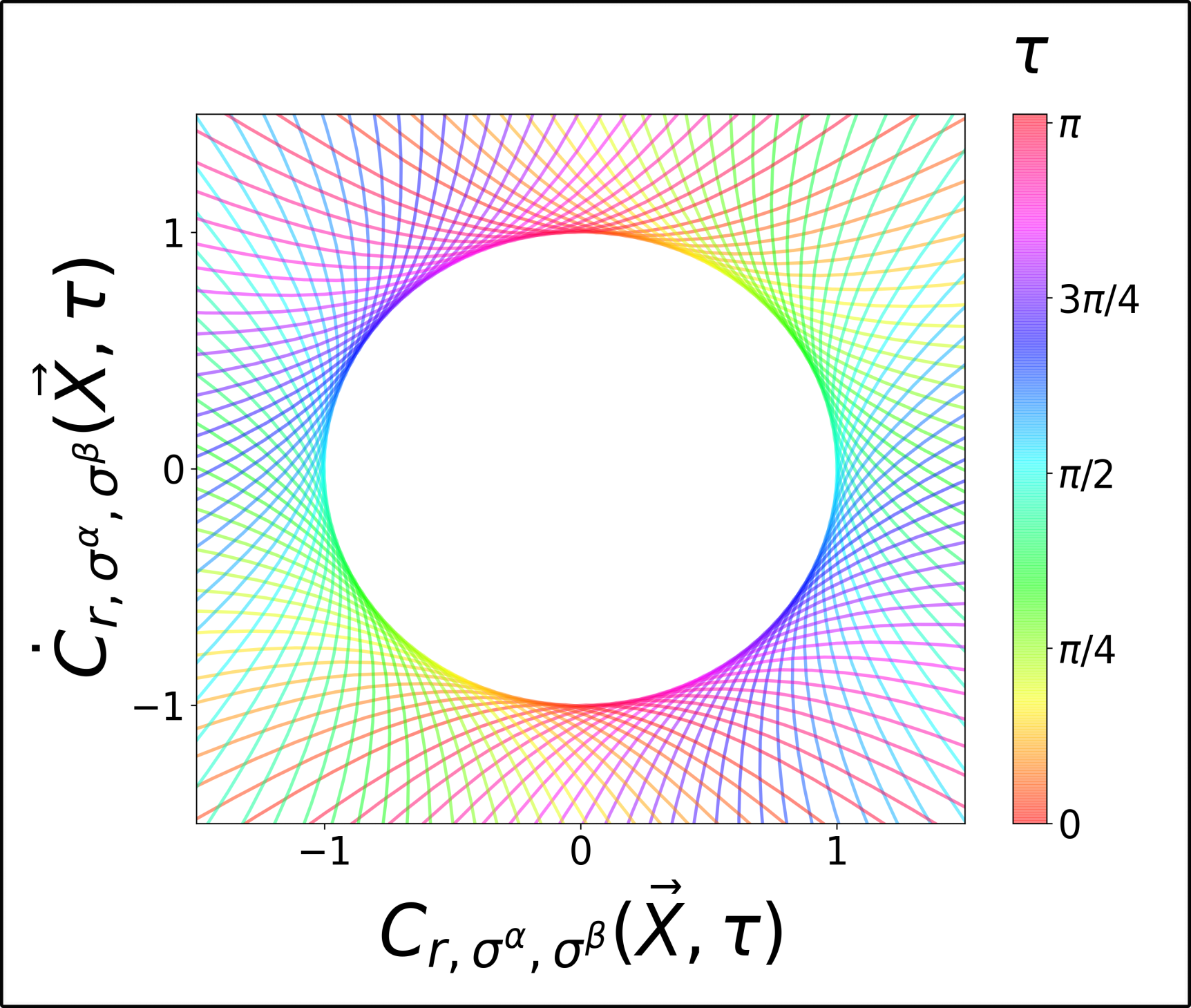}
\caption{The profile for the real part of the invariant $\frac{I_1I_2}{2}$ (see Section \ref{thm_1_proof}, Corollary \ref{Corr_2}), where  $C_{r, \sigma^\alpha, \sigma^\beta}(\vec{X}, \tau) = Re(C_{\sigma^\alpha, \sigma^\beta}(0,0,\vec{X}, \tau)) = \xi_{\sigma^\alpha, \sigma^\beta}(\vec{X},\tau)$ and $\alpha = \beta = x$ is substituted in Theorem 1(2) in main manuscript for a specific vis-hid neuron pair $(k,m)$ (See Eq.\ref{eq:C_t_rbm} in main manuscript).  Subscript `r' denotes the real part in the plot.
The quantity has been evaluated at $(\kappa_1=0,\kappa_2=0)$, hence the  explicit dependence on $(\kappa_1,\kappa_2)$ has been dropped for notational brevity. 
Each hyperbolic curve is the loci of points with the same fixed value of the invariant and a particular fixed value of $\tau = 4W^k_m t$ as the dependant variable (see colorbar). The loci of all points touching several hyperbolic curves with a fixed value of the invariant alone but wherein $\tau$ is continuously changed $\in$  $[0, \pi]$  forms a circle at the center. Since $\tau = 4 W^k_m t$, this change can be administered by changing real-time $(t)$ or the parameter vector $\vec{X}$ during training which affects $W^k_m$. Exactly similar profile exists for the imaginary part too where $C_{i, \sigma^\alpha, \sigma^\beta}(\vec{X}, \tau) = Im(C_{\sigma^\alpha, \sigma^\beta}(0,0,\vec{X}, \tau)) = \xi_{\sigma^\alpha, \sigma^\beta}(\vec{X},\tau)$ is substituted}
\label{Fig:Inv_OTOC_real}
\end{figure}

\begin{figure}[!htb]
\centering \includegraphics[width=0.40\textwidth]{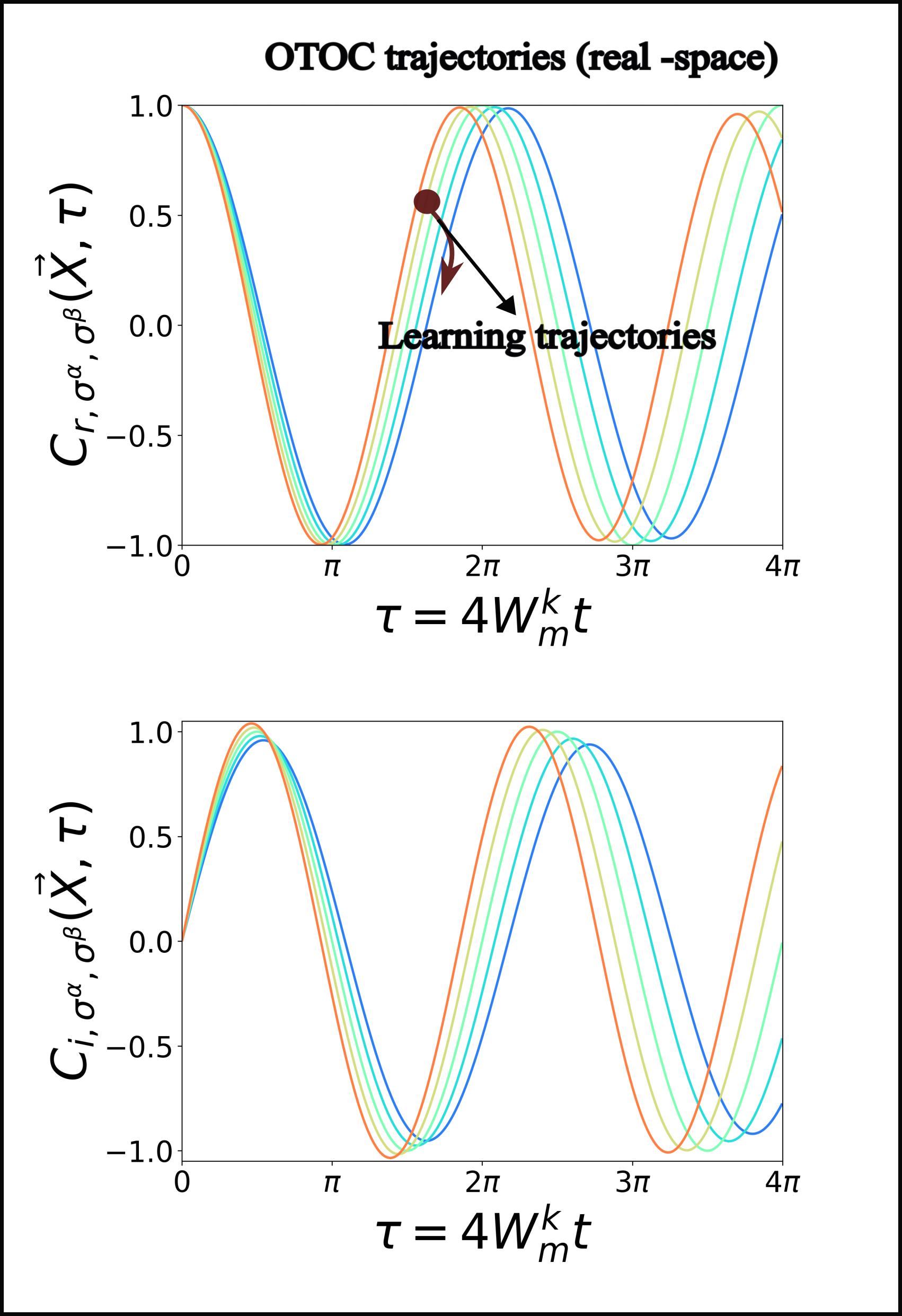}
\caption{(a) The harmonic behavior of $C_{r, \sigma^\alpha, \sigma^\beta}(\vec{X}, \tau)$ in real-time with the frequency equal to $4W^k_m$. Note while training $G$, a typical learning trajectory in parameter space $\vec{X}$ amounts to hopping from one such curve to another as indicated. (b) The harmonic behavior of $C_{i,\sigma^\alpha, \sigma^\beta}(\vec{X}, \tau)$ in real-time with the same frequency as in Fig.\ref{Fig:Time_traj_OTOC}(a) i.e. $4W^k_m$ but phase-shifted from Fig.\ref{Fig:Time_traj_OTOC}(a) by $\frac{\pi}{2}$ (refer to Section \ref{thm_1_proof}, Corollary \ref{soln_C_t} }
\label{Fig:Time_traj_OTOC}
\end{figure}

\subsection{Profile of Invariants during Training}\label{inv_training_plot_section}

For demonstration as how these invariants of motion changes during training of the network $G$ to learn the ground state of a given driver in epoch time , we have used the TFIM model (see Eq.\ref{gen_Ham}, Eq.\ref{eq:TFIM_Ham}) as the driver with $N=10$ spins. We have discussed the training algorithm in Section \ref{Train_plots}. We use $n=m=N=10$  spins in the network $G$ and plot the primitive invariants $I_{r1}$, $I_{r2}$, $I_{i1}$, $I_{i2}$ which are discussed and proven in Theorem \ref{lem1} in main manuscript and proven in this section before in Theorem \ref{Theorem_1_proof} (2). We also plot the compound invariants $\frac{I_{i2}^2 + I_{i1}^2}{4}$ for each $k,m$ pair where $k$ belongs to visible neurons and $m$ to hidden neurons and
$\frac{I_{r2}^2 + I_{r1}^2}{4}$ which is displayed in Fig.\ref{Fig:Inv_OTOC}(a-b) and proven in this section in Corollary \ref{Corr_2}. We use $100$ epochs for comparison of all invariants. For a given epoch with the incumbent instance of the parameter vector $\vec{X}$ one can plot each invariant in the space of $\xi_{\sigma^\alpha,\sigma^\beta}(\vec{X},\tau)$- 
$\dot{\xi}_{\sigma^\alpha,\sigma^\beta}(\vec{X},\tau)$ space where $\xi_{\sigma^\alpha,\sigma^\beta}(\vec{X},\tau)$ can either be the real or the imaginary part of $(C_{\sigma^\alpha,\sigma^\beta}(0,0,\vec{X},\tau))$ as has been done in the main manuscript. We see that certain invariants computed using the imaginary parts of OTOC string (see Eq.\ref{eq:C_t_rbm} in main manuscript) have non-trivial evolution (for example see Fig. \ref{Fig:Inv_plot_train_TFIM}(a,e)) unlike those computed using the real parts. This gives us a direct evidence why imaginary part of OTOC strings can be of use to gain insight into the learning mechanism which may not be obtainable from the real part - a claim central to the thesis of our manuscript. 
One functional importance that may stem from analyzing such invariants using the imaginary part is identifying certain pair of neurons $(k,m)$ which reports an invariant value that remains nearly conserved and close to $0$ during the entire course of training. Such neurons can be considered to not undergo information exchange (and hence remains nearly uncorrelated). A common set of neurons of the hidden node which shares such a property with any of the neurons of the visible set may be considered redundant neurons and hence can be discarded for more compact subsequent trial which can reduce the cost of the training. Designing markers through which identification of such markers can be enabled through estimation of the invariants may be a fruitful future direction that can benefit from a thorough investigation.

\begin{figure}[!htb]
\centering \includegraphics[width=0.75\textwidth]{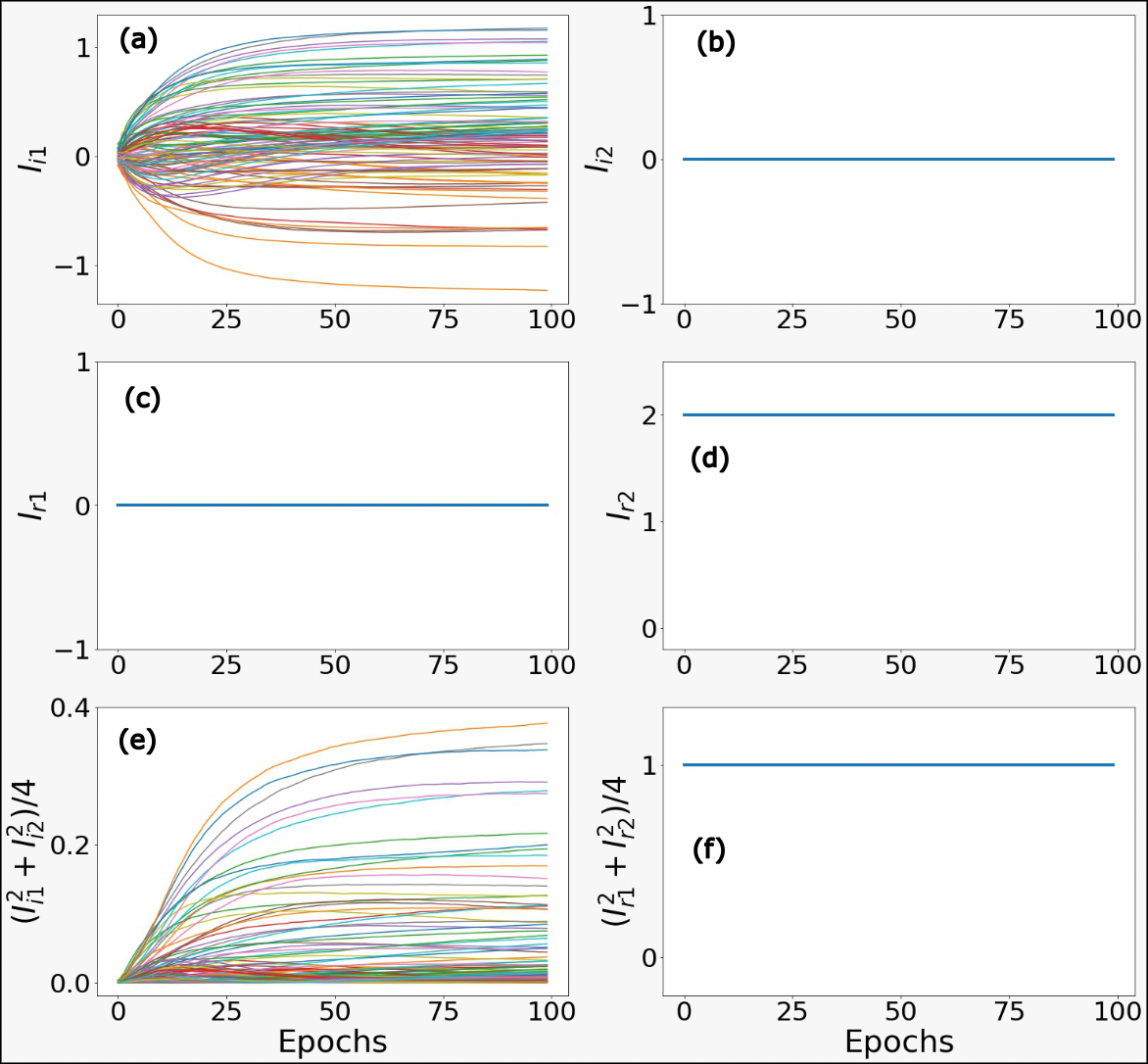}
\caption{\color{black} (a) The invariant $I_{i1}$ computed from Eq.\ref{eq:IV3} using the imaginary part of the OTOC string (defined in Eq.\ref{eq:C_t_rbm}) as a function of training epoch. This invariant changes as the training vector $\vec{X}$ varies during training and hence can provide information about the learning landscape (b) The invariant $I_{i2}$ computed from Eq.\ref{eq:IV4} using the imaginary part of the OTOC string (defined in Eq.\ref{eq:C_t_rbm}) as a function of training epoch. (c) Same invariant as in (a) but computed with the real part Eq.\ref{eq:IV1} of the OTOC string (defined in Eq.\ref{eq:C_t_rbm}). This remains constant during training
(d) Same invariant as in (c) but computed with the real part Eq.\ref{eq:IV2} of the OTOC string (defined in Eq.\ref{eq:C_t_rbm}). This remains constant during training
(e) The compound invariant $\frac{I_{i1}^2 + I_{i2}^2}{4}$ obtained from (a) and (b) vs training epoch. Such compound invariants are plotted in Fig.\ref{Fig:Inv_OTOC}(b) in main manuscript in $(C_{i\sigma^\alpha,\sigma^\beta}(\vec{X},\tau))$ vs 
$(\dot{C}_{i\sigma^\alpha,\sigma^\beta}(\vec{X},\tau))$ space. Since (a),(e) is sensitive to $\vec{X}$ and hence changes during training it is evident why studying such invariants generated from the imaginary part of OTOC strings can be useful compared to the real part alone.}
\label{Fig:Inv_plot_train_TFIM}
\end{figure}

\newpage
\section{Training of TFIM and c-TFIM and deduction of eigenvalues of $^2\rho(v_k,h_m)$}\label{Train_plots}

The figure Fig.\ref{Fig:Training_TFIM} and  Fig.\ref{Fig:Training_cTFIM} depicts the training of RBM network $G$ for TFIM and c-TFIM model respectively. It shows the variation of Energy accuracy with epochs for training $G$. For each training process parameters of RBM are initialized randomly. For training the network we use Variational Monte Carlo technique with Stochastic Reconfiguration based gradient updates as illustrated in Ref[74] of main manuscript. 
Learning rate used is 0.05 and $(n,p) = (4,4)$ in the network $G$. The convergence threshold set is $\le 10^{-2}$. The relative error in the converged state is less than 0.1\%. Energy accuracy = $\langle \psi(\vec{X})| H |\psi(\vec{X}) \rangle $ - $\lambda_0$ where Hamiltonian $H$ is that of the driver (TFIM or c-TFIM) and the state $\psi(\vec{X})$ is the ansatz for the corresponding ground state. $\lambda_0$ is the smallest eigenvalue (true ground state energy) of the Hamiltonian. 

\begin{figure}[!htb]
\centering \includegraphics[width=0.65\textwidth]{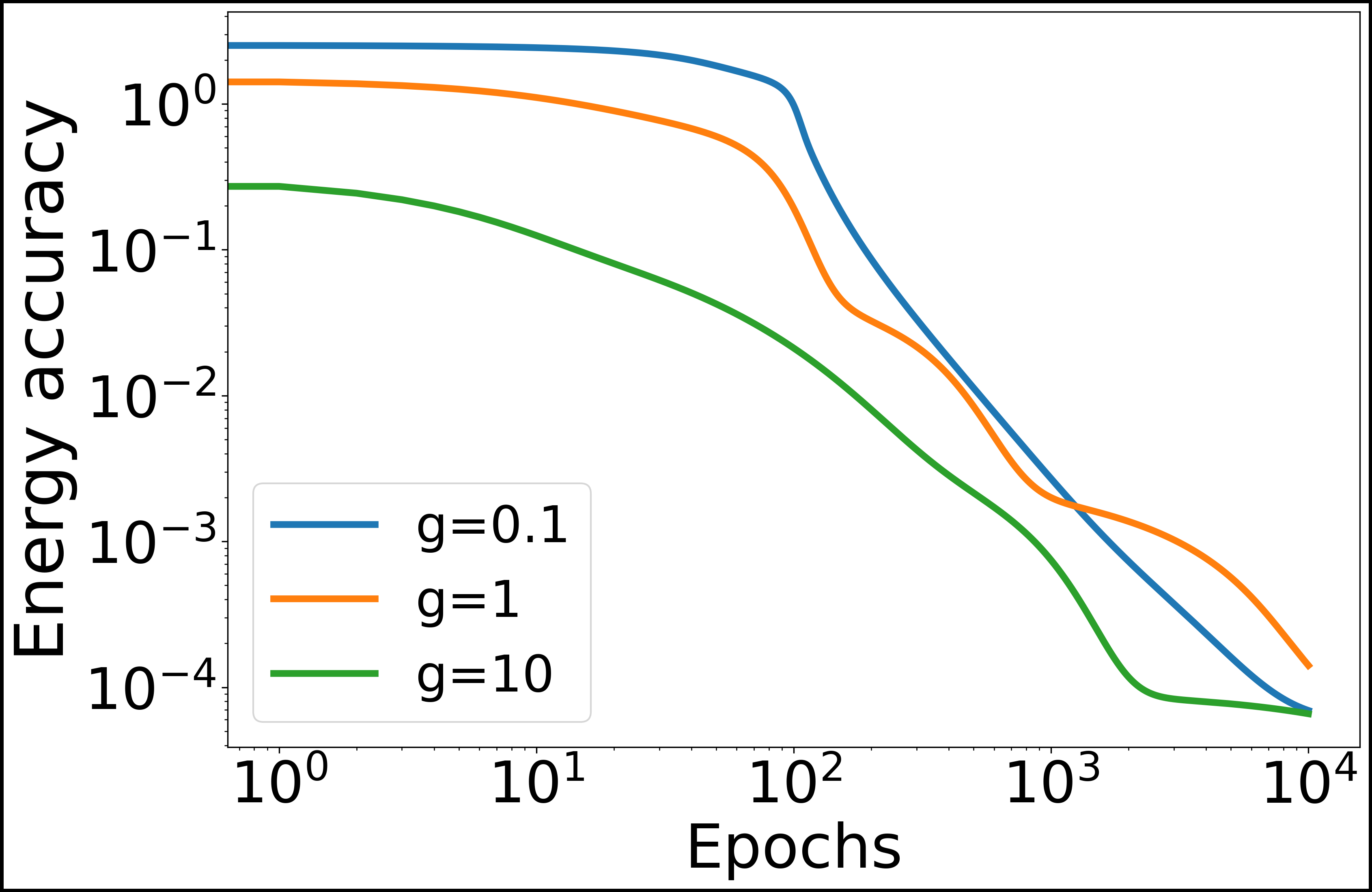}
\caption{\color{black} TFIM }
\label{Fig:Training_TFIM}

\centering \includegraphics[width=0.65\textwidth]{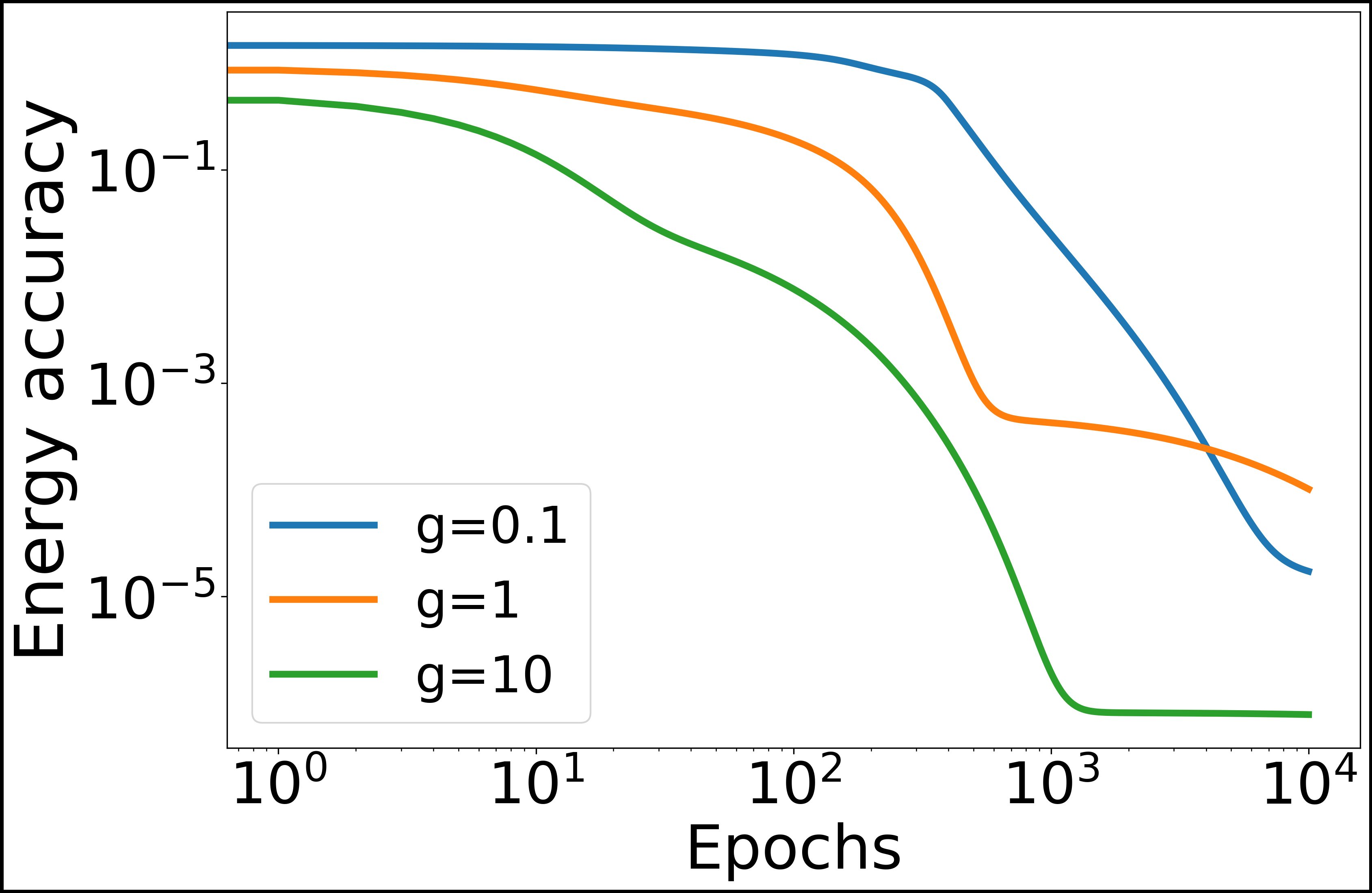}
\caption{\color{black} c-TFIM }
\label{Fig:Training_cTFIM}
\end{figure}

For Fig. 3 - (d, e, g, h) in the main manuscript, to obtain each point, we average over all pairs of visible and hidden indices for $\mathcal{I}(v_k,h_m)$ and $\eta(\vec{X})$ (i.e. $k=\{1,2,3,4\}$, $m=\{1,2,3,4\}$ implying 16 pairs) and another averaging over a sample of 100 best-converged points obtained from the training of 1000 randomly initialized networks.

\newpage

\subsection{Deduction of the eigenvalues of $^2\rho(v_k,h_m)$} \label{eig_2rdm_deduction}

In this sub-section we explictly deduce the eigenvalues of $^2\rho(v_k,h_m)$ in Eq.\ref{lambda_2RDM_1}-\ref{lambda_2RDM_4} in the main manuscript. To do that let us compute a specific matrix element (say $(v_k, h_m,v_k^\prime, h_m^\prime)$ th element) of $^2\rho(v_k,h_m)$ by direct contraction of $\rho_{th}(\vec{X},\vec{v}, \vec{h})$ in Eq.\ref{eq:rbm_dist} as follows:

\begin{align}
^2\rho(\sigma^z(v_k),\sigma^z(h_m))^{v_k, h_p}_{v_k^\prime, h_m^\prime} &= \frac{\delta_{v_k,v_k^\prime} \delta_{h_m,h_m^\prime}}{Tr_{\{v,h\}}e^{-\mathcal{H}(\vec{X}, \vec{v}, \vec{h})}}\sum_{v_i\ne v_k} e^{-\beta\sum_i^n a_i v_i}\sum_{h_j\ne h_m} e^{-\beta\sum_j^p b_j h_j - \beta\sum_{ij}^{n,p}W^i_j v_i h_j} \nonumber \\
&= \frac{\delta_{v_k,v_k^\prime} \delta_{h_m,h_m^\prime}}{Tr_{\{v,h\}}e^{-\mathcal{H}(\vec{X}, \vec{v}, \vec{h})}}\sum_{v_i\ne v_k} e^{-\sum_i^n a_i v_i}e^{-b_m h_m -\sum_{i}^{n}W^i_m v_i h_m} 
\sum_{h_j\ne h_m} \Pi_{j\ne m}^p e^{-b_j h_j -\sum_{i}^{n}W^i_j v_i h_j}\nonumber \\
&= \frac{\delta_{v_k,v_k^\prime} \delta_{h_m,h_m^\prime}}{Z}\sum_{v_i\ne v_k} e^{-a_k v_k -b_m h_m - W^k_m v_k h_m}e^{-\sum_{i \ne k}^n a_i v_i - W^i_m v_i h_m} 
\sum_{h_j\ne h_m} \Pi_{j\ne m}^p e^{- b_j h_j - \sum_{i}^{n}W^i_j v_i h_j}\nonumber \\
&= \frac{\delta_{v_k,v_k^\prime} \delta_{h_m,h_m^\prime}e^{-a_k v_k -b_m h_m - W^k_m v_k h_m}}{Z}\sum_{v_i\ne v_k} e^{-\sum_{i \ne k}^n a_i v_i - W^i_m v_i h_m} 
\sum_{h_j\ne h_m} \Pi_{j\ne m}^p e^{-b_j h_j - \sum_{i}^{n}W^i_j v_i h_j}\nonumber \\
&= \frac{\delta_{v_k,v_k^\prime} \delta_{h_m,h_m^\prime}e^{-a_k v_k -b_m h_m - W^k_m v_k h_m}}{Z}\sum_{v_i\ne v_k} e^{-\sum_{i \ne k}^n a_i v_i - W^i_m v_i h_m} 
\Pi_{j\ne m}^p 2cosh(-b_j - \sum_{i}^{n}W^i_j v_i)\nonumber \\
&= \frac{\delta_{v_k,v_k^\prime} \delta_{h_m,h_m^\prime}N_{h_m} e^{-a_k v_k -b_m h_m - W^k_m v_k h_m}}{Z}
\langle \Pi_{j\ne m}^p 2cosh(b_j + \sum_{i}^{n}W^i_j v_i) \rangle _{P(\{v_i\}_{i \ne k}^n, h_m)}\label{final_2rdm_matrix_elem}
\end{align} 
where in Eq.\ref{final_2rdm_matrix_elem} the following definitions are used
\begin{align}
cosh(x) &= cosh(-x) \\
P(\{v_i\}_{i\ne k}^n, h_m) &= \frac{e^{-(\sum_{i \ne k}^n a_i v_i + W^i_m v_ih_m))}}{N_{h_m}} \\
N_{h_m} & = \Pi_{i\ne k}^n 2Cosh(a_i + \sum_{i\ne k}^n W^i_m h_m) \\
Z&=Tr_{\{v,h\}}e^{-\mathcal{H}(\vec{X}, \vec{v}, \vec{h})}
=\sum_{(\vec{v}, \vec{h})}\hspace*{-0.01in} e^{(-\sum_{i}^n -a_i v_i - \sum_{i}^p b_j h_j -\sum_{i,j}^{n,p} W^i_j v_i h_j)}
\end{align}
Eq.\ref{final_2rdm_matrix_elem} indicates due to the presence of $\delta_{v_k,v_k^\prime} \delta_{h_m,h_m^\prime}$ factors that $^2\rho(v_k,h_m)$ is an entirely diagonal matrix. Thus the 4 eigenvalues of $^2\rho(v_k,h_m)$ can be written as
\begin{align}
\lambda_i(^2\rho(v_k,h_m)) &= \lambda(^2\rho(v_k=x,h_m=y)) \nonumber \\ &= \frac{N_{1}}{Z} e^{-a_kx - b_my - W^k_mxy} \langle \Pi_{j\ne m}^p 2Cosh (b_j + \sum_{i\ne k}^n W^i_j v_i + W^k_jx)\rangle_{P(\{v_i\}_{i\ne k}^n, h_m =y)}   \:\:\:\:\:\forall (x,y) \in \{-1,1\}^2 \label{eig_deduction_final}
\end{align}
Eq.\ref{eig_deduction_final} yields the four eigenvalues displayed in Eq.\ref{lambda_2RDM_1}-Eq.\ref{lambda_2RDM_4} for various values of $(x,y)$

\newpage

\section{Relationship between $\eta(\vec{X})$ (see Eq.9 in main manuscript) and $\mathcal{I}(v_k,h_m)$ - Proof of Theorem 2 in main manuscript} \label{thm_2_proof}

The four eigenvalues $\{\lambda_i(^2\rho(v_k,h_m))\}_{i=1}^4$ of the two-particle density matrix $^2\rho(v_k,h_m)$ for the learner $G$ between a specific pair of visible and hidden spins (say $(k,m)$) as deduced in the previous section can be readily expressed as 
\begin{align}
\lambda_1(^2\rho(v_k,h_m)) &= \lambda(^2\rho(v_k=1,h_m=1)) \nonumber \\ &= \frac{N_{1}}{Z} e^{-a_k - b_m - W^k_m} \langle \Pi_{j\ne m}^p 2Cosh (b_j + \sum_{i\ne k}^n W^i_j v_i + W^k_j)\rangle_{P(\{v_i\}_{i\ne k}^n, h_m =1)} \label{lambda_2RDM_A1}\\
\lambda_2(^2\rho(v_k,h_m)) &= \lambda(^2\rho(v_k=1,h_m=-1)) \nonumber \\ &= \frac{N_{-1}}{Z} e^{-a_k + b_m + W^k_m} \langle \Pi_{j\ne m}^p 2Cosh (b_j + \sum_{i\ne k}^n W^i_j v_i + W^k_j)\rangle_{P(\{v_i\}_{i\ne k}^n, h_m =-1)} \label{lambda_2RDM_A2}\\
\lambda_3(^2\rho(v_k,h_m)) &= \lambda(^2\rho(v_k=-1,h_m=1)) \nonumber \\ &= \frac{N_{1}}{Z} e^{a_k - b_m + W^k_m} \langle \Pi_{j\ne m}^p 2Cosh (b_j + \sum_{i\ne k}^n W^i_j v_i - W^k_j)\rangle_{P(\{v_i\}_{i\ne k}^n, h_m =1)} \label{lambda_2RDM_A3}\\
\lambda_4(^2\rho(v_k,h_m)) &= \lambda(^2\rho(v_k=-1,h_m=-1)) \nonumber \\ &= \frac{N_{-1}}{Z} e^{a_k + b_m - W^k_m} \langle \Pi_{j\ne m}^p 2Cosh (b_j + \sum_{i\ne k}^n W^i_j v_i - W^k_j)\rangle_{P(\{v_i\}_{i\ne k}^n, h_m =-1)} \label{lambda_2RDM_A4}
\end{align}

where each of the averages are computed over the distribution $P(\{v_i\}_{i\ne k}^n, h_m)$ and $N_{h_m}$ is the associated normalization constant. These are defined as 
\begin{align}
P(\{v_i\}_{i\ne k}^n, h_m) &= \frac{e^{-(\sum_{i \ne k}^n a_i v_i + W^i_m v_ih_m))}}{N_{h_m}} \\
N_{h_m} & = \Pi_{i\ne k}^n 2Cosh(a_i + \sum_{i\ne k}^n W^i_m h_m) \label{P_2RDM_A}
\end{align}

The corresponding eigenvectors of the two particle density matrix for the eigenvalues in Eq.\ref{lambda_2RDM_A1}-\ref{lambda_2RDM_A4} are $|0(v_k)0(h_m)\rangle$, $|0(v_k)1(h_m)\rangle$ and $|1(v_k)0(h_m)\rangle$, $|1(v_k)1(h_m)\rangle$ respectively for the four eigenvalues Eq.\ref{lambda_2RDM_A1}-Eq.\ref{lambda_2RDM_A4} where $(0,1)$ is notationally equivalent to $(1,-1)$ for each spins $(v_k, h_m)$

The quantity $Z$ is the partition function defined as 
\begin{align}
Z&=\sum_{(\vec{v}, \vec{h})}\hspace*{-0.01in} e^{(-\sum_{i}^n -a_i v_i - \sum_{i}^p b_j h_j -\sum_{i,j}^{n,p} W^i_j v_i h_j)}
\end{align}
However $Z$ need not be explicitly computed as it can be eliminated using the normalization condition of the eigenvalues. 
The eigenvalues for one-particle density matrix $^1\rho(\xi_i,0)$ for a neuron $\xi_i$ in the learner $G$, by contraction from 
Eq.\ref{lambda_2RDM_A1}, \ref{lambda_2RDM_A2}, \ref{lambda_2RDM_A3}, \ref{lambda_2RDM_A4} are
\begin{align}
\lambda_1(^1\rho(\xi_i)) &= \lambda_i(^2\rho(v_k,h_m)) + \lambda_j(^2\rho(v_k,h_m)) \label{1RDM_eig_A1}\\&\:\:\: (\text{if} \:\: \xi_i=v_k, \:\: (i,j)=(1,3)) \nonumber \\&\:\:\:
(\text{if} \:\: \xi_i=h_m, \:\: (i,j)=(1,4)) \nonumber \\
\lambda_2(^1\rho(\xi_i)) &= \lambda_i(^2\rho(v_k,h_m)) + \lambda_j(^2\rho(v_k,h_m)) \label{1RDM_eig_A2}\\&\:\:\: (\text{if} \:\: \xi_i=v_k, \:\: (i,j)=(2,4)) \nonumber \\&\:\:\:
\text{if} \:\: (\xi_i=h_m, \:\: (i,j)=(2,3)) \nonumber 
\end{align}
with respective eigenvectors are $|0(\xi_i)\rangle$ and $|1(\xi_i)\rangle$
where $\xi_i \in (v_k, h_m)$.

Now using these information, one can deduce expressions for $S(^2\rho(v_k,h_m))$ and $S(^1\rho(v_k))$, $S(^1\rho(h_m))$
and hence of $\mathcal{I}(v_k,h_m)$ as follows
\begin{align}
\mathcal{I}(v_k,h_m) &= S(^1\rho(v_k)) +S(^1\rho(h_m))-S(^2\rho(v_k,h_m)) \nonumber \\
&=-\lambda_1(^1\rho(v_k))log_2(\lambda_1(^1\rho(v_k))) -\lambda_2(^1\rho(v_k))log_2(\lambda_2(^1\rho(v_k)))- \lambda_1(^1\rho(h_m))log_2(\lambda_1(^1\rho(h_m)))
\nonumber \\ 
&-\lambda_2(^1\rho(h_m))log_2(\lambda_2(^1\rho(h_m)))+\lambda_1(^2\rho(v_k,h_m))log_2(\lambda_1(^2\rho(v_k,h_m)))+\lambda_2(^2\rho(v_k,h_m))log_2(\lambda_2(^2\rho(v_k,h_m))) \nonumber \\
&+\lambda_3(^2\rho(v_k,h_m))log_2(\lambda_3(^2\rho(v_k,h_m))) +\lambda_4(^2\rho(v_k,h_m))log_2(\lambda_4(^2\rho(v_k,h_m))) 
\nonumber \\ &-(\lambda_1(^2\rho(v_k,h_m) + \lambda_3(^2\rho(v_k,h_m))\nonumber log_2(\lambda_1(^2\rho(v_k,h_m)) + \lambda_3(^2\rho(v_k,h_m)))\nonumber \\
&\:\:\:\:-(\lambda_2(^2\rho(v_k,h_m) + \lambda_4(^2\rho(v_k,h_m))log_2(\lambda_2(^2\rho(v_k,h_m)) + \lambda_4(^2\rho(v_k,h_m)))\nonumber \\
&\:\:\:\:-(\lambda_1(^2\rho(v_k,h_m) + \lambda_4(^2\rho(v_k,h_m)) log_2(\lambda_1(^2\rho(v_k,h_m)) + \lambda_4(^2\rho(v_k,h_m)))\nonumber \\
&\:\:\:\:\: -(\lambda_2(^2\rho(v_k,h_m) + \lambda_3(^2\rho(v_k,h_m))\nonumber log_2(\lambda_2(^2\rho(v_k,h_m)) + \lambda_3(^2\rho(v_k,h_m)))\nonumber\\
&\:\:\:\:+\lambda_1(^2\rho(v_k,h_m))log_2(\lambda_1(^2\rho(v_k,h_m))) +\lambda_2(^2\rho(v_k,h_m))log_2(\lambda_2(^2\rho(v_k,h_m))) \nonumber \\
&\:\:\:\:+\lambda_3(^2\rho(v_k,h_m))log_2(\lambda_3(^2\rho(v_k,h_m))) +\lambda_4(^2\rho(v_k,h_m))log_2(\lambda_4(^2\rho(v_k,h_m))) \label{I_expr}
\end{align}\\
wherein in the last equality in Eq.\ref{I_expr}, contraction schemes from Eq.\ref{1RDM_eig_A1},Eq.\ref{1RDM_eig_A2} were used. Eq.\ref{I_expr} thus expresses $\mathcal{I}(v_k, h_m)$ in terms of the eigenvalues of $^2\rho(v_k,h_m))$. Apart from this the normalization condition also inter-relates the eigenvalues as 
\begin{align}
&\lambda_1(^2\rho(v_k,h_m)) + \lambda_2(^2\rho(v_k,h_m)) + \lambda_3(^2\rho(v_k,h_m)) + \lambda_4(^2\rho(v_k,h_m))
= 1 \label{Eq_2rdm_tr}
\end{align}
The expression for the $\eta(\vec{X})$ as follows
\begin{align}
\eta(\vec{X}) &= Cov(\sigma^z(v_k,0), \sigma^z(h_m,0)) \nonumber \\
&= Tr(\sigma^z(v_k,0)\sigma^z(h_m,0)\:\:^2\rho(v_k,h_m))- \langle \sigma^z(v_k,0)\rangle_{^1\rho_{v_k}}\langle \sigma^z(h_m,0)\rangle_{^1\rho_{h_m}} \nonumber \\
&= \lambda_1(^2\rho(v_k,h_m)) - \lambda_2(^2\rho(v_k,h_m)) - \lambda_3(^2\rho(v_k,h_m)) +\lambda_4(^2\rho(v_k,h_m)) - (\sum_{i=1,3}\lambda_i(^2\rho(v_k,h_m)) - \sum_{i=2,4}\lambda_i(^2\rho(v_k,h_m)))\nonumber \\
&\:\:\:\:\:\:\:\:\:\:\:(\sum_{i=1,4}\lambda_i(^2\rho(v_k,h_m))- \sum_{i=2,3}\lambda_i(^2\rho(v_k,h_m)))
\label{eta_x_expr}
\end{align}
Also to ensure positive-semi-definiteness of $^2\rho(v_k,h_m)$, we have 
\begin{align}
\lambda_i(^2\rho(v_k,h_m)) \geq 0  \:\:\:\: \forall i \in \{1,2,3,4\} \label{Eq_2rdm_pos}
\end{align}
Thus we see that Eq.\ref{I_expr} combined with \ref{Eq_2rdm_tr}, Eq.\ref{Eq_2rdm_pos} for a given value of $\eta(\vec{X})$ in Eq.\ref{eta_x_expr} completely defines the $\mathcal{I}-\eta$ space in terms of the spectrum of $^2\rho(v_k,h_m)$\\

\underline{Lower Bound(LB)}\\
The minimization of Eq.\ref{I_expr} over the spectrum of $^2\rho(v_k,h_m)$ with the constraints defined in Eq.\ref{Eq_2rdm_tr}, Eq.\ref{Eq_2rdm_pos} and Eq.\ref{eta_x_expr}(for a fixed value of $\eta(\vec{X})$ gives the following condition

\begin{align}
\lambda_1(^2\rho(v_k,h_m)) = \lambda_2(^2\rho(v_k,h_m)) = \lambda_a (say) \nonumber \\
\lambda_3(^2\rho(v_k,h_m)) = \lambda_4(^2\rho(v_k,h_m)) = \lambda_b (say) \label{LB_condn}
\end{align}

Substituting the above conditions in Eq.\ref{Eq_2rdm_tr} and in Eq.\ref{I_expr} leads to the following

\begin{align}
\lambda_b &= \frac{1}{2} - \lambda_a \label{Tr_LB} \\
\eta(\vec{X}) &= 4\lambda_a - 1 \:\:\:\:\: (\because Eq.\ref{Tr_LB}) \label
{Cov_LB}
\end{align}
Substituting Eq.\ref{Cov_LB}, \ref{Tr_LB} and the conditions in Eq.\ref{LB_condn}, in Eq.\ref{I_expr} yields the lower bound $LB$ as
\begin{align}
LB =\mathcal{I}_{LB}(v_k, h_m) &= 2+ 2\lambda_a log_2 (\lambda_a) + 2\lambda_b log_2 (\lambda_b) \nonumber \\
&= 2+ 2\lambda_a log_2 (\lambda_a) + 2(\frac{1}{2}-\lambda_a) log_2 ((\frac{1}{2}-\lambda_a) \nonumber \\
&= 2+ 2\lambda_a log_2 (\lambda_a) + (1-2\lambda_a) log_2 ((\frac{1}{2}-\lambda_a) \nonumber \\
&= 2+ (\frac{\eta(\vec{X})+1}{2}) log_2 (\frac{\eta(\vec{X})+1}{4}) \nonumber \\&+ (\frac{1-\eta(\vec{X})}{2}) log_2(\frac{1-\eta(\vec{X})}{4})   \:\:\:\: \because Eq.\ref{Cov_LB}
\end{align}
Note that $LB$ is symmetric about $\eta(\vec{X})=0$ and remains unchanged on substitution of $\eta(\vec{X}) \to -\eta(\vec{X})$. Also, $LB$ is mathematically only defined if $-1 \le \eta(\vec{X}) \le 1$ which is true for covariances of Pauli operators in the learner $G$ (a direct by-product of Cauchy-Schwartz inequality and idempotency of Pauli operators which bounds their respective variances to within 1). \\

\underline{Upper Bound(UB)}\\
The maximization of Eq.\ref{I_expr} over the spectrum of $^2\rho(v_k,h_m)$ with the constraints defined in Eq.\ref{Eq_2rdm_tr}, Eq.\ref{Eq_2rdm_pos} and Eq.\ref{eta_x_expr}(for a fixed value of $\eta(\vec{X})$ gives the following condition

\begin{align}
\lambda_1(^2\rho(v_k,h_m)) = \lambda_2(^2\rho(v_k,h_m)) = 0 \nonumber \\
\lambda_3(^2\rho(v_k,h_m)) = \lambda_a (say) \nonumber \\
\lambda_4(^2\rho(v_k,h_m)) = \lambda_b (say) \label{UB_condn}
\end{align}

Now using Eq.\ref{UB_condn} in Eq.\ref{Eq_2rdm_tr} and Eq.\ref{eta_x_expr} we get 
\begin{align}
\lambda_b &= 1 - \lambda_a \label{Tr_UB} \\
\eta(\vec{X}) &= -1 + (2\lambda_a -1)^2 
\nonumber \\
&= 4\lambda_a(4\lambda_a - 1) \:\:\:\:\: (\because\:\:\: Eq.\ref{Tr_UB})\nonumber \\
\lambda_a &= \frac{1\pm \sqrt{1+(-1)^\gamma \eta(\vec{X})}}{2}
 \label{Cov_UB}
\end{align}
wherein in the last equality to maintain positive semi-definiteness of $\lambda_a$ (hence for $\lambda_3$) in both the roots, a factor of $(-1)^\gamma$ was used along with the condition $-1 \le \eta(\vec{X}) \le 1$. Note that $\gamma = 0$ when $\eta(\vec{X}) < 0$ and $\gamma =1$ when $\eta(\vec{X}) \ge 0$. Now substituting Eq.\ref{UB_condn}, Eq.\ref{Tr_UB} and  Eq.\ref{Cov_UB} in Eq.\ref{I_expr} we obtain UB as 
\begin{align}
UB &=\mathcal{I}_{UB}(v_k, h_m) \nonumber \\&=  -\lambda_a log_2 (\lambda_a) -\lambda_b log_2 (\lambda_b) \nonumber \\
&= -\lambda_a log_2 (\lambda_a) -(1-\lambda_a) log_2 (1-\lambda_a) 
\:\:\:\:\: (\because\:\:\: Eq.\ref{Tr_UB})\nonumber \\
& = \big(\frac{1}{2} + \frac{\sqrt{1+(-1)^\gamma\eta(\vec{X})}}{2} \big) log_2 \big(\frac{1}{2} + \frac{\sqrt{1+(-1)^\gamma\eta(\vec{X})}}{2} \big) \nonumber \\
& - \big(\frac{1}{2} - \frac{\sqrt{1+(-1)^\gamma\eta(\vec{X})}}{2} \big) log_2 \big( \frac{1}{2} - \frac{\sqrt{1+(-1)^\gamma\eta(\vec{X})}}{2} \big) \\&\:\:\:\:\: (\because\:\:\: Eq.\ref{Cov_UB}) \label{UB_expr} 
\end{align}
where substitution of either root from Eq.\ref{Cov_UB} would lead to the same $UB$ due to symmetry

\newpage

\newpage
\section{Saturation of Lower Bound (LB) in $\mathcal{I}(v_k,h_m)$ and $\eta(\vec{X})$ space in eigenpair learning of network $G$ } \label{Sat_LB_Section}

In this section we discuss the several systems we have used as a driver for our task of learning eigenpairs and training the network $G=(V,E)$. We use for demonstration a wide variety of systems wherein the ground state is non-negative due to Perron-Frobenius theorem as discussed in the main manuscript. Each of these model is endowed with a hamiltonian $H(\lambda_1, \lambda_2, \lambda_3...\lambda_n)$ with several generic controllable parameters $\{\lambda\}_{i=1}^n$. Tuning these controllable parameters allows one to access ground states with various phase properties for different sizes of the system. We show that for each cases for different sizes, both the assertions in the manuscript - (a) saturation of lower bound (LB) in the $\mathcal{I}(v_k,h_m)$ and $\eta(\vec{X})$ space by the representation chosen by the trained/learned network, (b) sliding of the $(|\eta(\vec{X})|, \mathcal{I}(v_k,h_m))$ points on the LB as the controllable parameters $\{\lambda\}_{i=1}^n$ are tuned across the phase boundaries. The last point evidentiates how spin correlation among the sub-units or spins of the actual driver system gets translated or mimicked on the spin correlation between the visible and hidden sub-units of the learner network $G$. It must be emphasized this associates a functionally quantifiable importance of the neurons of the hidden layer which are oblivious to the spins of the driver. It is only the spins of the visible neurons whose state in the basis of its acceptable configurations is trained to represent the physical ground state of the driver.

We start with the familiar TFIM and c-TFIM which has been explicitly discussed in the main manuscript. We have shown in Fig.\ref{Fig:MI_O_TFIM}(c-f) that for $N=4, N=20$ for both TFIM and c-TFIM the lower bound (LB) in the $\mathcal{I}(v_k,h_m)$ and $\eta(\vec{X})$ space is saturated and in Fig.\ref{Fig:MI_O_TFIM}(g-j) we show that mean $\mathcal{I}(v_k,h_m)$ and $|\eta(\vec{X})|$ values slides along the LB as the controllable parameter $g$ in $H(g)$ is changed from a ferromagnet to a disordered phase for sizes $N=4,6,8,10,12,14,16,18,20,24$. Herein we show that the same assertions hold for all intermediate sizes as well for both the models TFIM and c-TFIM with $n=p=N$ for the network $G$.

\begin{figure}[!htb]
\centering \includegraphics[width=0.8\textwidth]{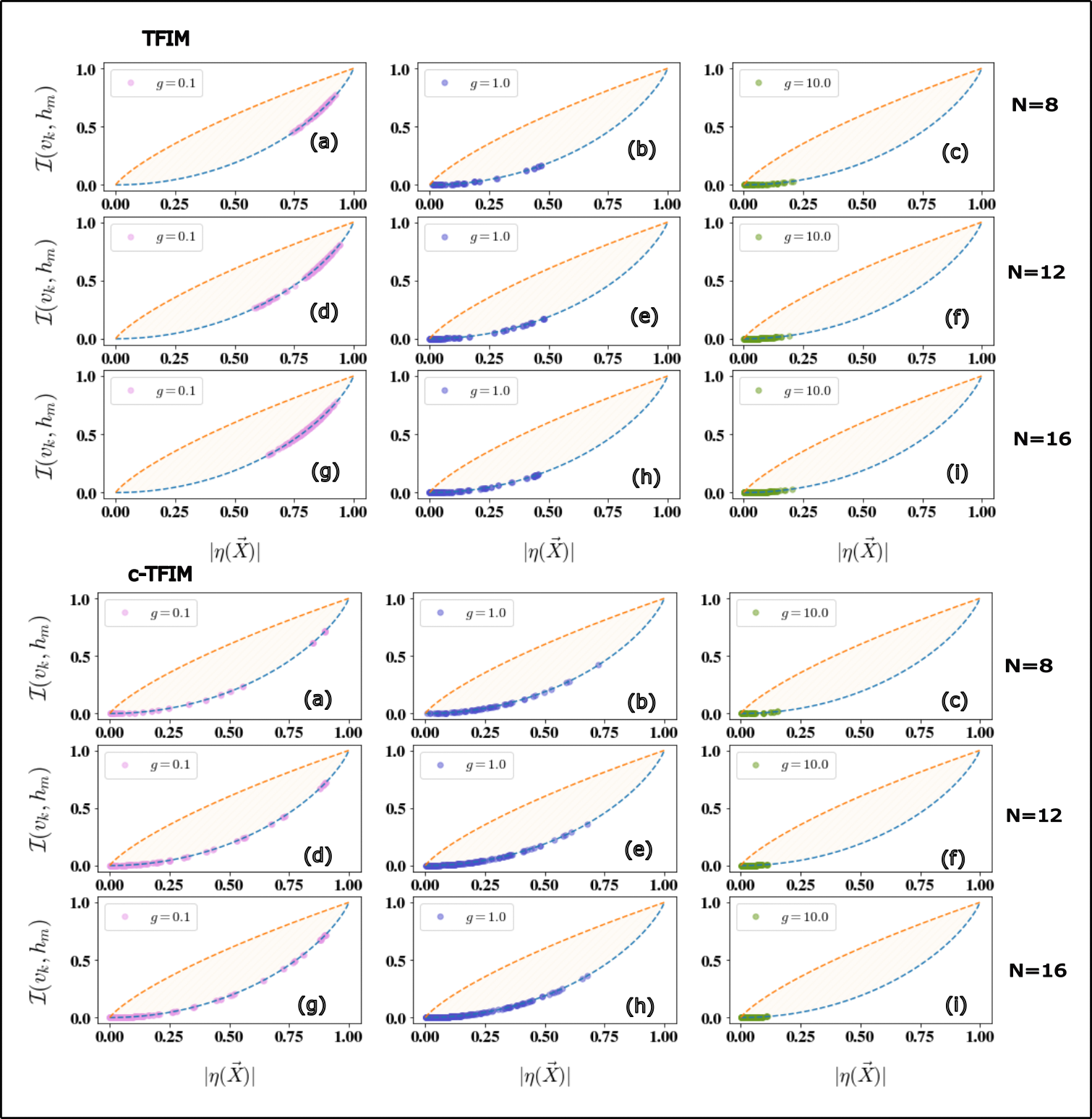}
\caption{\color{black} (a) The saturation of the lower bound for $g=0.1$ (b) Same as in (a) for $g=1.0$ (c) for $g=10.0$ for $N=8$. The plots also show clustering of  $(|\eta(\vec{X})|, \mathcal{I}(v_k,h_m))$ near upper regime for $g=0.1$, middle and lower regime in $g=1.0$ and close to the origin in $g=10.0$ (d-f) Same as in (a-c) but for $N=12$ (g-h) Same as in (a-c) but for $N=16$. For the c-TFIM panel we have similar plots as in TFIM (a-i). Note that the spread of the points in c-TFIM is larger than in TFIM, feature consistent with the discussion in the manuscript which says many equivalent representations the network chooses for volume-law connectivity}
\label{Fig:TFIM_cTFIM_all_N_LB_sat}
\end{figure}

\newpage
We next discuss the Sherrington-Kirkpatrick (SK) model with transverse magnetic field. The form of the Hamiltonian is 

\begin{align}
H &= -B\sum_{i_d}^N \sigma^x(i_d) - \sum_{{i_d}{j_d}} J_{{i_d}{j_d}} \sigma^z(i_d) \sigma^z(j_d) 
\end{align}

where unlike in TFIM and c-TFIM , each of the the coupling matrix elements $J_{{i_d}{j_d}}$ are different and is randomly sampled from a normal distribution i.e. $\mathcal{N}(0, 1)$. Note that this can mean a model with extreme inhomogeneous couplings too such that interaction strength along the length of the sites do not change monotonically unlike the connectivity pattern used in $c-TFIM$ and $TFIM$ which . We draw 5 samples from that distribution and for each train the model to show the representation chosen by the network for the ground state. In each case we see the $(|\eta(\vec{X})|, \mathcal{I}(v_k,h_m))$ pair saturates the lower bound. Note that there is no analogous definition of $g$ in this case as has been defined for TFIM and c-TFIM as the couplings are inhomogeneous. We use $N=n=p=10$ and $B=1$ for all calculations.

\begin{figure}[!htb]
\centering \includegraphics[width=0.5\textwidth]{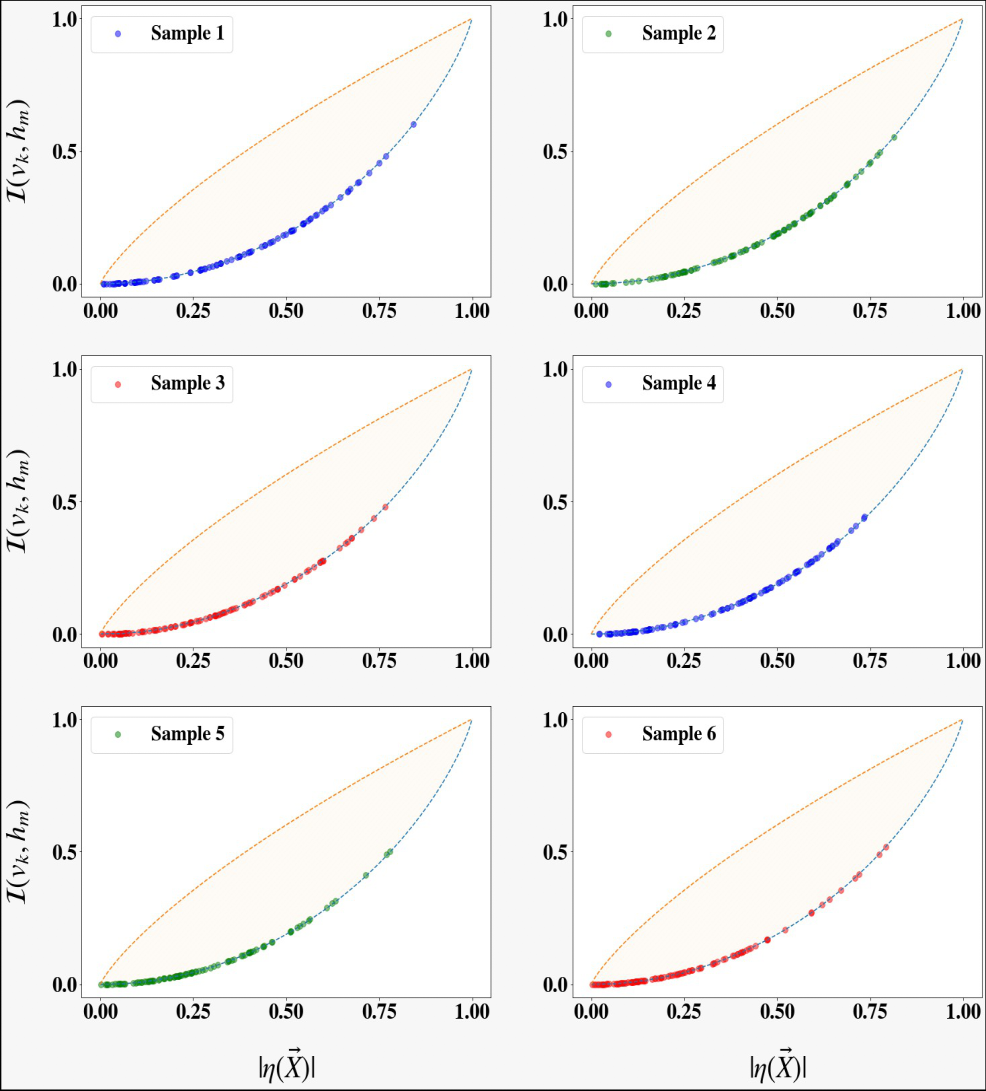}
\caption{\color{black} The saturation of the lower bound in the $\mathcal{I}(v_k,h_m)$ and $\eta(\vec{X})$ space for all five samples studied of the SK model. The sample indices are displayed in the corresponding legend }
\label{Fig:SK_all_N10_LB_sat}
\end{figure}

\newpage
We now turn our attention to another spin model that unlike the previous set of models also have another set of interaction terms involving $\sigma^y(i_d)\sigma^y(j_d)$. Moreover the interactions are even anisotropic. The Hamiltonian of the system is 

\begin{align}
H &= -B\sum_{i_d}^N \sigma^x(i_d) - \sum_{\langle {i_d}{j_d} \rangle} J(1+\gamma) \sigma^z(i_d) \sigma^z(j_d) -
\sum_{\langle {i_d}{j_d} \rangle} J(1-\gamma) \sigma^y(i_d) \sigma^y(j_d) 
\label{ZY_Ham}
\end{align}

The system possesses ground state with non-negative coefficients too as a consequence of Perron-Frobenius Theorem. We demonstrate even for this system with anisotropic YZ type interaction profile the representation chosen by the trained state of the learner ground state saturates the lower bound in the $\mathcal{I}(v_k,h_m)$ and $\eta(\vec{X})$ space. We choose $n=p=N=10$ spins and vary $(g,\gamma)$ in Eq.\ref{ZY_Ham} as shown in Fig.\ref {Fig:ZY_all_N10_LB_sat}. The parameter $g=B/J$ is defined as in the case of TFIM and c-TFIM.

\begin{figure}[!htb]
\centering \includegraphics[width=0.8\textwidth]{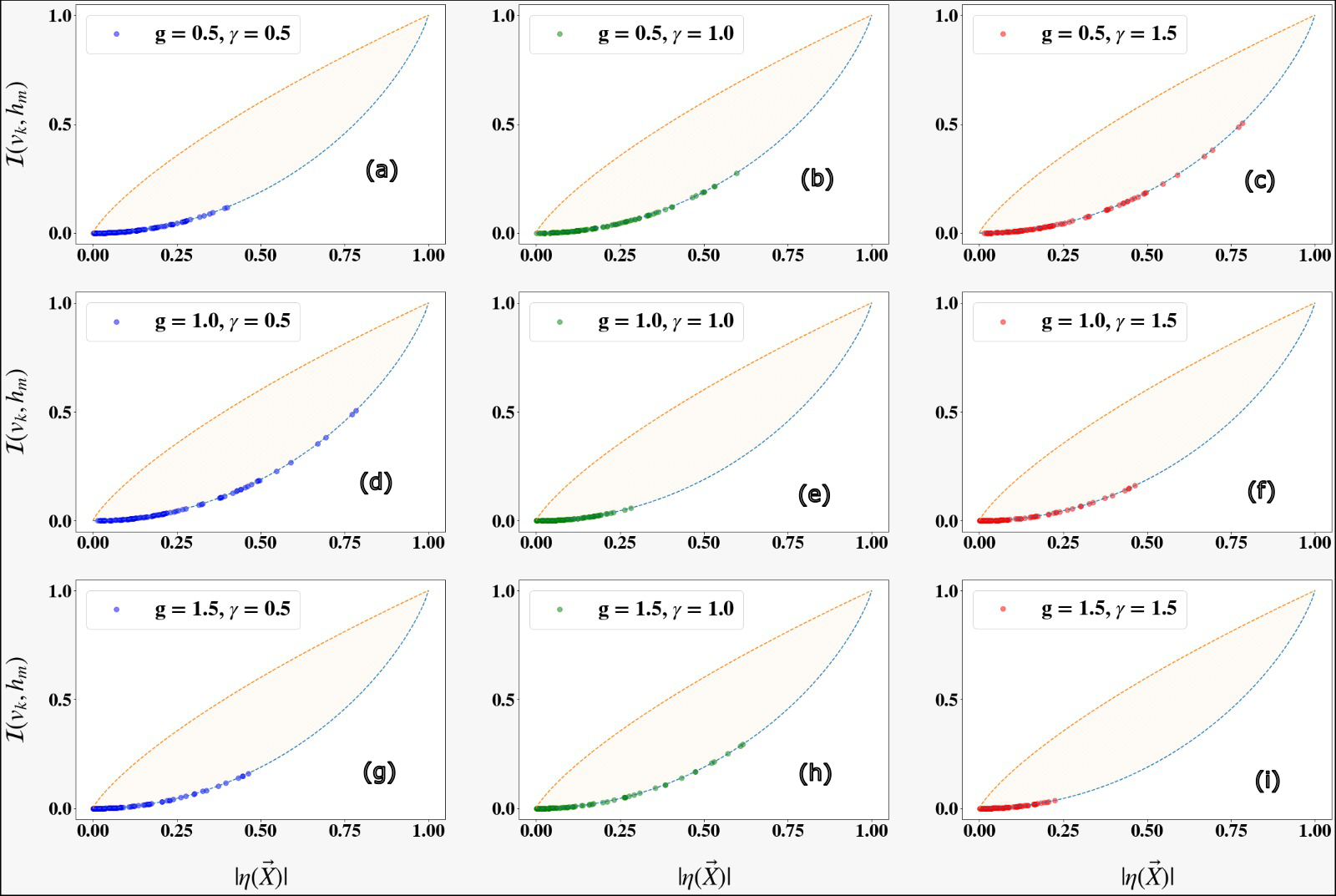}
\caption{\color{black} The saturation of the lower bound in the $\mathcal{I}(v_k,h_m)$ and $\eta(\vec{X})$ space for the trained learner for the ground state of Eq.\ref{ZY_Ham} corresponding to (a) ($g,\gamma$) = (0.5,0.5) (b) ($g,\gamma$) = (0.5,1.0), (c) ($g,\gamma$) = (0.5,1.5), (d) ($g,\gamma$) = (1.0,0.5), (e) ($g,\gamma$) = (1.0,1.0), (f) ($g,\gamma$) = (1.0,1.5), (g) ($g,\gamma$) = (1.5,0.5), (h) ($g,\gamma$) = (1.5,1.0), (i) ($g,\gamma$) = (1.5,1.5)  }
\label{Fig:ZY_all_N10_LB_sat}
\end{figure}

\newpage
\section{Effect of Hidden Node density on Training of Network $G$}\label{Hid_node_density_Section}

In this section we provide a systematic study of the variation in hidden node density for the network $G$ while training the latter network for obtaining the ground state of TFIM model (see Eq.\ref{gen_Ham},\ref{eq:TFIM_Ham} in main manuscript).
The hidden node density is defined as the ratio of the number of hidden neurons $p$ vs the number of visible $n$ neurons used in the network $G=(V,E)$ i.e. $\alpha = \frac{p}{n}$. We use for demonstration the said TFIM model as the driver with $N=10$ spins which corresponds to a Hilbert space of dimension $2^{20}$. We vary the $g=[0.2, 1.0, 5.0]$ parameter of the driver to study the ground states in various phases (see main manuscript for details). The network $G$ is trained with $n=N=10$ spins in the visible layer and $p= \alpha n$ spins in the hidden layer wherein $\alpha$ $\in$  $[1,2,3]$. The results are displayed in the Fig.\ref{Fig:Alpha_variation}. In Fig.\ref{Fig:Alpha_variation}(a-c) shows the accuracy threshold reached compared to exact diagonalization when the network is trained using different $\alpha$ for various values of $g$. We see in all cases the acquired energy error in the trained network is $10^{-2}-10^{-4}$ which registers a relative error percentage of $\le 0.1 \%$ in the worst case with no appreciable dependance on $\alpha$ within the range studied. 

\begin{figure}[!htb]
\centering \includegraphics[width=0.8\textwidth]{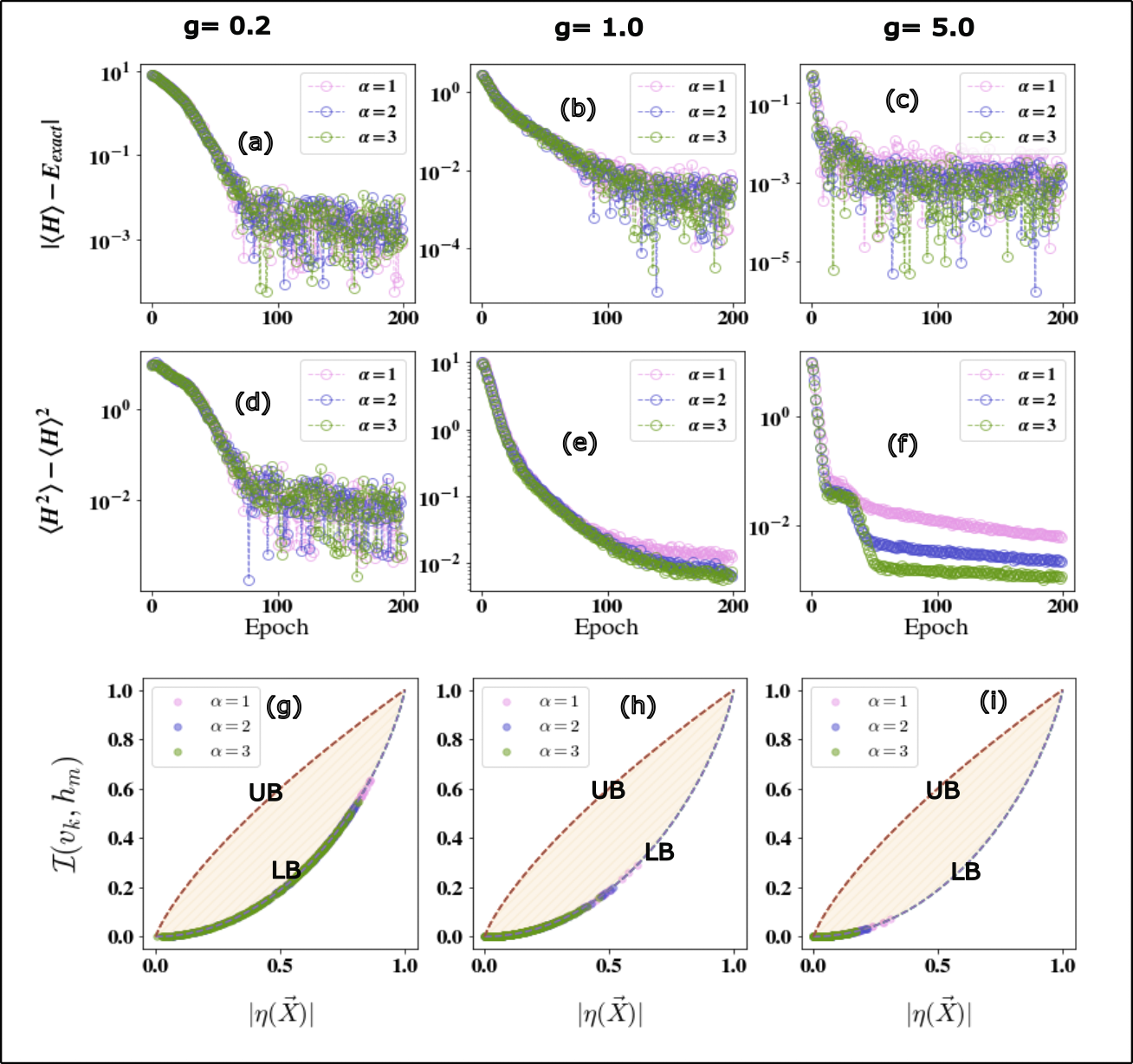}
\caption{\color{black} The energy error (with respect to exact diagonalization) acquired for training TFIM model with $n=N=10$, $p=\alpha n$ and $\alpha = [1,2,3]$ for (a) $g=0.2$ (b) $g=1$ (critical point) (c) $g=5$. The resultant energy variance of the trained state with $\alpha = [1,2,3]$ for (a) $g=0.2$ (b) $g=1$ (critical point) (c) $g=5$. The representation of the trained state in $\mathcal{I}(v_k,h_m)$ and $\eta(\vec{X})$ space for (a) $g=0.2$ (b) $g=1$ (critical point) (c) $g=5$ showing saturation of lower bound in all cases and progressive mirroring of the correlation induced between the visible and hidden spins due to the spin correlation within the driver model.}
\label{Fig:Alpha_variation}
\end{figure}

We use $\alpha=1$ for all results in the main manuscript and in the remaining portion of Appendix unless otherwise specified. In Fig.\ref{Fig:Alpha_variation}(d-f) we see that the final trained state is an eigenstate of the Hamiltonian with a very low energy variance of $\le 10^{-2}$. In Fig.\ref{Fig:Alpha_variation}(g-h) we plot the representations chosen by the trained network in the 
$\mathcal{I}(v_k,h_m)$ and $\eta(\vec{X})$ space discussed in main manuscript. We see both our assertions are individually valid i.e. remain true even when $\alpha$ is varied. Firstly we see the trained state for all values of $\alpha$ saturates the lower bound and secondly the mean density of points shifts towards the origin in the $\mathcal{I}(v_k,h_m)$ and $\eta(\vec{X})$ space signalling that the trained state of the learner $G$ is devoid of spin correlation between the hidden and the visible units mirroring the same behavior between the spins of the driver when $g \to \infty$ even though the hidden units are not directly related to the driver spins.

\newpage
\section{Standard Deviation associated with the Mean $\mathcal{I}(v_k,h_m)$ and Mean $\eta(\vec{X})$ for TFIM and c-TFIM} \label{Sample_SE_TFIM_cTFIM}

In this section we display the standard deviation (($\sqrt{Var(X)}$) where $X$ can be the mean values of variates (like $\mathcal{I}(v_k,h_m)$ etc) obtained from the mean $\mathcal{I}(v_k,h_m)$ and $\eta(\vec{X})$ values displayed in Fig.\ref{Fig:MI_eta_spin_transfer}(a-d). $Var(\cdot)$ is the variance of the respective variate. We see from the standard deviation plot below that the quantity for c-TFIM is nearly 1.5-2.0 times higher than that of TFIM for a given $g$ and a given size $N$ especially in the low $g$ limit ($g \rightarrow 0_+$)

\begin{figure}[!ht]
\centering \includegraphics[width=0.8\textwidth]{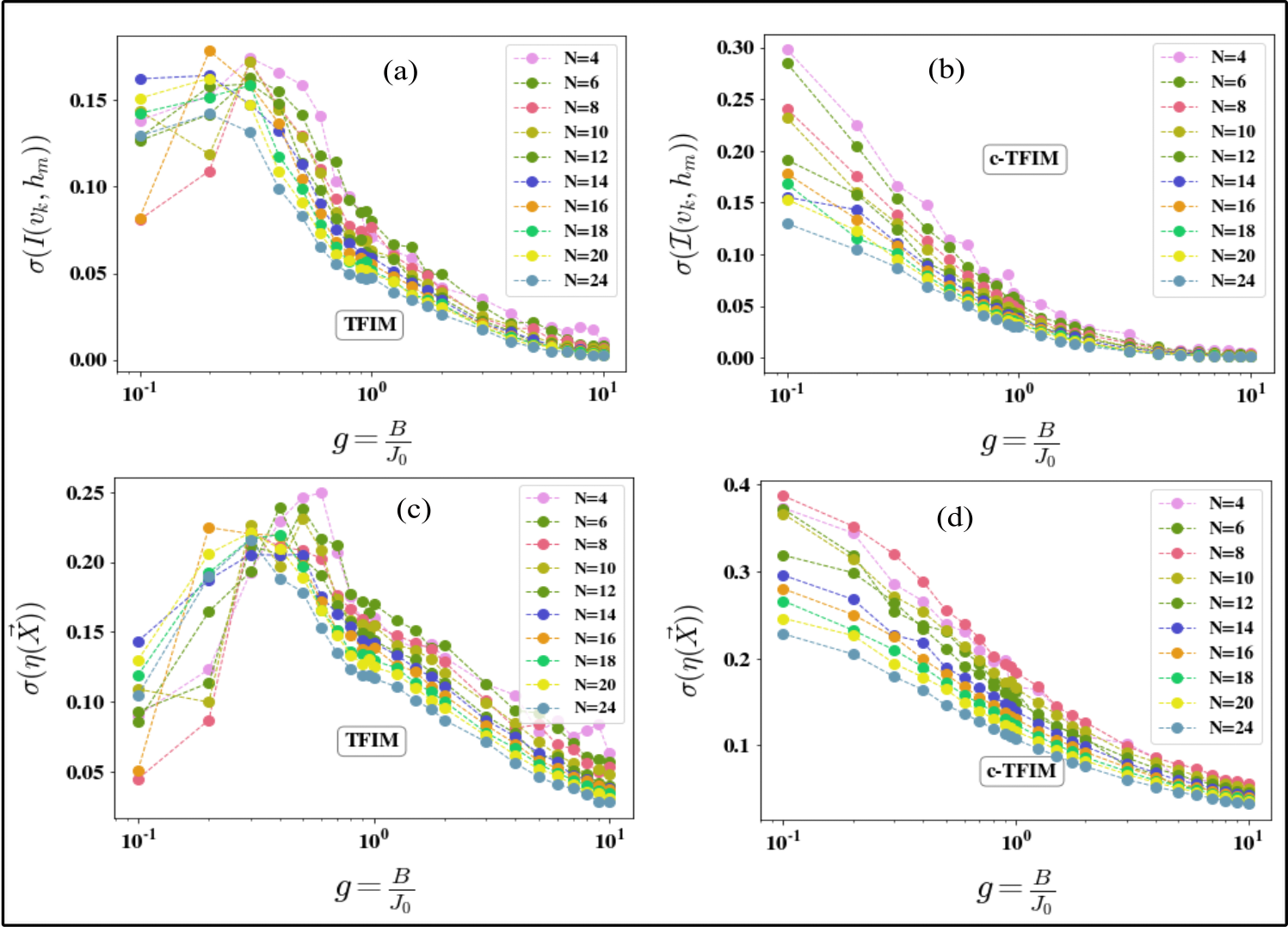}
\caption{\color{black} (a) The standard deviation associated with the averaging process of $\mathcal{I}(v_k,h_m)$ (where the mean is obtained by averaging over all $mn$ $(k,m)$ pairs of visible and hidden neurons and different initializations. See Section \ref{Result_Disc}) for TFIM for all sizes of the driver (c) The standard deviation associated with the averaging process of the mean $\eta(\vec{X})$ (where the mean is obtained by averaging over all $mn$ $(k,m)$ pairs of visible and hidden neurons and different initializations. See Section \ref{Result_Disc}) for TFIM for all sizes of the driver (b-d) Same as in (a-c) but for c-TFIM. It is clear from the standard deviation plots that the variance of the $\mathcal{I}(v_k,h_m)$ and $\eta(\vec{X})$ is much larger by almost 1.5-2.0 times (especially in the low $g$ limit) in c-TFIM (see (b-d)) than in TFIM (see (a-c)) for all values of $N$. Since the variance is computed over all $(k,m)$ pairs and over many initialization, this indicates many compatible configurations/representations for a given $N$ that the trained network can `learn' and display for the volume-law entangled c-TFIM than the area-law entangled TFIM. In the higher $g$ limit unanimity sets in both models as both the drivers display an unentangled ground state devoid of spin correlation. (See Section \ref{Result_Disc}) }
\label{Fig:SE_Dev}
\end{figure}


\newpage
\section{Von-Neumann Entropy of TFIM and c-TFIM}
\label{VNE_plots}
The Von Neumann Entropy of the ground state for (a) TFIM (Fig. \ref{Fig:Entropy_TFIM}) and (b) c-TFIM (Fig. \ref{Fig:Entropy_cTFIM}) models across the central cut (passing through the mid-point of the chain) for varying system sizes (N) is depicted. The ground state of the TFIM model obeys the Area-Law (and hence is constant in 1D) while that of c-TFIM follows the Volume-Law scaling of entanglement entropy and increases proportionally to the increasing size of the chain.  

\begin{figure}[!ht]
\centering \includegraphics[width=0.45\textwidth]{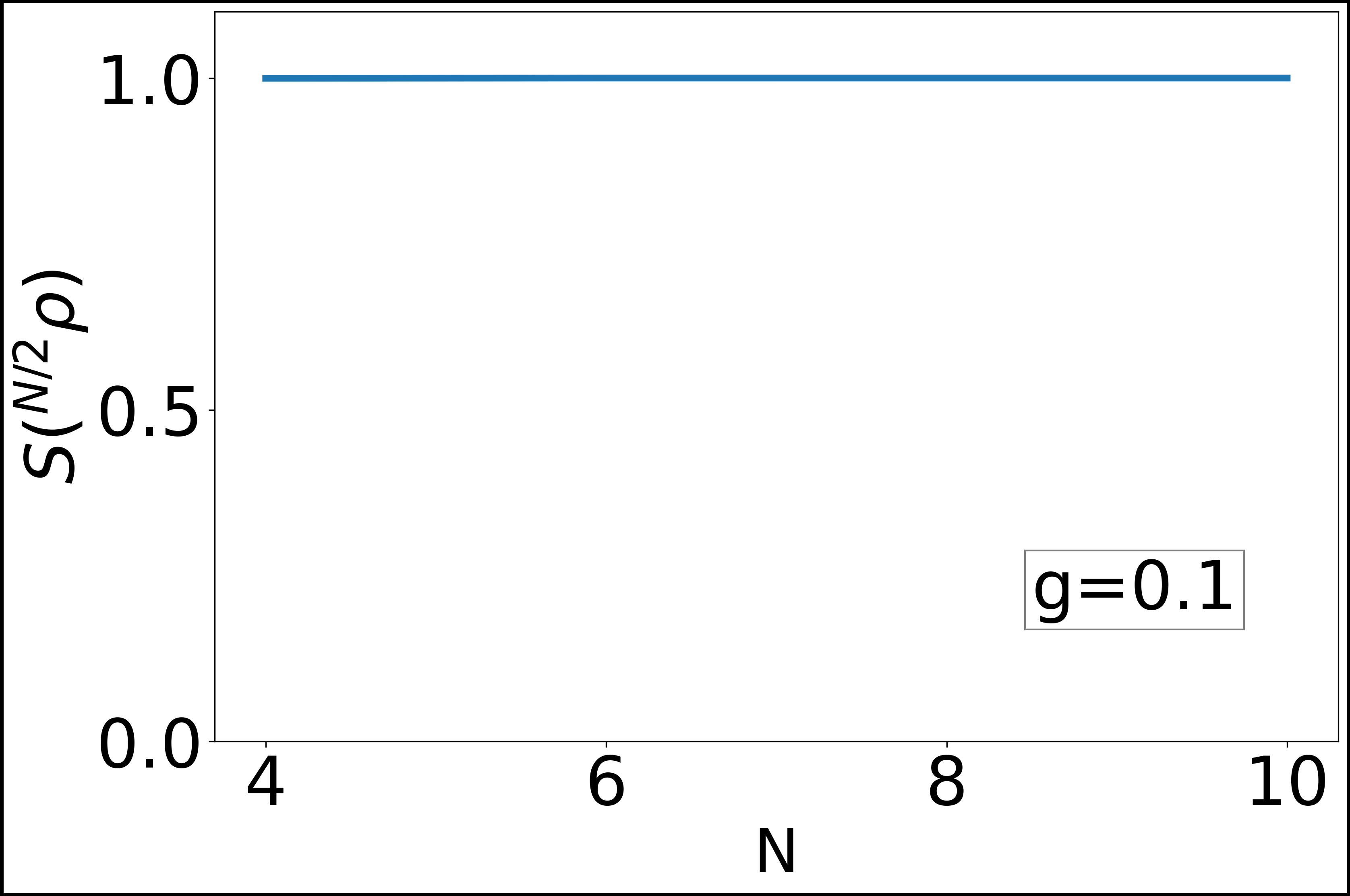}
\caption{\color{black} TFIM }
\label{Fig:Entropy_TFIM}

\centering \includegraphics[width=0.45\textwidth]{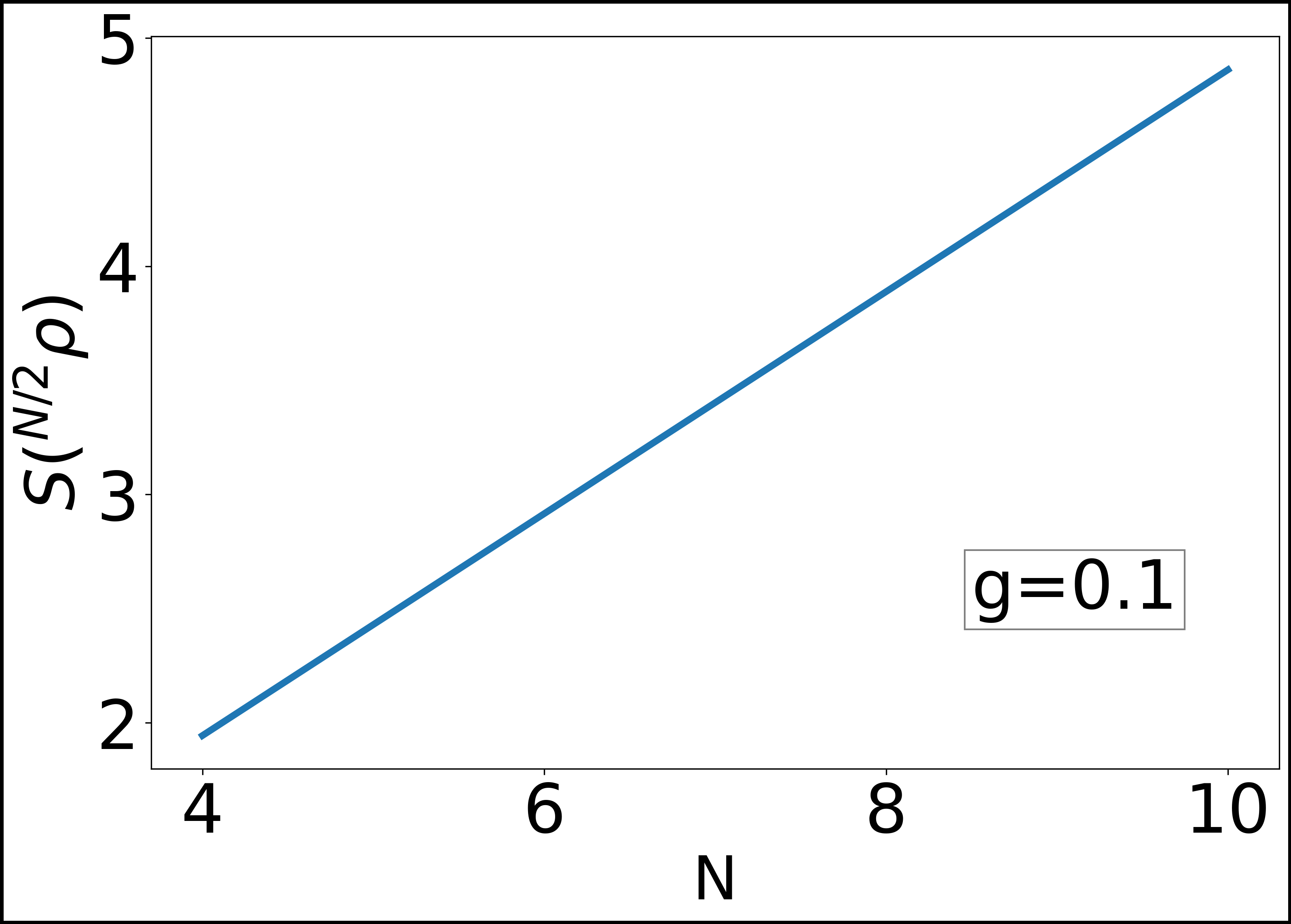}
\caption{\color{black} c-TFIM }
\label{Fig:Entropy_cTFIM}
\end{figure}

\newpage
\section{Fisher Information of TFIM and c-TFIM}
\label{Fischer_plots}
The variation of the largest eigenvalue of the Fisher Information Matrix with g for (a) TFIM and (b) c-TFIM models. The color gradation depicts the increasing value of g from red to blue. Each point on the Fig. \ref{Fig:Fisher_TFIM} and \ref{Fig:Fisher_cTFIM} is obtained by averaging over several best-converged points after training many randomly initialized networks.

\begin{figure}[!htb]
\centering \includegraphics[width=0.45\textwidth]{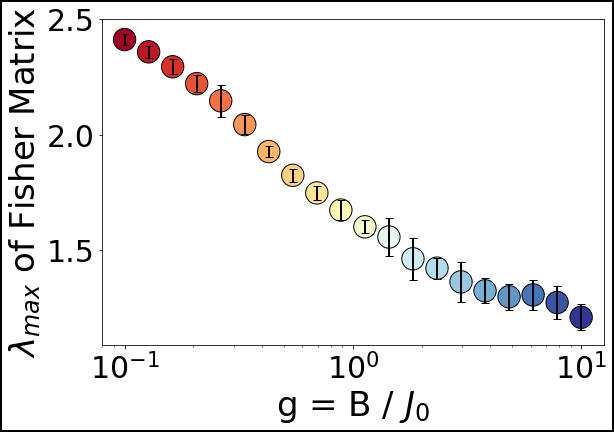}
\caption{\color{black} TFIM }
\label{Fig:Fisher_TFIM}

\centering \includegraphics[width=0.45\textwidth]{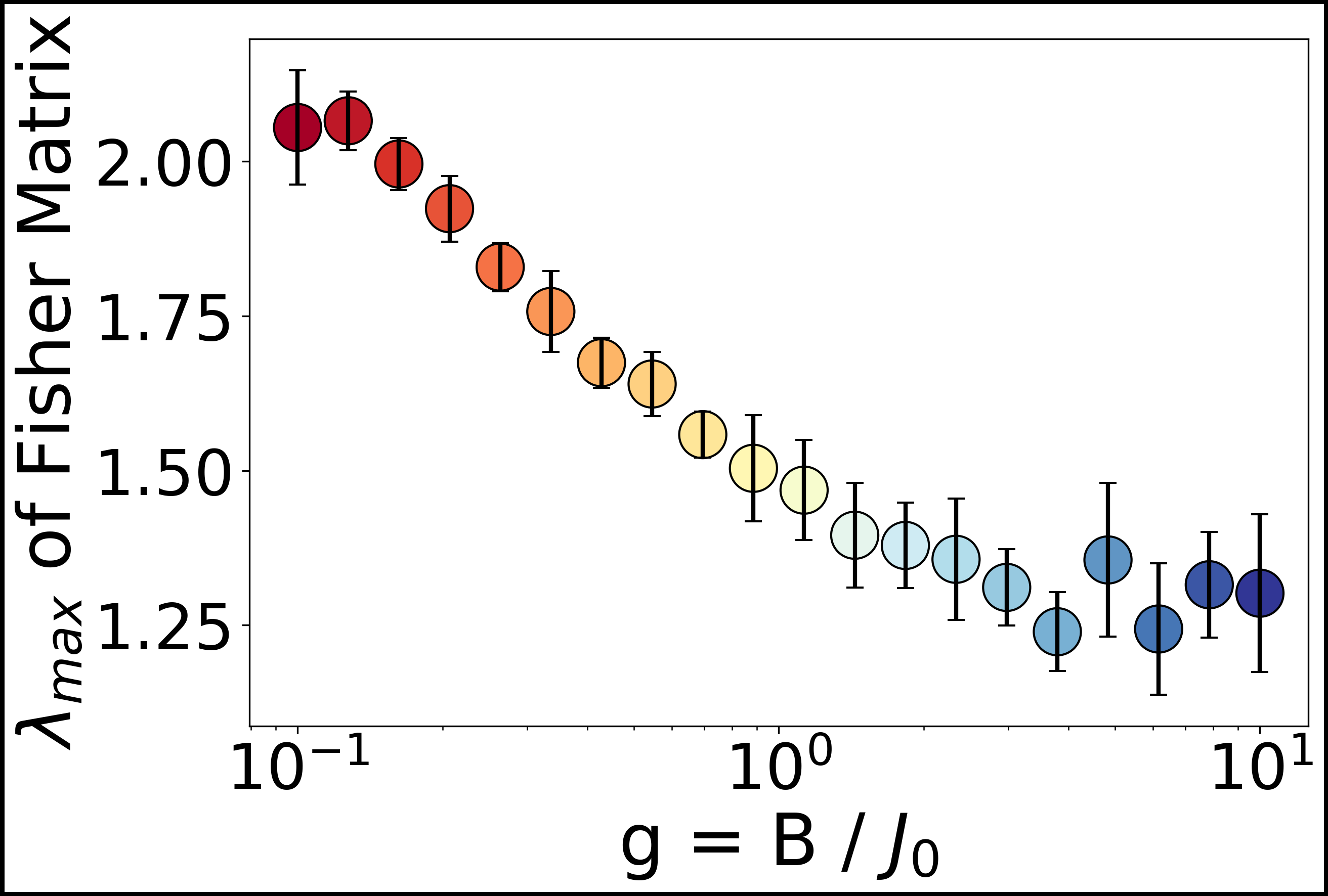}
\caption{\color{black} c-TFIM }
\label{Fig:Fisher_cTFIM}
\end{figure}

\end{document}